\documentclass[10pt]{elsarticle}

\usepackage[utf8]{inputenc}
\usepackage{amsmath,amsfonts,amssymb,amsthm,dsfont}
\usepackage[numbers]{natbib}
\usepackage[english]{babel}
\usepackage{hyperref,cleveref}
\usepackage[usenames, dvipsnames]{color}
\usepackage{multicol,cuted}
\usepackage{bbold,graphicx}
\usepackage{algorithm}
\usepackage{algorithmic}
\usepackage{enumerate}
\usepackage{epstopdf}
\usepackage{tikz}
\usetikzlibrary{automata,arrows,positioning,calc}

\def\bfo{\mathbf{1}}
\def\rset{\mathbb{R}}
\def\eg{\textit{e.g.\,}}
\def\ie{\textit{i.e.\,}}
\def\Scal{\mathcal{S}}
\def\Salg{\mathfrak{S}}
\def\nset{\mathbb{N}}
\def\tpi{\tilde{\pi}}
\def\esp{\mathbb{E}}
\def\rmd{\mathrm{d}}

\def\1{\mathds{1}}
\def\T{\mathrm{T}}
\def\vareps{\varepsilon}
\def\eps{\epsilon}
\def\var{\text{var}}
\def\Ltwo{\text{L}^2}
\newcommand{\pscal}[2]{\left\langle #1,#2\right\rangle}
\newtheorem{proposition}{Proposition}

\newtheorem{example}{Example}

\newtheorem{lemma}[proposition]{Lemma}
\newtheorem{assumption}{Assumption}
\newtheorem{remark}{Remark}
\newtheorem{fact}{Fact}

\begin{document}

\begin{frontmatter}

\title{On the convergence time of some non-reversible Markov chain Monte Carlo methods}
\author[ensae]{Marie Vialaret}
\author[ucd]{Florian Maire\corref{cor1}}
\ead{maire@dms.umontreal.ca}
\address[ensae]{ENSAE, Universit\'e Paris Saclay}
\address[ucd]{Universit\'e de Montr\'eal, d\'epartement de math\'ematiques et de statistique}
\cortext[cor1]{Corresponding author}

\begin{abstract}
It is commonly admitted that non-reversible Markov chain Monte Carlo (MCMC) algorithms usually yield more accurate MCMC estimators than their reversible counterparts. In this note, we show that in addition to their variance reduction effect, some non-reversible MCMC algorithms have also the undesirable property to slow down the convergence of the Markov chain. This point, which has been overlooked by the literature, has obvious practical implications. We illustrate this phenomenon for different non-reversible versions of the Metropolis-Hastings algorithm on several discrete state space examples and discuss ways to mitigate the risk of a small asymptotic variance/slow convergence scenario.
\end{abstract}


\end{frontmatter}

\section{Introduction}

Markov chain Monte Carlo (MCMC) methods enjoy a wide popularity in numerous fields of applied mathematics and are used for instance in statistics for parameter estimation or model validation. The purpose of MCMC is to approximate quantities of the form
\begin{equation}
\pi f := \int_\mathcal{S} f(x)\rmd\pi(x)\,,
\end{equation}
\ie the expectation of some $\pi$-measurable function $f$ with respect to a distribution $\pi$ defined on a state space $\mathcal{S}$, when an analytic expression of $\pi f$ is not available and direct simulation from $\pi$ is not doable. MCMC methods aim to simulating an ergodic Markov chain whose invariant distribution is $\pi$. As the chain converges towards its stationary distribution, it is possible to compute an empirical average of $f$, by using the sample path of the Markov chain.

\paragraph{Notations} In the following, $\pi$ will be referred to as the target distribution, $\Salg$ will denote a sigma-algebra on $\Scal$, $\Pr$ and $\esp$ will stand for the probability distribution and the expectation operator generated by the underlying random experiment, in absence of ambiguity. For a Markov chain $\{X_t,\,t\in\nset\}$  with transition kernel $P$ operating on $(\Scal,\Salg)$, we denote by $P^t$ the iterated kernel defined as $P^t(x,A):=\Pr(X_t\in A\,|\,X_0=x)$, for all $(x,A)\in\Scal\times\Salg$. For any measure $\mu$ on $(\Scal,\Salg)$, $\mu P$ defines the measure $\mu P:=\int_{\Scal} P(x,\,\cdot\,)\mu(\rmd x)$. We define by $\mathcal{M}_1(\Scal)$ the set of probability measures on $(\Scal,\Salg)$ and by $\Ltwo(\pi)$ the space of $\pi$-measurable function such that $\pi f^2<\infty$. For any signed measure $\mu$ on $(\Scal,\Salg)$, $\|\mu\|:=\sup_{A\in\Salg}|\mu(A)|$ denotes the total variation distance. The inner product on $\Ltwo(\pi)$ is denoted by $\pscal{\cdot}{\cdot}_\pi$. Finally a MCMC algorithm is identified with its Markov kernel $P$.

\paragraph{Efficiency of MCMC algorithms} Let us recall that the efficiency of a particular MCMC algorithm $P$ is traditionally assessed from two different points of view.
\begin{itemize}
\item[$\bullet$] \textbf{Convergence rate:} let $\mu$ be any initial distribution on $(\mathcal{S},\Salg)$. In the following, the convergence of $\mu P^t$ towards $\pi$ is measured with the total variation distance and quantified by a rate function $r(P,\,\cdot\,):\nset\to\rset^+$ satisfying
    \begin{equation}
    \label{eq:rate}
\lim_{t\to\infty}r(P,t)\|\mu P^t-\pi\|=0\,,\quad\text{for all}\;\mu\in \mathcal{M}_1(\Scal)\,.
\end{equation}
This notion is essential as it is related to the so-called \textit{burn-in} time $\tau\equiv\tau(\epsilon)$, i.e. the number of initial Markov chain states that are discarded so that the law of $X_{t}$ ($t\geq \tau$) is in an ball of radius $\epsilon$ centered on $\pi$. Few techniques allow to derive a theoretical expression for $\tau$ (see e.g. \cite{meyn1994computable,rosenthal1995minorization}) and in practice it is often estimated using convergence diagnostics (see \cite{plummer2006coda}).
\item[$\bullet$] \textbf{Asymptotic variance:} in stationary regime, the Markov chain should wander through the state space as efficiently as possible, so as to avail a MC estimator of $\pi f$ as accurate as possible. In particular, the variance of the empirical estimator $\widehat{{\pi}f}_n := \frac{1}{n}\sum_{t=1}^n f(X_t)$ should be as small as possible. This is quantified by the asymptotic variance of the Markov kernel $P$ for a function $f\in\Ltwo(\pi)$, which is defined, whenever it is finite, as
    \begin{equation}\label{eq:asy_var}
      v(f,P)=\lim_{n\to\infty}n\var\left\{\frac{1}{n}\sum_{t=1}^nf(X_t)\right\}\,,
    \end{equation}
    where the variance is w.r.t. $X_1\sim\pi$ and $X_{t+1}\sim P(X_t,\,\cdot\,)$ for $t\geq 1$.
\end{itemize}
Central to this work is the fact that those two measures of efficiency can sometimes be clashing, see e.g. \cite{rosenthal2003asymptotic}. In other words, it is possible to find two ergodic Markov chains $P_1$ and $P_2$ satisfying
\begin{equation}\label{eq:clash}
r(P_1,\,\cdot\,)\leq r(P_2,\,\cdot\,)  \qquad\text{and}\qquad
v(f,P_1)\leq v(f,P_2)\,.
\end{equation}
From a statistical viewpoint, a practitioner is likely to prefer MCMC estimators which offer narrow confidence intervals rather than those optimal for either above-mentioned markers of efficiency.  MCMC confidence intervals are typically related to the mean squared error (MSE) of the MCMC estimator. As a first approximation (for large $n$), we note that  for some function $f\in\Ltwo(\pi)$ and for any $\epsilon>0$, it can be readily checked that there exists $\tau$ such that the MSE is approximately equal to
\begin{equation*}
  \esp\left(\frac{1}{n}\sum_{t=1}^nf(X_{\tau+t}) -\pi f\right)^2 \approx \frac{1}{n}\left\{v(f,P)+\frac{\epsilon^2}{n}\left(\sum_{t=1}^n \frac{1}{r(P,\tau+t)}\right)^2\right\}\,.
\end{equation*}
which thus depends simultaneously on $v(f,P)$ and $r(P,\,\cdot\,)$. This analysis is carried out much more rigorously in \cite{latuszynski2013nonasymptotic}. In particular, Theorems 4.2 and 5.2 therein derive upper bounds of the MSE, for geometrically and polynomially ergodic Markov chains respectively, in function of constants related to $v(f,P)$ and $r(P,\cdot)$. Hence, for most statistical applications it is desirable to control jointly the asymptotic variance and the speed of convergence of the MCMC algorithm.

\paragraph{Context}
Recent contributions in Statistical Physics (see \eg  \cite{turitsyn2011irreversible} and \cite{vucelja2016lifting}) have rekindled interest in a specific family of MCMC algorithms relying on non-reversible Markov chains, see \eg \cite{bierkens2016non}, \cite{ma2016unifying}, \cite{bouncyParticleSampler2017} and \cite{andrieu2019peskun}, among others. This research is motivated by the fact that adding a divergence free  drift (with respect to $\pi$) to a Langevin diffusion process, whereby breaking its reversibility, speeds up the convergence to equilibrium \cite{lelievre2013optimal} and reduces the asymptotic variance of the estimator \cite{hwang2015variance}, see also \cite{hwang2005accelerating}. A natural question to ask is whether those results extend to discrete time settings and to possibly other types of non-reversible Markov chains. To the best of our knowledge, very few general results are available in the discrete-time setting apart from \cite{andrieu2019peskun} in which a novel framework to compare the asymptotic variances of several non-reversible MCMC algorithms is introduced.

\paragraph{Contribution}
In this work, we identify several situations where non-reversible Markov chains based on the Metropolis-Hastings algorithm reduce, as expected, the MCMC asymptotic variance but have also the adversarial effect to slow down, sometimes dramatically, the convergence of the Markov chain. We stress that this paper contains very few general theoretical statements but presents a collection of examples, in discrete state space, which illustrate this point. In some examples where the non-reversibility of the Markov kernel can be quantified by a positive scalar (in the spirit of \cite{lelievre2013optimal}), we find that the larger the non-reversibility, the slower the convergence. Such a conjunction can typically be observed if the vector field or the guiding direction imposed by the non-reversible perturbation is not adapted to the geometry of $\pi$, as already observed in \cite{diaconnis2000}. While it might be argued that our examples are simplistic and synthetic by nature, we believe that given the usual lack of knowledge on $\pi$ inherent to many practical applications, the risk of stumbling onto such slowly converging non-reversible MCMC algorithms is inevitable and should thus be taken into account in methodological developments. Indeed, from a statistical viewpoint, this note shows that when using non-reversible MCMC estimators, it is perhaps preferable to trade the optimal MCMC estimator for the asymptotic variance, for an estimator whose small but possibly sub-optimal asymptotic variance is not overshadowed by a large bias. Several ways to construct such non-reversible Markov chains are discussed.

\paragraph{Related work}
There are surprisingly very few works studying simultaneously the asymptotic variance and the convergence speed of non-reversible MCMC algorithms. This is perhaps due to the fact that for non-reversible Langevin the speed of convergence of the process in $\Ltwo(\pi)$ control both the bias and the asymptotic variance, see \cite{duncan2017using}. We nevertheless mention two recent contributions which motivate this research: in \cite{ma2016unifying}, the authors illustrate several experiments showcasing their non-reversible MCMC sampler. While the reduction in the Markov chain autocorrelation compared to the reversible alternative is striking, the speed of convergence to stationarity is, on a number of cases, similar or slightly slower for the non-reversible algorithm, see \cite[section 6]{ma2016unifying}. Finally, in \cite{andrieu2019peskun}, the authors highlight the fact that little is known on the speed of convergence of the  non-reversible Markov chains (Remark 2.11) and that novel methodological frameworks need to be developed.

\paragraph{Organization of the paper}
Section \ref{sec:rev_nonrev_MCMC} starts with a brief recap on reversible Markov chains and introduces the two families of non-reversible Markov chains that are considered in this paper: the lifted Markov chains and the marginal non-reversible Markov chains. Sections \ref{sec:lifted} and \ref{sec:NRMH} present situations where each type of non-reversible Markov chain exhibits slow convergence behaviour. In Section \ref{sec:two_vort}, a lifted version of a marginal non-reversible MH is presented which aims at solving, in some extent, the bias-variance tradeoff.

\section{Reversible and non-reversible Markov chains}
\label{sec:rev_nonrev_MCMC}
\paragraph{Reversible Markov chains}
The Metropolis-Hastings (MH) algorithm \cite{metropolis1953,hastings1970MH} (Algorithm \ref{algo_MH}) is arguably the most popular MCMC algorithm. The acceptance probability \eqref{eq:MH_ratio} guarantees that, by construction, MH generates a $\pi$-reversible Markov kernel which, therefore, admits $\pi$ as limiting distribution. Recall that a Markov kernel $P$ is said to be time reversible (or simply reversible) with respect to $\pi$ if $(\pi,P)$ satisfies
\begin{equation}\label{DBC2}
\forall(A,B)\,\in\Salg^{\otimes 2}\,,\qquad\int_A \pi(\rmd x)P(x,B)=\int_B\pi(\rmd x)P(x,A)\,.
\end{equation}
Reversible chains present numerous advantages, as several theoretical results (rate of convergence, spectral analysis, etc.) make their quantitative analysis relatively accessible. The main reason for their popularity is perhaps the property that a $\pi$-reversible Markov chain is necessarily $\pi$-invariant. Hence, constructing a Markov chain satisfying \eqref{DBC2} avoids further questions regarding the existence of a stationary distribution. Nevertheless, as Eq.\eqref{DBC2} imposes that the joint probabilities $\Pr(X_t\in A,X_{t+1}\in B)$ and $\Pr(X_t\in B,X_{t+1}\in A)$ are equal, reversibility may prevent the Markov chain from roaming efficiently through the state space, especially when $\pi$'s topology is irregular. This fact is illustrated by the following example.

\begin{algorithm}
\begin{algorithmic}
\caption{\label{algo_MH} Metropolis-Hastings algorithm}
\STATE Initialize in $X_0\sim \mu_0$ and let $X_t=x$
\STATE Propose $Y \sim Q(x,\,\cdot\,)\rightsquigarrow y$
\STATE Set $X_{t+1}=y$ with probability $A(x,y) = 1 \wedge R(x,y)$ where
\begin{equation} \label{eq:MH_ratio}
R(x,y) :=
\begin{cases} {\pi(y)Q(y,x)}\big\slash{\pi(x)Q(x,y)} & \mbox{ if } \pi(x)Q(x,y) \neq 0 \\ 1 & \mbox{ otherwise} \end{cases}
\end{equation}
\STATE If the proposal is rejected, set $X_{t+1}=x$
\end{algorithmic}
\end{algorithm}
\begin{example} \label{ex1}
Let $S$  be an integer such that $S\geq 4$ is even and $\rho\in(0,1]$. Define the discrete distribution on the circle $\Scal=\{1,2,\ldots,S\}$ ordered in the counterclockwise direction where $\pi_\rho(x)\propto 1$ if $x$ is odd and $\pi_\rho(x)\propto \rho$ if $x$ is even. This example is characteristic of probability distributions whose topology is rugged with valleys depth controlled by the parameter $1/\rho$. We consider the $\pi_\rho$-reversible MH Markov chain which attempts moving between neighbouring states, \ie for all $(x,y)\in\Scal^2\backslash\{(1,S),(S,1)\}$, we have $Q(x,y)=(1/2)\delta_{|x-y|=1}$ and $Q(1,S)=Q(S,1)=1/2$. When $\rho$ is small, the $\pi_\rho$-reversibility and the fact that two consecutive modes are separated by a state whose probability is in $\mathcal{O}(\rho)$ make the chain reluctant to move between them. In fact, the expected returning time to a given mode is of order $\mathcal{O}(1/\rho)$ implying that the Markov chain is mixing very slowly.
\end{example}

\noindent For reversible Markov chains, the convergence rate and the asymptotic variance are typically measured by two spectral quantities, the spectral gap and the spectral interval (as defined in \cite{rosenthal2003asymptotic}), the larger the better. In the context of Example \ref{ex1} with $S=4$, it can be readily checked that the spectrum of the Metropolis-Hastings transition kernel is $\{1,1-\rho,0,-\rho\}$ and thus the spectral gap and the spectral interval are both equal to $\rho$. Moreover, a careful derivation shows that the asymptotic variance is of order $\mathcal{O}(1/\rho)$, which illustrates the poor quality of the MH estimator on this example.

\paragraph{Non-reversible Markov chains}
As reversible chains have, by construction (see \eqref{DBC2}), the tendency to backtrack, it is desirable to transform their transition kernel to obtain chains whose dynamic departs from a random walk. Non-reversible Markov chains are thought to address this problem. The construction of non-reversible Markov chains can be traced back to  \cite{diaconis1991geometric,mira2000non,neal2004improving} for finite state space and \cite{horowitz1991generalized,gustafson1998guided} for general state space, but the analysis of these methods has been, until recently, essentially restricted to the finite case. Most of those methods consist in a subtle modification of standard reversible algorithms, designed so as to retain their $\pi$-invariance. In essence, the non-reversibility can be thought of as a dynamic giving the Markov chain some sort of inertia in one specific direction of the state space which thus attenuates the diffusive behaviour characteristic of reversible chains. In this paper, the non-reversible Markov chains are  categorized into two families:
\begin{itemize}
  \item \textbf{Marginal non-reversible chains}: these Markov chains operate on the marginal probability space $(\Scal,\Salg)$. They are obtained by introducing skew-symmetric perturbations, such as  cycles or vortices, in the transition kernel of a reversible Markov chain. This ensures that one specific direction is privileged by the Markov chain. In the case of MH algorithms, the probability of moving in the privileged direction can be increased in Eq. \eqref{eq:MH_ratio} by a quantity, say $\epsilon(x)$, that depends on the current state of the chain $x$, while the probability of the reverse move (in the opposite direction) is decreased by the same quantity. Algorithms proposed in \cite{bierkens2016non, chen2013accelerating, sun2010improving} follow this approach.
  \item \textbf{Lifted non-reversible chains}: even though precise definitions vary, this terminology which can be traced back to \cite{chen1999lifting}  often refers to Markov chains operating on an enlarged sampling space, typically $\mathcal{S}\times\Omega$. More precisely, the dynamic of the marginal sequence $\{X_t\in\Scal\,,\,t\in\nset\}$ is closely related to a privileged direction encoded in the sequence of auxiliary r.v. $\{\xi_t\in\Omega\,,\,t\in\nset\}$, often referred to as the momentum or spin variable. The two sequences are correlated: for example in  \cite{turitsyn2011irreversible,sakai2015nonrevRW, vucelja2016lifting,gustafson1998guided}, the momentum is preserved ($\xi_{t+1}=\xi_t$) as long as a proposal is accepted and is possibly switched  ($\xi_{t+1}=-\xi_t$) otherwise. Similarly, the generalized Metropolis-adjusted Langevin algorithm (GMALA) method \cite{ma2016unifying, poncet2017GMALA} uses several proposition kernels, according to the value of the auxiliary variable the chain is currently at. Markov chains based on Piecewise Deterministic Markov Processes (PDMP) such as the Zig-Zag algorithm \cite{bierkens2019zig} and the Bouncy Particle samplers \cite{bouncyParticleSampler2017,sherlock2017discrete} can also be considered as particular instances of this family.
\end{itemize}

\noindent While general results are scarce, it is commonly admitted that the asymptotic variance of MCMC algorithms using a non-reversible Markov chain is typically higher than those using reversible dynamic. We refer the reader to \cite{chen2013accelerating, neal2004improving,andrieu2019peskun} for some precise statements in certain specific contexts. Intuitively, the variance reduction feature can be explained by those guiding features which reduce, to some extent, the uncertainty on the Markov chain sample paths. However, apart from the general bounds on mixing time derived in \cite{chen1999lifting} and \cite{ramanan2018bounds}, little is known about the rate of convergence of those algorithms. We nevertheless note that more results exist for certain non-reversible Markov processes, see \eg \cite{andrieu2018hypercoercivity} and \cite{duncan2017using}.

\paragraph{Comparison of algorithms}
Since the message of this paper relies heavily on comparing Markov chains, we briefly explain how, in absence of analytical results, such comparisons can be carried out. The examples deal only with discrete probability distributions $\pi$ and thus comparing the convergence of algorithms can be quantitatively achieved by comparing the vectors $\mu P^t$ ($t\in\nset$) with $\pi$ in total variation distance.  In order to compare the asymptotic variance of two algorithms, we will use the representation of Theorem 4.8 of \cite{iosifescu2014finite} which states that for a discrete Markov kernel $P$ on $(\Scal,\Salg)$ and a function $f\in\Ltwo(\pi)$ with $v(f,P)<\infty$, we have
\begin{equation}
\label{eq:asy_var_form}
v(f,P)=2\pscal{\left[(I-P+\Pi)^{-1}-\Pi\right](f-\pi f)}{(f-\pi f)}_\pi-\|f-\pi f\|_\pi^2\,,
\end{equation}
where $\Pi$ is a matrix whose rows all equal $\pi$.

\section{Lifted non-reversible Metropolis-Hastings}
\label{sec:lifted}
When $\Scal\subseteq\rset$ or $\Scal\subseteq \mathbb{Z}$, the simplest form of non-reversible MCMC algorithm is perhaps the \textit{Guided Walk} (Algorithm \ref{algo_GW}), proposed by \cite{gustafson1998guided}, which belongs to the category of lifted Markov chains. It is essentially MH with an auxiliary variable that ``guides'' the walk: as long as the marginal chain moves, the direction of proposition is kept constant but it switches to the opposite direction as soon as a move is rejected. It can be checked that GW generates a Markov chain $\{(X_t,\xi_t),\,t\in\nset\}$ on $\bar{\Scal}:=\Scal\times\{-1,1\}$  which is $\bar\pi$-invariant, where $\bar\pi(x,\xi):=(1/2)\pi(x)$, but which is not $\bar\pi$-reversible, see \eg \cite{andrieu2019peskun}. Nevertheless, the sequence $\{X_t,\,t\in\nset\}$ is marginally $\pi$-invariant.

\begin{remark}
The marginal sequence of r.v. $\{X_t,\,t\in\nset\}$ produced by a lifted Markov chain (such as GW) is not itself a Markov chain and is therefore not characterized by any operator on $\Ltwo(\pi)$. Since reversibility qualifies the self-adjointness of an operator, the sequence $\{X_t,\,t\in\nset\}$ cannot be referred to as non-reversible, which is a common abuse of language.
\end{remark}

\begin{algorithm}
\begin{algorithmic}
\caption{\label{algo_GW} The Guided Walk algorithm}
\STATE Initialize in $X_0\sim \mu_0$ and $\xi_0\in\{-1,1\}$ and let $(X_t=x,\xi_t=\vareps)$
\STATE Propose $Y=x+\vareps|Z| \rightsquigarrow y$ where $Z\sim Q$
\STATE Set $X_{t+1}=y$ and $\xi_{t+1}=\vareps$ with probability $A(x,y) = 1 \wedge R(x,y)$ where
\begin{equation} \label{eq:MH_ratio}
R(x,y) :=
\begin{cases} {\pi(y)Q(y-x)}\big\slash{\pi(x)Q(x-y)} & \mbox{ if } \pi(x)Q(x-y) \neq 0 \\ 1 & \mbox{ otherwise} \end{cases}
\end{equation}
\STATE If the proposal is rejected, set $X_{t+1}=x$ and $\xi_{t+1}=-\vareps$
\end{algorithmic}
\end{algorithm}

\setcounter{example}{0}
\begin{example}[continued]
We apply the Guided Walk algorithm to Example \ref{ex1}. In this context, GW decreases, sometimes dramatically, the asymptotic variance of MC estimators obtained with MH. This is particularly striking given the fact that GW is merely an elementary modification of the MH algorithm which comes at no additional computational cost.
\end{example}

The GW asymptotic variance derivation is not straightforward and thus turn to Eq. \eqref{eq:asy_var_form} for numerical evaluation. The MH and GW asymptotic variances for $f=\text{Id}$ are illustrated in Figure \ref{fig:ex1} in function of the parameter $\rho\in(0,1]$. Rigorously, it should be noted that since the GW Markov chain operates on the state space $\bar{\Scal}$, we compare $v(f,P_{\text{MH}})$ and $v(\bar{f},P_{\text{GW}})$ where for all $(x,\xi)\in\bar{\Scal}$, $\bar{f}(x,\xi):=f(x)$ and where the inner product in $v(\bar{f},P_{\text{GW}})$ is implicity defined as $\pscal{\cdot}{\cdot}_{\bar \pi}$.

\begin{fact}
The non-reversible Guided Walk does not improve upon the $\mathcal{O}(1/\rho)$ MH inflation rate of the asymptotic variance as $\rho\downarrow 0$. However, the constants are significantly better with GW. In particular, asymptotically in the number of MCMC draws $n$ and when $\rho\downarrow 0$, MH needs twice as many samples to form an estimator with the same accuracy as the GW estimator and this comparison is even more dramatic when $\rho\uparrow 1$.
\end{fact}

\begin{figure}
\centering
\hspace*{-2cm}
\includegraphics[scale=0.6]{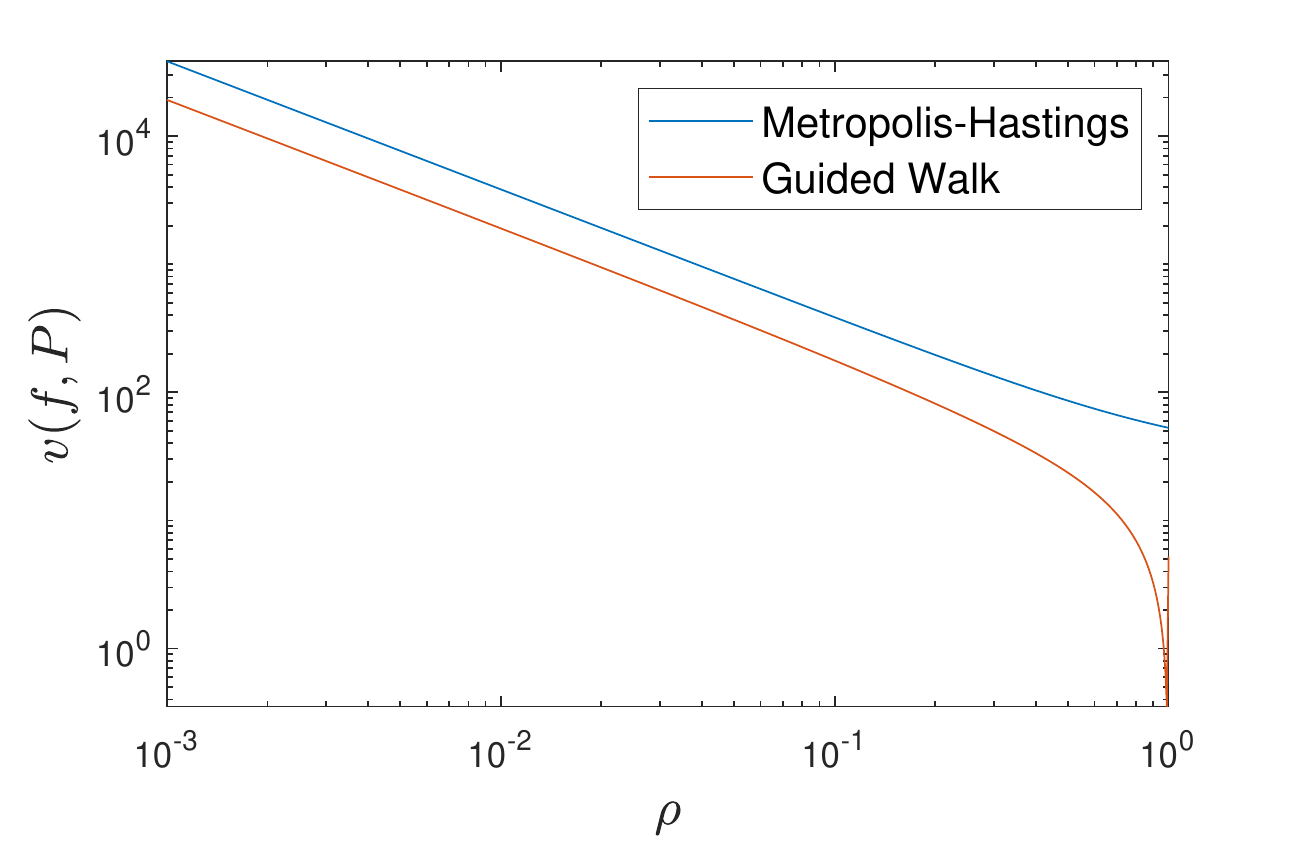}\includegraphics[scale=0.6]{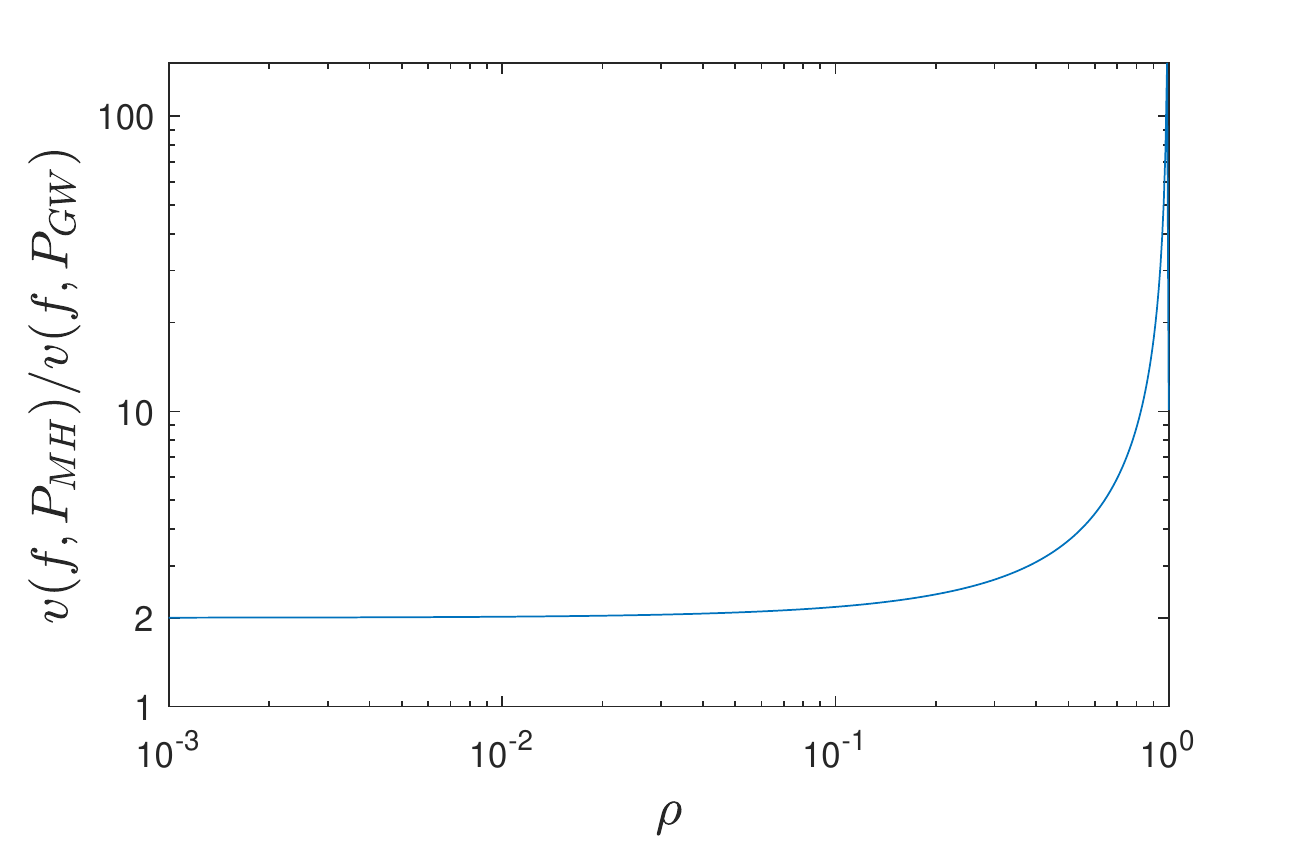}
\caption{(Example \ref{ex1}) Left: Asymptotic variances of the Metropolis-Hasting and the Guided Walk algorithms for the identity function and $S=10$. Note that for $\rho=1$, the asymptotic variance of the GW algorithm is 0 since the Markov chain is deterministic and that for a fixed $n>0$, there exists a $p\in\nset$ and $\tau<S$ such that $n=pS+\tau$ and $\sum_{t=1}^nf(X_k)=p S(S+1)/2+X_1+\ldots+X_\tau$. Combining this with the asymptotic variance definition yields, $v(f,P)=\lim_{n\to\infty} \var(X_1+\ldots+X_\tau)/n=0$. Right: ratio of the two asymptotic variances. \label{fig:ex1}}
\end{figure}

\noindent To put Example \ref{ex1} in the perspective of this note, we now turn to the convergence of the two algorithms.

\begin{proposition}
In the context of Example \ref{ex1}, the GW Markov kernel is $\bar\pi$-invariant but is not ergodic: for some initial measure $\mu$ on $(\bar\Scal,\bar\Salg)$, the TV distance $\|\mu P_{\text{GW}}^t-\bar\pi\|$ does not converge to zero as $t$ increases.
\end{proposition}
\begin{proof}
The GW transition mechanism is illustrated at Fig. \ref{fig:ex1_chain}: it can be readily checked that it is reducible. However, it is easy to show by induction that if $k$ is odd, for all $\xi\in\{-1,1\}$ and all $p\in\nset$ then $\Pr(X_{2p+1}=k,\xi_{2p+1}=\xi\,|\,X_0=k,\xi_{0}=\xi)=0$ and $\Pr(X_{2p}=k,\xi_{2p+1}=\xi\,|\,X_0=k,\xi_{0}=\xi)>0$. It is therefore 2-periodic and thus the GW Markov kernel is not ergodic. In particular, it does not converge to its stationary distribution for all initial measures.
\end{proof}

The periodicity of the Guided Walk in Example \ref{ex1} is caused by the fact that any state $(x,\xi)$ where $x=2p$ $(p\in\nset)$ is followed by a deterministic transition, which is a by-product of the non-reversibility of the GW. The GW Markov chain is, in a sense, ``too irreversible'' to be ergodic. To break the periodicity of the Guided Walk in Example \ref{ex1} and to obtain a non-reversible yet ergodic Markov chain, it is possible to ``reduce'' the amount of irreversibility of the initial GW by introducing a random switch of the momentum variable. This step is in line with the discussion on the need for refreshment in the Bouncy Particle Sampler, see \cite[Section 4.3]{bouncyParticleSampler2017}, see also the comments at the end of Section 4 of \cite{diaconnis2000}. For all $\alpha\in[0,1]$, consider the kernel
\begin{equation}\label{eq:GW_mom_swi}
\tilde{P}_{\mathrm{GW},\alpha}:=P_{\mathrm{GW}}(\alpha P_{\mathrm{flip}}+(1-\alpha)\text{Id})\,,
\end{equation}
where $P_{\mathrm{flip}}$ is the Markov transition kernel on $\bar{\Scal}\times\bar\Salg$ which freezes $X$ and draw $\xi$ afresh, with probability $1/2$ for both outcomes. In other words, with probability $\alpha$, the usual GW transition is immediately followed by a momentum switching operation. It is easy to check that $\tilde{P}_{\mathrm{GW},\alpha}$ $\pi$-invariant for all $\alpha\in[0,1]$ and that it is non-reversible if and only if $\alpha<1$ and aperiodic if and only if $\alpha>0$. For $\alpha=1$, it can be seen that the marginal chain $\{X_t,\,t\in\nset\}$ is Markov since independent of the past momentum and indeed coincides with MH. Figure \ref{fig:ex1_2} illustrates the behaviour of $\tilde{P}_{\mathrm{GW},\alpha}$ for three different refreshing rates $\alpha$.  The existence of a tradeoff between a low asymptotic variance (for the identity function) and a fast convergence is here obvious: among the tested parameters $\alpha$, for a given parameter $\rho$, say $\rho=0.1$, one would choose $\alpha=0.1$. Indeed, the convergence rate of $\tilde{P}_{\mathrm{GW},0.1}$ is more than two times faster than $P_{\text{MH}}$ (and more than five times faster than $\tilde{P}_{\mathrm{GW},0.01}$) while the optimal asymptotic variance $v(\text{Id},{\tilde{P}}_{\mathrm{GW},0.1})$  is hardly larger than $v(\text{Id},{P}_{\mathrm{GW}})$, at least for $\rho\ll 1$.


\begin{figure}
\centering
\hspace*{-1.25cm}
\includegraphics[scale=0.5]{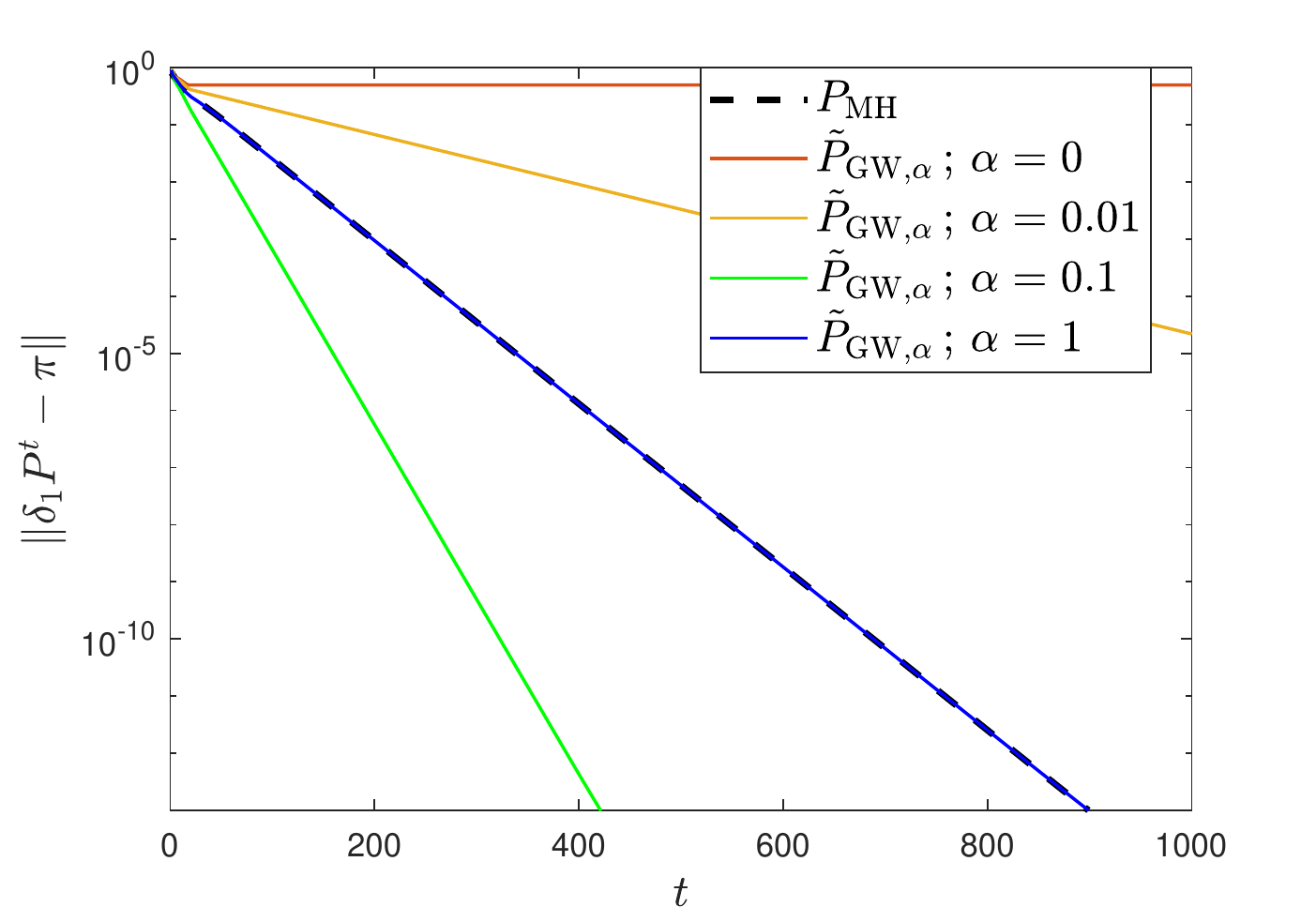}\includegraphics[scale=0.5]{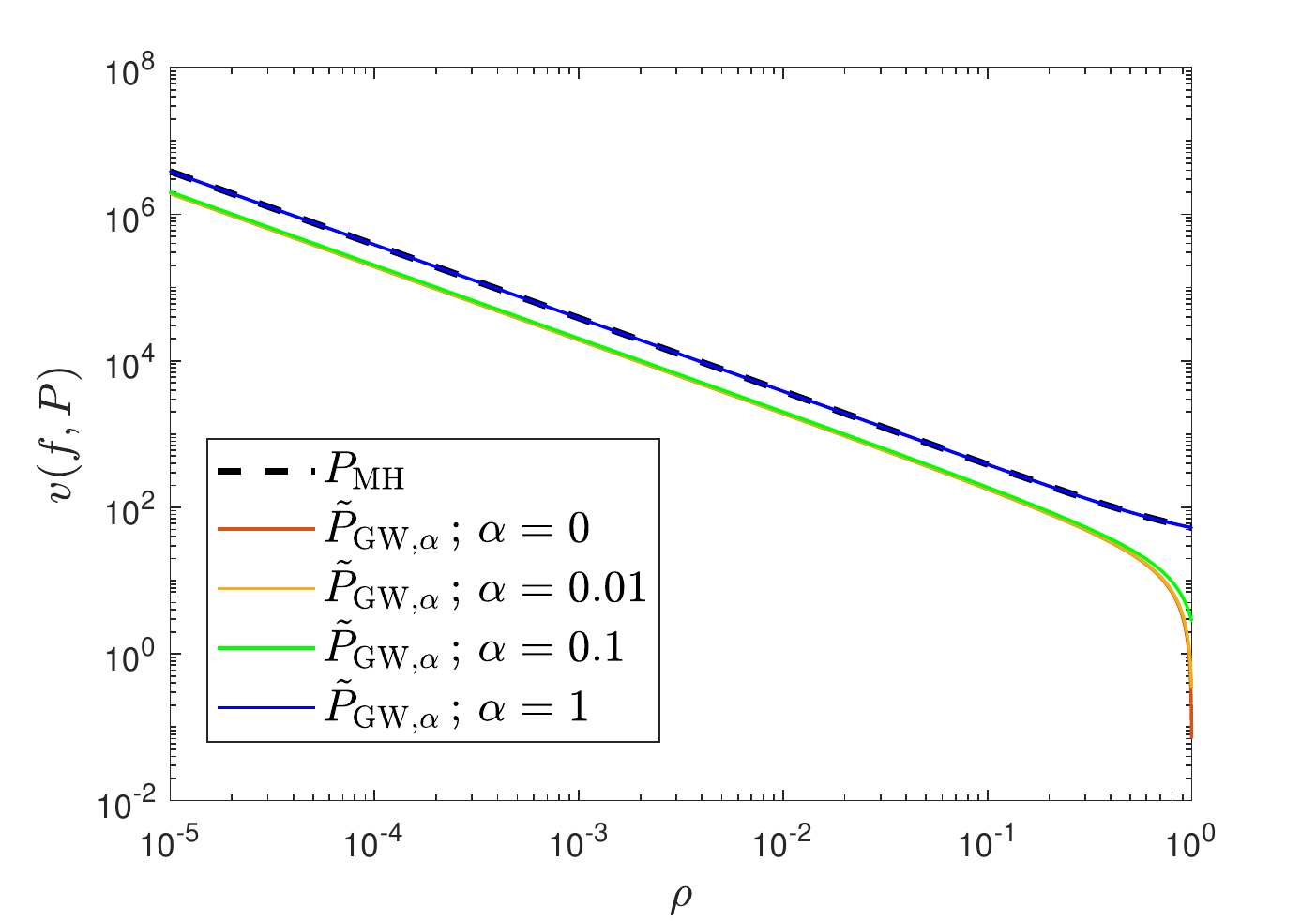}
\caption{(Example \ref{ex1}) left: convergence (in TV norm) from $\delta_1P^t$ to $\pi$ (with $\rho=0.1$ and $S=10$) for the MH, GW and three mixed strategies, where $\delta_1$ is the Dirac mass at $\{X_0=1\}$ for MH and at $\{X_0=1,\xi_0=1\}$ for the GW and its variants. Right: Asymptotic variance of the Metropolis-Hasting and the Guided Walk algorithms for the identity function, in function of $\rho$.\label{fig:ex1_2}}
\end{figure}

The following example is a slight modification of Example \ref{ex1} that allows to depart from the previous somewhat extreme case, where the plain GW (with $\alpha=0$) is not even ergodic.
\begin{example}
\label{ex2}
Let $\pi$ be the distribution defined on the circle $\{1,\ldots,S\}$, oriented counter-clockwise, as $\pi(k)\propto k$ for all $k\in\{1,\ldots,S\}$, with $S$ odd and $S\geq 5$. We compare the reversible (MH) and non-reversible (GW) Markov chains to sample from this distribution\footnote{Both Markov chains are represented in Figure \ref{fig:ex2_1}.}. Compared to Example \ref{ex1}, $\pi$ is an archetypal probability distribution whose topology is smooth and heavy-tailed and the focus of the analysis is on the two samplers performances in function of the space dimension $S$ rather than on the distribution ruggedness.
\end{example}

\begin{figure}
\centering
\hspace*{-0.75cm}
\includegraphics[scale=0.5]{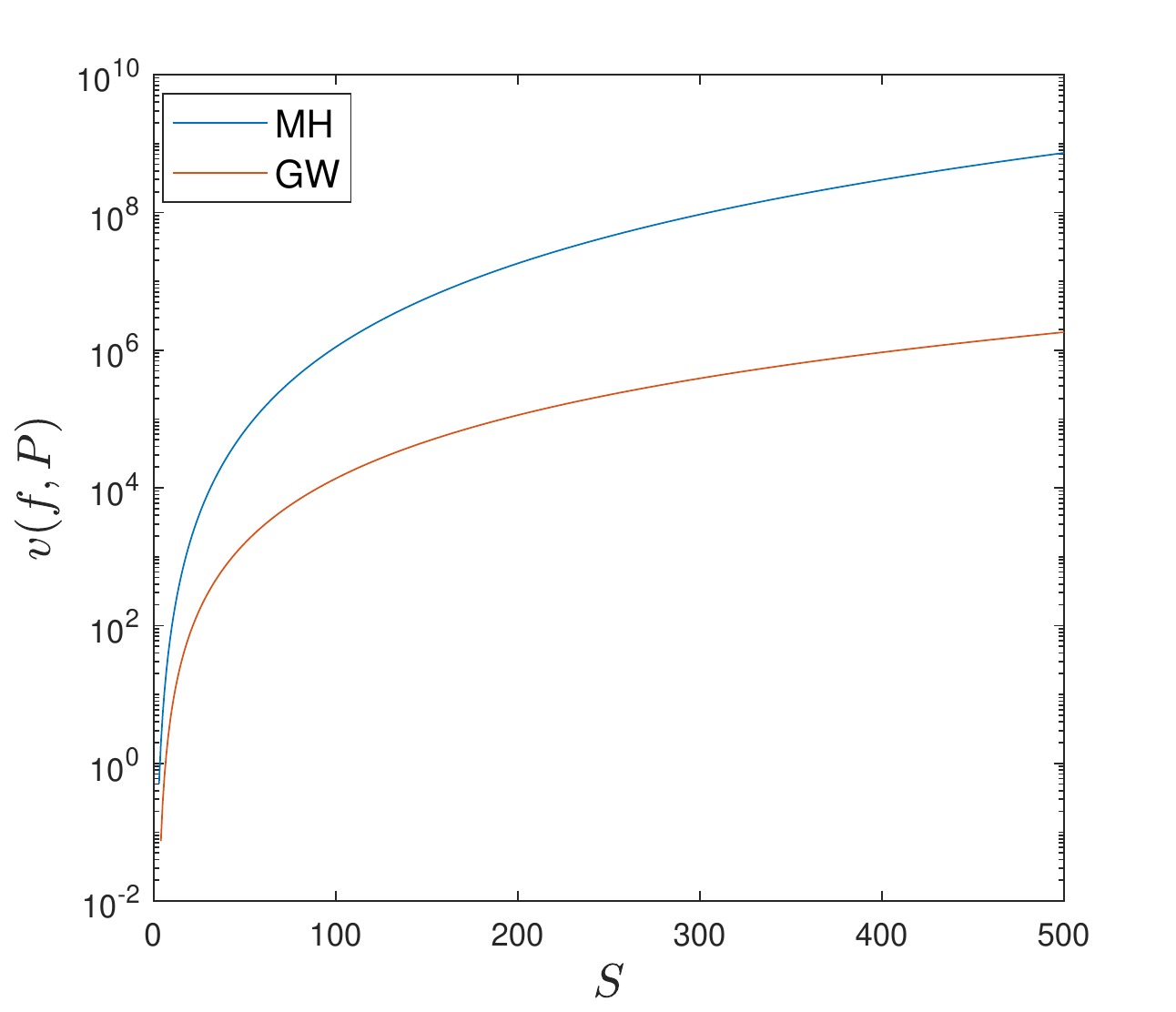}\includegraphics[scale=0.54]{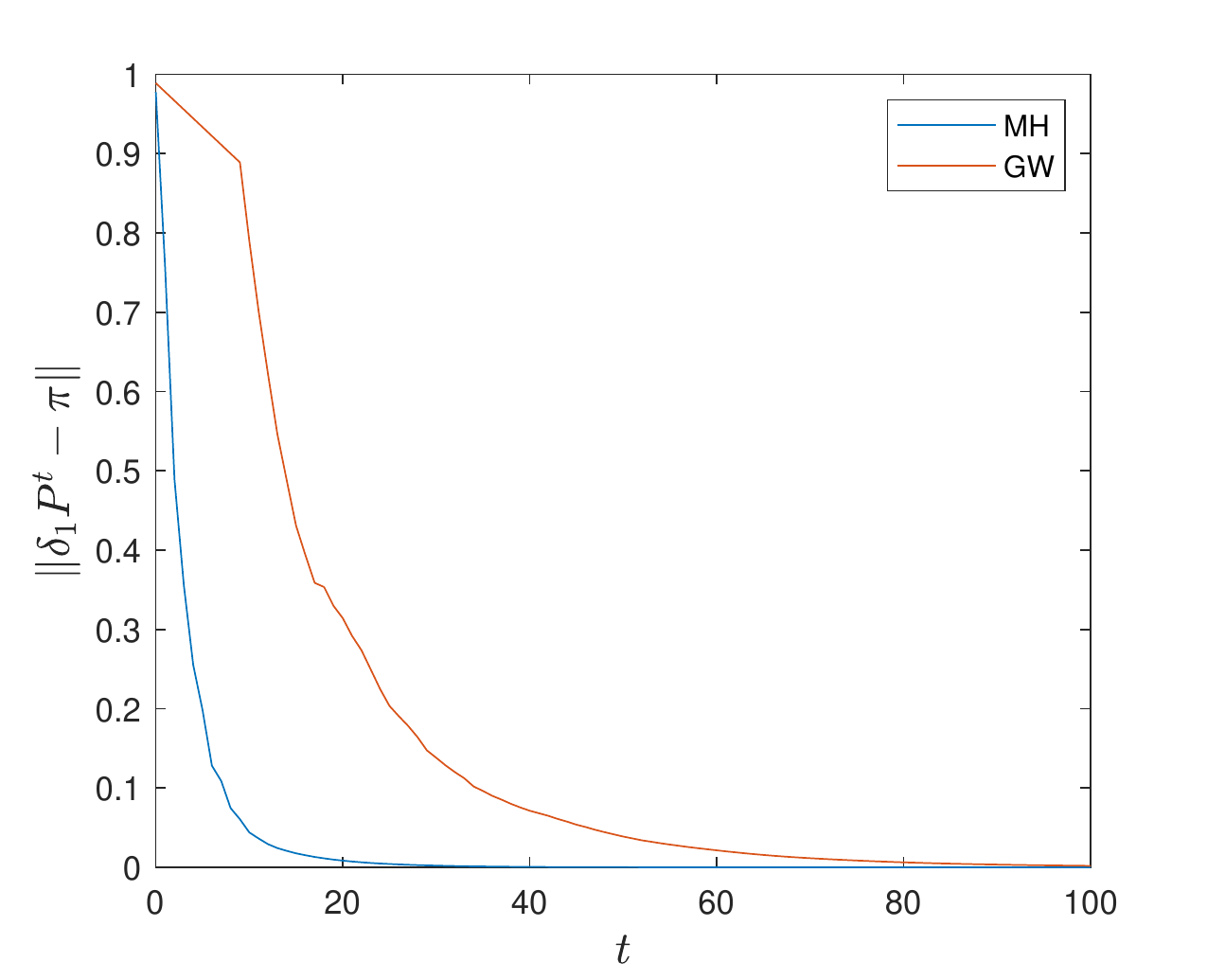}
\caption{(Example \ref{ex2}) Left: asymptotic variance of the Metropolis-Hasting and the Guided Walk algorithms for the identity function in function of $S$.  Right: convergence in TV norm from $\delta_1$ to $\pi$ for the two algorithms with $S=9$, where $\delta_1$ is the Dirac mass at $\{X_0=1\}$ for MH and $\{X_0=1,\xi_0=1\}$ for GW. \label{fig:ex2_2}}
\end{figure}

\begin{figure}
\centering
\hspace*{-3.5cm}
\includegraphics[scale=.8]{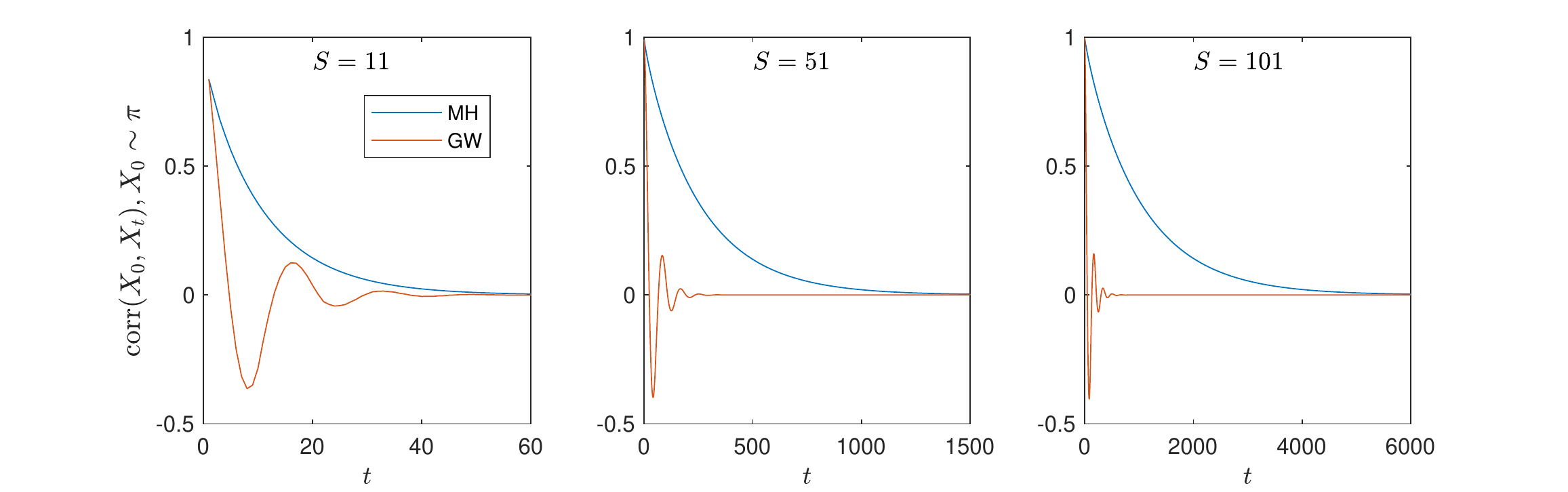}
\caption{(Example \ref{ex2}) Correlation pattern for MH and GW for $S\in\{11,51,101\}$. Since $\mu:=\esp(X)=(2S+1)/3$ and $v:=\var(X)=(S^2+S-2)/18$, we have for any function $f:\Scal\to\rset$ that $\text{corr}_{P_{\text{MH}}}(f(X_0),f(X_t))=(f^\T \Pi P_{\text{MH}}^t f-\mu^2)/v$, where $f=[f(1),f(2),\,\cdots\,,f(S)]$. To obtain $\text{corr}_{P_{\text{GW}}}(f(X_0),f(X_t))$, the derivation is similar mapping the product space $\bar\Scal$ with $\{1,\ldots,2S\}$ and replacing $f$ by $\bar{f}(k):=k\vee k-S$, $\Pi$ by $\bar\Pi:=\text{diag}(\bar\pi)$ and $P_{\text{MH}}$ by $P_{\text{GW}}$ in the previous expression.\label{fig:ex2_4}}
\end{figure}

\noindent Figure \ref{fig:ex2_2} summarizes the comparison of MH and GW in the context of Example \ref{ex2}. On the one hand, GW dominates MH  in the asymptotic variance sense (left panel) but on the other hand MH converges much faster to stationarity than GW (right panel). Again, as the analytical derivation for the GW asymptotic variance is not straightforward, we rely on the expression provided by Eq. \ref{eq:asy_var_form}. To further understand the difference in asymptotic efficiency between MH and GW, we recall that the series of a Markov chain autocorrelations is directly related to the asymptotic variance by
$$
v(f,P)=\var(f(X_0))\left(1+2\sum_{t>0}\mathrm{corr}_P(f(X_0),f(X_t))\right)\,,\quad X_0\sim \pi\,.
$$
Since $\pi(k+1)>\pi(k)$, if $(X_t=k,\xi_t=1)$ then for the $S-k$ next transitions the chain will visit deterministically all the states in increasing order until reaching $S$  at which point randomness resumes, \ie for all $k<S$ and all $t\geq 0$,
\begin{equation}\label{eq:cycles}
\Pr\left(X_{t+1}=k+1,X_{t+2}=k+2,\ldots,X_{t+S-k}=S\,|\,X_t=k,\xi_t=1\right)=1\,.
\end{equation}
The GW appealing variance reduction compared to MH is a direct consequence of those deterministic cycles resulting from the non-reversibility. Indeed, we note that the efficiency of the GW chain is due to the large-lag autocorrelation terms which are significantly smaller for GW than for MH, a fact which is thus more pronounced when $S$ increases, as illustrated by Figure \ref{fig:ex2_4}. By contrast, the first lag autocorrelation terms do not differ a lot, as quantified by  Proposition \ref{prop1}.
\begin{proposition}
\label{prop1}
Denoting $\esp_{\mathrm{MH}}$ and $\esp_{\mathrm{GW}}$ as the expectation operators generated by the MH and GW Markov chains respectively, we have in the context of Example \ref{ex2} that, at stationarity:
\begin{eqnarray*}
&&\esp_{\mathrm{MH}}(X_t X_{t+1})=\esp_{\mathrm{GW}}(X_t X_{t+1})=(1/2)S^2+o(S^2)\,, \\
&&\esp_{\mathrm{MH}}(X_t-X_{t-1})^2=\esp_{\mathrm{GW}}(X_t-X_{t-1})^2=3+o(1)\,,\\
&&\esp_{\mathrm{MH}}(X_t-X_{t-2})^2<5+o(1)\,,\quad \esp_{\mathrm{GW}}(X_t-X_{t-2})^2=15/2+o(1)\,.
\end{eqnarray*}
\end{proposition}
\begin{proof}
These results are obtained by direct calculation using the probabilities given in Figure \ref{fig:ex2_1}.
\end{proof}

\noindent We now turn to the convergence of the two Markov chains. The convergence of GW towards stationarity is penalized by the existence of those deterministic cycles, see Eq. \eqref{eq:cycles}, which is precisely where the GW variance reduction stems from. This is reflected in the convergence in TV norm which satisfies, for all $t<S$,
\begin{equation}
\label{eq:ex2_GW}
\|\delta_1 P_{\mathrm{GW}}^t(\cdot\,\times\,\{-1,1\})-\pi\|=1-\frac{1+t}{S(S+1)}\,,
\end{equation}
indicating that the convergence proceeds with an initial linear regime. This  observation can be related to the ``slow transient phase'' result obtained in the second part of \cite[Theorem 1]{diaconnis2000}.  By contrast, the convergence of the MH chain occurs at an exponential rate, as shown by Proposition \ref{prop_geo}.  It is possible to use minorization techniques or coupling constructions to find an upper bound of the MH convergence in TV norm. However, such bounds are typically too loose to be informative in the context of this example, especially since on the one hand Eq. \eqref{eq:ex2_GW} is an equality and on the other hand we are interested in the convergence nature of the two chains far from stationarity, \ie $t\approx S$. We instead turn to spectral techniques which eventually allows to provide an ordering on the convergence of MH and GW in the large $S$ regime in the L2 norm (see Proposition \ref{prop:comp_CV}).

\begin{proposition}
\label{prop_geo}
  In the context of Example \ref{ex2}, the MH Markov chain satisfies for all $t\in\nset$
  \begin{equation}\label{eq:ex2:MH}
    \|\delta_{1}P_{\text{MH}}^t-\pi\|_2\leq\left\{1-\frac{4}{S(S+1)}+\frac{2(2S+1)}{3S(S+1)}\right\}^{1/2}
    e^{-t/S(S+1)}\,.
  \end{equation}
\end{proposition}

\noindent We can now compare GW and MH in terms of convergence.

\begin{proposition}
\label{prop:comp_CV}
In the large $S$ regime, the L2 distance between the marginal in $X$ of the GW Markov chain after $S-1$ iterations to $\pi$ is larger than that between the MH Markov chain after $S-1$ iterations and $\pi$:
\begin{equation}\label{eq:prop:compL2}
\left\|\delta_1 P_{\text{MH}}^{(S-1)}-\pi\right\|_2^2\leq
\left\|\delta_1 P_{\text{GW}}^{(S-1)}(\,\cdot\,\times \{-1,1\})-\pi\right\|_2^2\,.
\end{equation}
\end{proposition}

\noindent Propositions \ref{prop_geo} and \ref{prop:comp_CV} offer rather conservative estimates for MH and, as such, our work only shows a marginal superiority of MH over GW. It would be useful to compare the speed of convergence of MH and GW closer to stationarity. Indeed, the MH L2 convergence estimate is expected to be much more accurate in this regime. However quantifying the GW L2 convergence beyond $t>S$ is harder, making comparison between the two methods more challenging. We leave this analysis for future work and for now, we report (top panel of Figure \ref{fig:ex2_3}) a comparison between the mixing time of the two algorithms calculated on a computer, for moderate size $S$.

\begin{fact}
As $S$ increases, GW becomes much slower than MH before returning to an initial level of inefficiency of about $5/2$ meaning that the non-reversible algorithm requires more than $5/2$ times iterations than MH to reach a similar neighborhood of $\pi$. It remains to be seen at what rate in $S$, if any, the two algorithms achieve a similar convergence speed or even if the GW becomes asymptotically in $S$ faster than MH. Such questions motivate a deeper analysis of the GW convergence.
\end{fact}

\begin{remark}
On a more practical side, we considered the GW $\alpha$-hybrid kernel $\tilde{P}_{\text{GW};\alpha}$  featuring the momentum refreshing operator, see Eq. \eqref{eq:GW_mom_swi}. 
We identified for several parameters $S$, the refreshing rate $\alpha^\ast\equiv\alpha^\ast(S)$ achieving the same asymptotic convergence rate between  $\tilde{P}_{\text{GW};\alpha^\ast}$ and $P_{\text{MH}}$ (in L2 norm). The red plot in the bottom panel of Figure \ref{fig:ex2_3} indicates how the asymptotic variance of the  $\alpha^\ast$-hybrid GW deteriorates that of GW with $S$. Interestingly, for moderately large $S$, the two algorithms achieve nearly the same asymptotic variance, meaning that there exists an algorithm which converges as fast as MH but which reduces the asymptotic variance of MH by a factor larger than $S/2$.
\end{remark}
\begin{figure}
\centering
\includegraphics[scale=0.61]{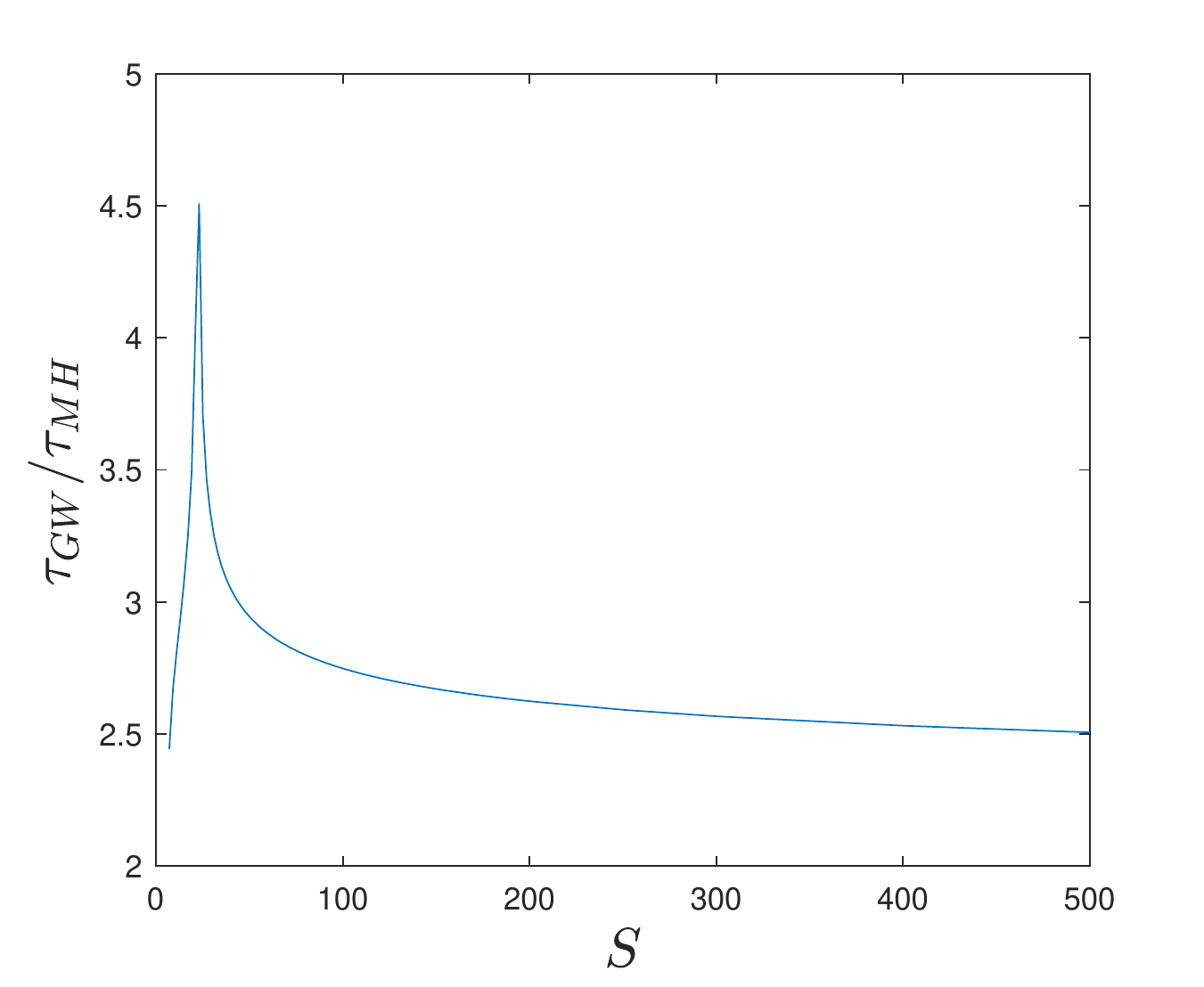}
\includegraphics[scale=0.6]{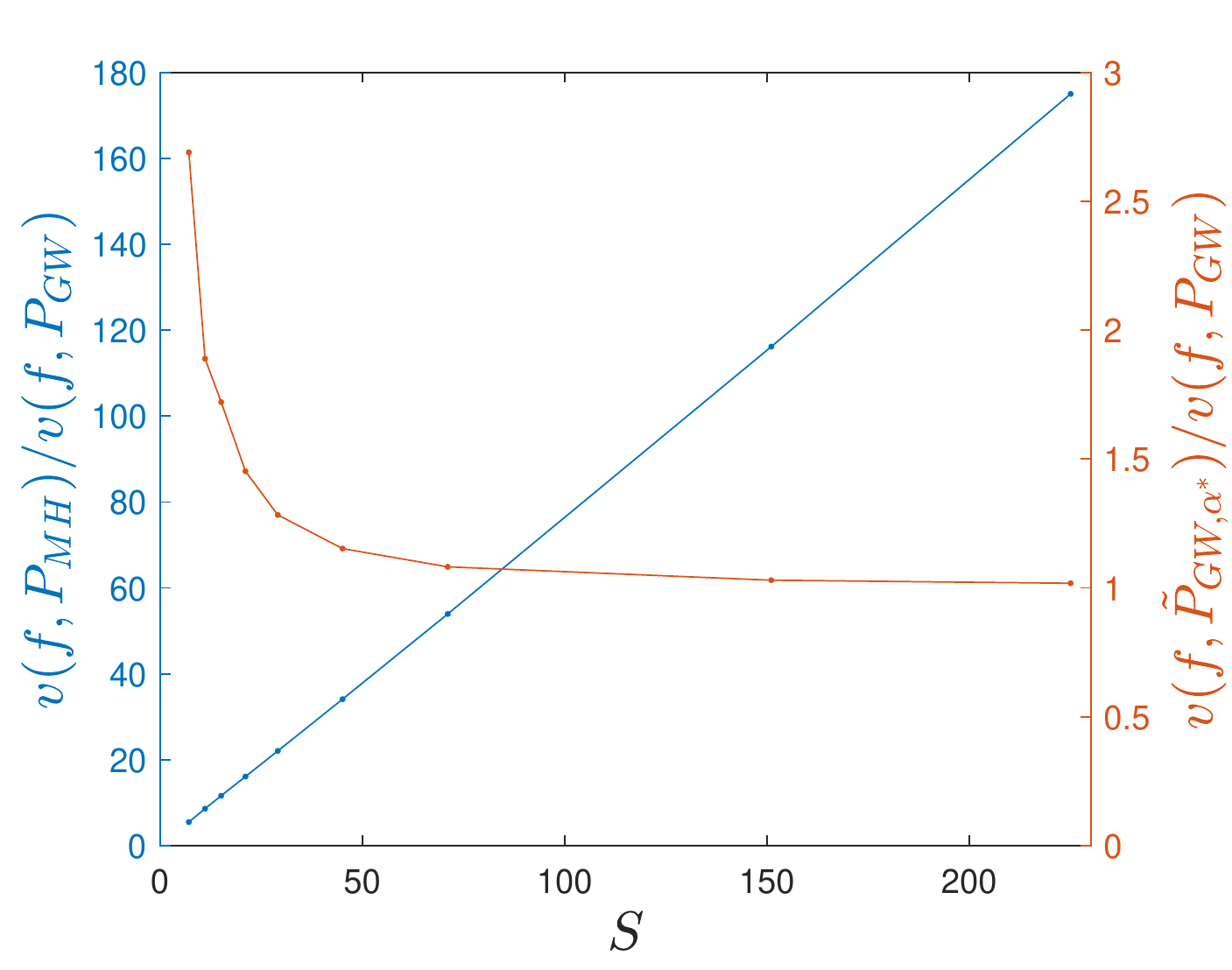}
\caption{(Exemple \ref{ex2}) Top: Comparison of the mixing time for GW and MH in function of $S$. Here the mixing time $\tau$ is defined as $\tau(P):=\inf\{t\in\nset\,:\,\|\delta_1 P^t-\pi\|\leq \eps\}$ with $\epsilon=10^{-5}$.  Bottom: in blue, ratio of the MH and GW asymptotic variances  and in red, same ratio for GW and its momentum switched versions $\tilde{P}_{\text{GW};\alpha^\ast}$ (Eq. \eqref{eq:GW_mom_swi}) where $\alpha^\ast:=\inf\{\alpha\in(0,1)\;:\;\tau(P_{\text{GW};\alpha^\ast})\leq\tau(P_{\text{MH}})\}$. \label{fig:ex2_3}}
\end{figure}

\begin{remark}
It was showed in \cite{andrieu2019peskun} (see Example 3.18) that a slight modification of GW can lead to an algorithm, referred to as \textit{Lifted GW}, with better asymptotic variance than GW: instead of switching the momentum when a proposal is rejected with probability one, this event could happen with a  well-chosen (state-dependent) probability, without affecting the stationary distribution of the chain. However, such algorithm is of little practical use since this event probability is typically impossible to calculate. This is not the case for Example \ref{ex2} which thus offer an illustration of this algorithm. Figure \ref{fig:liftedGW}  shows that the asymptotic variance reduction effect of the \textit{Lifted GW} compared to GW (proven in \cite{andrieu2019peskun}) comes along with a faster rate of convergence as well.
\end{remark}

\section{Marginal non-reversible Metropolis-Hastings}\label{sec:NRMH}

We now turn to marginal non-reversible Markov chains. For conciseness, we only study a specific instance of this family, namely the non-reversible Metropolis-Hastings (NRMH) algorithm recently proposed in \cite{bierkens2016non} and outlined at Algorithm \ref{algo_NRMH}. For notational simplicity, we only present the case where $\mathcal{S}$ is discrete but the ideas and results discussed hereafter have direct implications for general state space setups. NRMH modifies the original MH ratio by adding a skew-symmetric perturbation referred to as a vorticity matrix/field, $\Gamma:\Scal\times\Scal\to \rset$ in the MH ratio numerator. Conceptually, the vorticity field increases the  acceptance probability when moves are attempted in certain directions (\eg $x\to y$) and conversely decreases it for moves in opposite directions (\eg $y\to x$). Several assumptions on $\Gamma$ are considered in \cite{bierkens2016non}:
\begin{assumption}
\label{assumption1}
The vector field $\Gamma$ should satisfy a skew-symmetry condition
$$
\Gamma \neq 0\,,\qquad \forall\,(x,y)\in\Scal^2\,, \quad \Gamma(x,y) = - \Gamma(y,x)\,,
$$
and a non-explosion condition
$$
\forall\,x\in\Scal\,,\quad\sum_{y\in\Scal}\Gamma(x,y)=0\,.
$$
\end{assumption}
\noindent In addition, the MH proposal kernel $Q$ and the vorticity field $\Gamma$ are assumed to satisfy jointly the following condition:
\begin{assumption}
\label{assumption2}The proposal distribution satisfies a symmetric structure condition \ie
for all $(x,y)\in\mathcal{S}^2$, $Q(x,y) = 0 \Rightarrow Q(y,x) = 0$ and the non-negativity of the MH acceptance probability imposes a lower bound condition on $\Gamma$, \ie  for all $(x,y)\in\mathcal{S}^2$, $\Gamma(x,y) \geqslant -\pi(y)Q(y,x)$.
\end{assumption}

\begin{remark}
It can be noted that NRMH construction to ``dereversibilize'' MH takes the opposite route to the Guided Walk (Alg. \ref{algo_GW}).  While the former does not change the proposal $Q$ and modifies the MH acceptance ratio through $\Gamma$, the latter changes the proposal through the momentum variable $\xi_t$ and sticks to the canonical MH acceptance ratio.
\end{remark}

\makeatletter
\newcommand{\plusline}{%
  \let\old@ALC@lno=\ALC@lno%
  \renewcommand{\ALC@lno} {%
    \global\let\ALC@lno=\old@ALC@lno}%
}
\makeatother

\begin{algorithm}
\begin{algorithmic}[1]
\caption{\label{algo_NRMH} Non-reversible Metropolis-Hastings algorithm (NRMH).}
\STATE Initialize in $X_0\sim \mu_0$

Transition $X_t=x\to X_{t+1}$:
\STATE Propose $Y \sim Q(x,\,\cdot\,)\rightsquigarrow y$
\STATE Set $X_{t+1}=y$ with probability $A_\Gamma(x,y) = 1 \wedge R_\Gamma(x,y)$ where
\begin{equation} \label{eq:NRMH_ratio_0}
R_\Gamma(x,y) := \begin{cases} \frac{\Gamma(x,y) + \pi(y)Q(y,x)}{\pi(x)Q(x,y)} & \mbox{ if } \pi(x)Q(x,y) \neq 0 \\ 1 & \mbox{ otherwise} \end{cases}
\end{equation}
\STATE If the proposal is rejected, set $X_{t+1}=x$
\end{algorithmic}
\end{algorithm}

\noindent If $\Gamma$ and $Q$ satisfy Assumptions \ref{assumption1}--\ref{assumption2}, the NRMH Markov chain admits $\pi$ as invariant distribution (see \cite[Theorem 2.5]{bierkens2016non}) and is non-reversible. The intuition behind the non-explosion condition is that the non-reversibility introduced in the algorithm must compensate overall through the state space.  As noted in \cite{bierkens2016non}, the vorticity field quantifies a measure of ``non-reversibility'' of NRMH since
  $$
  \pi(x)P_{\text{NRMH}}^{(\Gamma)}(x,y)-\pi(y)P_{\text{NRMH}}^{(\Gamma)}(y,x)=\Gamma(x,y)\,,
  $$
  provided that $\Gamma$ satisfies Assumptions \ref{assumption1}--\ref{assumption2}. When $\Gamma$ can be parameterized by some scalar $\zeta$, \ie  $\Gamma\equiv \Gamma_\zeta$, we will use the shorthand notation $P_{\zeta}\equiv P_{\text{NRMH}}^{(\Gamma_\zeta)}$.

\setcounter{example}{0}
\begin{example}[continued]
\label{ex1_tbc}
We implement NRMH to infer the distribution defined at Example \ref{ex1}, with $\rho=0.1$. The following vector flow is considered: for all $(x,y)\in\Scal^2$,
\begin{equation}
\Gamma_\zeta(x,y):=
\left\{
\begin{array}{ll}
\zeta & \text{if} \; y=x+1\; \text{or}\; (x,y)=(S,1)\,,\\
-\zeta & \text{if} \; y=x-1\; \text{or}\; (x,y)=(1,S)\,,\\
0 &\text{otherwise}\,.
\end{array}
\right.
\label{eq:Gamma_mat}
\end{equation}
where $0<\zeta\leq\zeta_{\text{max}}:= \rho(S(1+\rho))^{-1}$. This condition on $\zeta$ and the structure of $\Gamma_\zeta$ ensures that Assumptions \ref{assumption1} and \ref{assumption2} are both satisfied. Figure \ref{fig:1} gives an illustration of the efficiency of MH and NRMH. In particular, it shows that as expected, NRMH allows to reduce significantly the variance of the Monte Carlo estimate: the asymptotic variance of NRMH for the test function $f:x\mapsto \1_{x=1}$ was in this case nearly $10$ times less than MH.  This is confirmed theoretically by Figure \ref{fig:NRMH:2}. The MH and NRMH estimators have a remarkably different behaviour asymptotically in $S$ ($\rho$ being fixed), there exists a constant $\beta>0$ such that for all $x\in\Scal$,
$$
v(\1_x,P_{\mathrm{MH}})=\mathcal{O}(1)\quad\text{and}\quad v\left(\1_x,P_{\zeta_{\max}}\right)=\mathcal{O}(1/S^\beta)
$$
and that for any polynomial function of order $n$, say $p_n$, we have
$$
v(p_n,P_{\mathrm{MH}})=\mathcal{O}(S^{2(n+1)})\quad\text{and}\quad v\left(p_n,P_{\zeta_{\max}}\right)=\mathcal{O}(S^{2n})\,.
$$
However, Figure \ref{fig:1} also shows that, similarly to the GW, the non-reversibility slows down the convergence of the Markov chain. Indeed, by construction, for any odd state $x$, a NRMH transition  satisfies $P_{\zeta_{\text{max}}}(x,y)=0$ for all $y<x$. Hence, starting with a measure $\mu_0=\delta_1$, the larger states will be explored at a much slower rate than with the reversible MH since the NRMH Markov chain must first visit all the intermediate states in increasing order. It is illustrated quantitatively at Figure \ref{fig:NRMH:3} (left panel) which shows that when $S$ increases the NRMH Markov chain with $\zeta=\zeta_{\text{max}}$ converges slower relatively to MH. The right panel of Figure \ref{fig:NRMH:3} indicates that when $\zeta$ decreases, the NRMH Markov chain convergence is similar to MH.

\begin{figure}
\centering
\includegraphics[scale=.7]{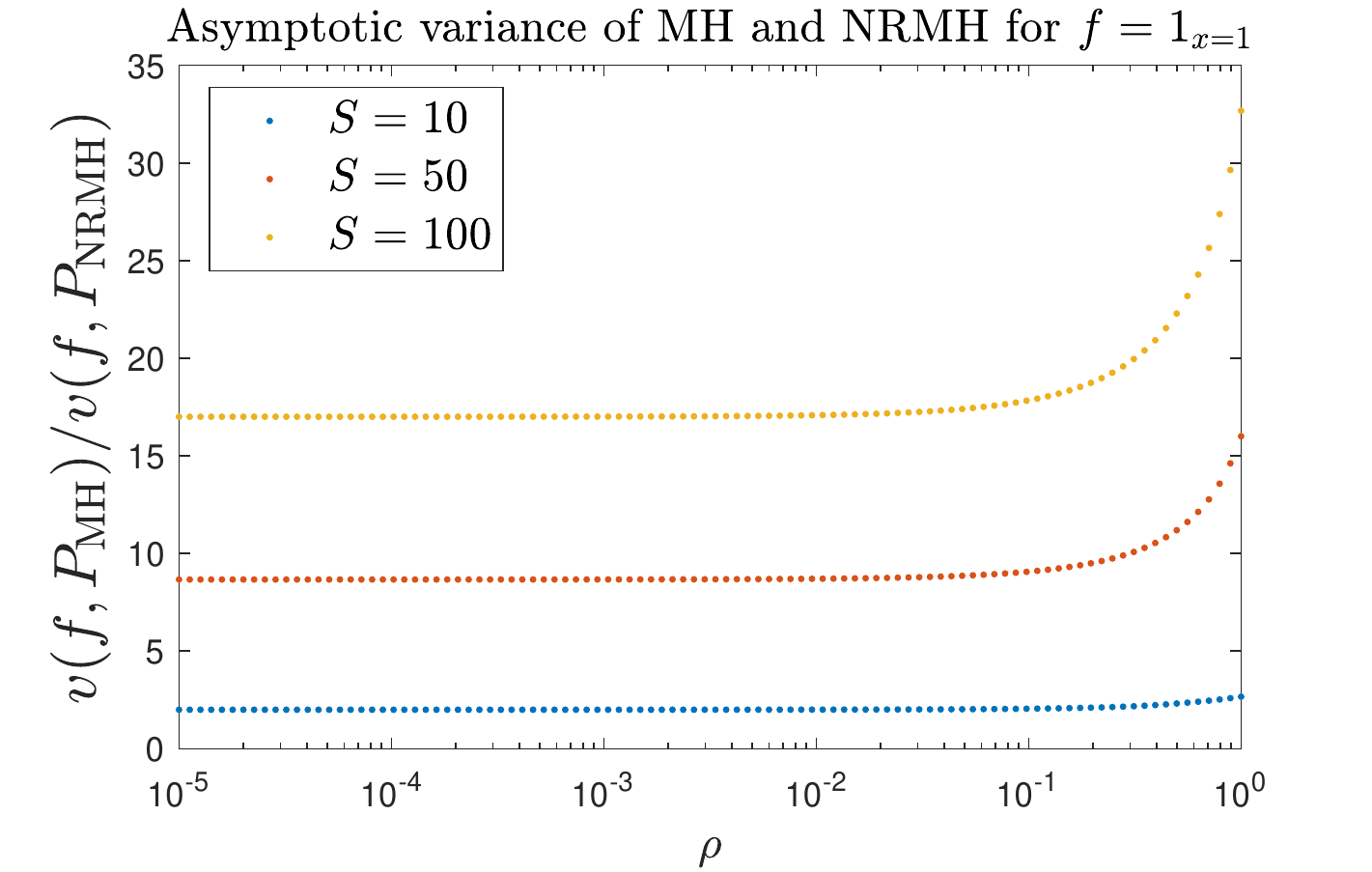}
\caption{(Example \ref{ex1}) Comparison of MH and NRMH asymptotic variance for the functional $f=\1_{x=1}$for different parameters $\rho$ and $S\in\{10,50,100\}$. Asymptotic variances were calculated with Eq. \eqref{eq:asy_var_form}. \label{fig:NRMH:2}}
\end{figure}

\begin{figure}
\centering
\hspace*{-1.1cm}\includegraphics[scale=.6]{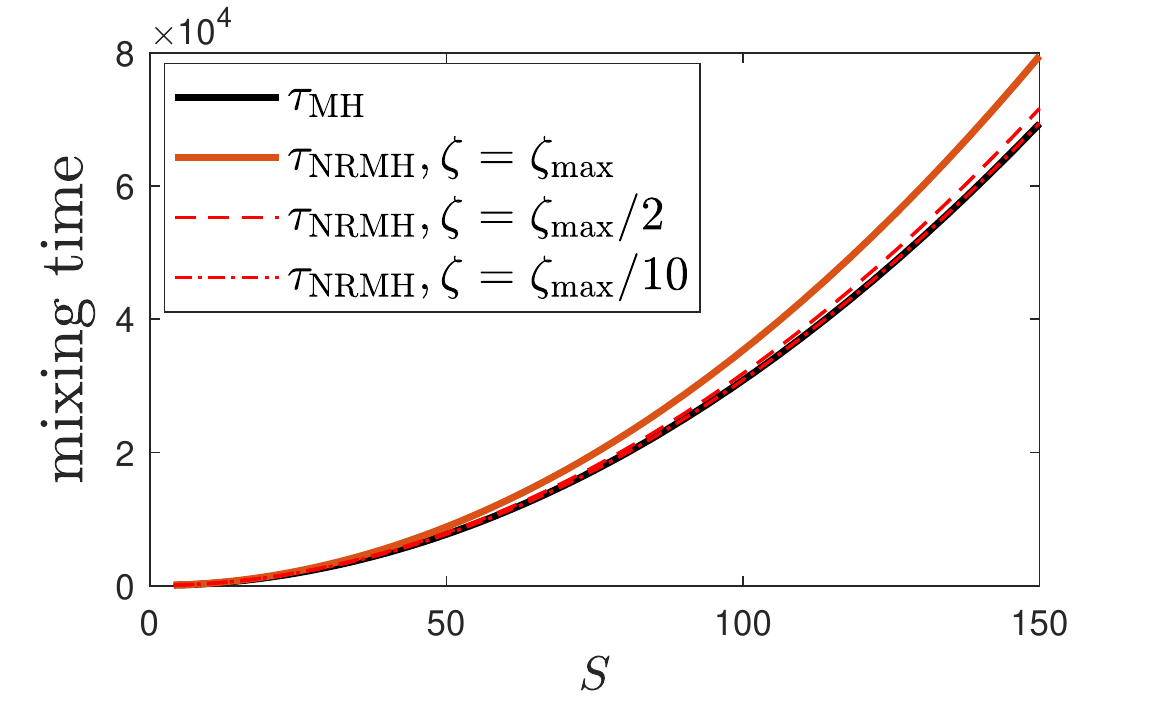}\includegraphics[scale=.6]{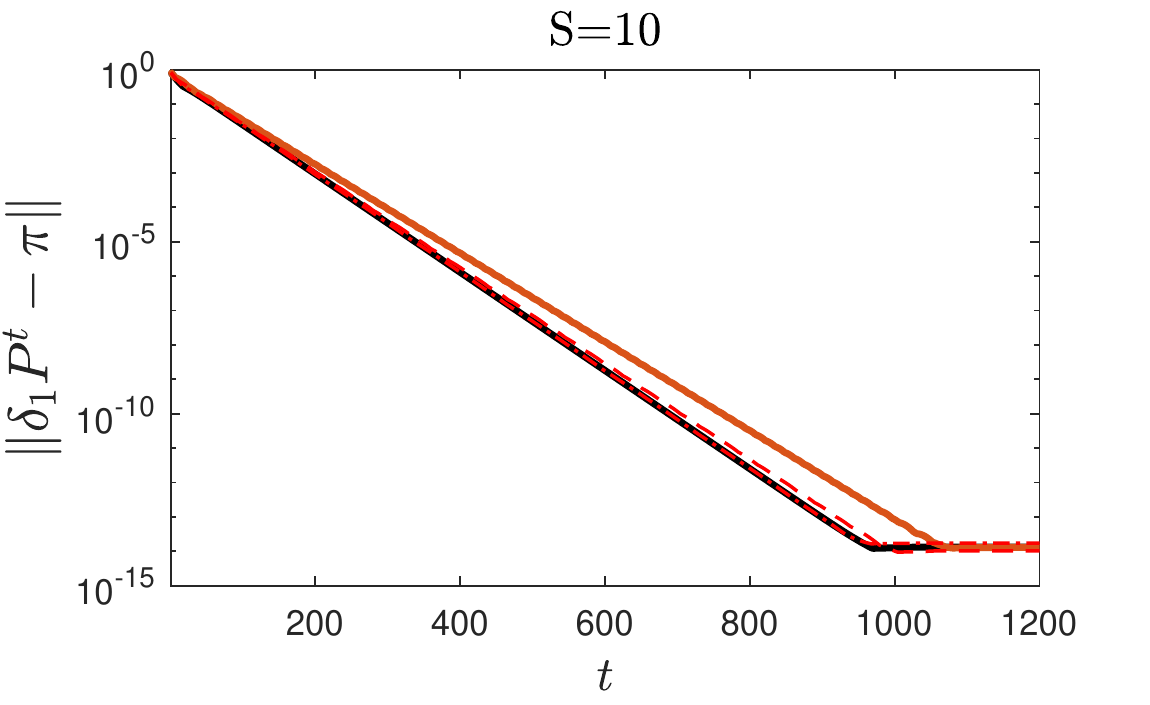}

\caption{(Example \ref{ex1}) Comparison of MH and NRMH convergence from $\delta_1$ to $\pi$ in total variation. The left panel gives, for each algorithm, the mixing time defined here as $\tau(P):=\inf\{t\in\nset\,:\,\|\delta_1P^t-\pi\|\leq 1/10^5\}$. The right panel gives, for the same algorithm, the convergence for $S=10$.  \label{fig:NRMH:3}}
\end{figure}

\end{example}

\noindent The following example shows that even when $\pi$ is as smooth as it can possibly get, it is possible to find situations where one cannot obtain simultaneously a rapidly converging Markov chain and a low MCMC asymptotic variance with NRMH.

\setcounter{example}{2}
\begin{example}\label{ex_cercle}
Consider the state-space $\Scal$ of Example \ref{ex1_tbc}, where $\pi$ is now the uniform distribution on $\Scal$. The proposal distribution is defined for some $\epsilon>0$ \footnote{Note that if $\epsilon=0$, the MH transition kernel does not satisfy $\lim_{t\rightarrow\infty}\sup_{x\in\Scal}\|\delta_x P_{\mathrm{MH}}^t-\pi\|=0$ as for all $i\in\Scal$, $P_{\mathrm{MH}}^t(i,i)=0\Leftrightarrow t$ is odd, meaning that (being irreducible) $P_{\mathrm{MH}}$ is 2-periodic.}
 and $(x,y)\in\Scal$ as
\begin{equation}
Q_\epsilon(x,y)=
\left\{
\begin{array}{ll}
\epsilon &\text{if}\;x=y\,,\\
(1-\epsilon)/2 & \text{if}\; |x-y|=1\,,\\
(1-\epsilon)/2 & \text{if}\; (x,y)=(1,S)\,\text{or}\,(x,y)=(S,1)\,.
\end{array}
\right.\end{equation}
The vorticity matrix $\Gamma_\zeta$ is defined as in Eq. \eqref{eq:Gamma_mat} for some $\zeta\in[0,\zeta_{\max}]$. Setting $\zeta_{\max}=(1-\epsilon)/2S$ and accepting/rejecting a proposed move with the probability given at Eq. \eqref{eq:NRMH_ratio_0} are sufficient to define a $\pi$-invariant and non-reversible NRMH Markov chain.
\end{example}

\begin{figure}
\centering
\hspace*{-.5cm}
\includegraphics[scale=0.5]{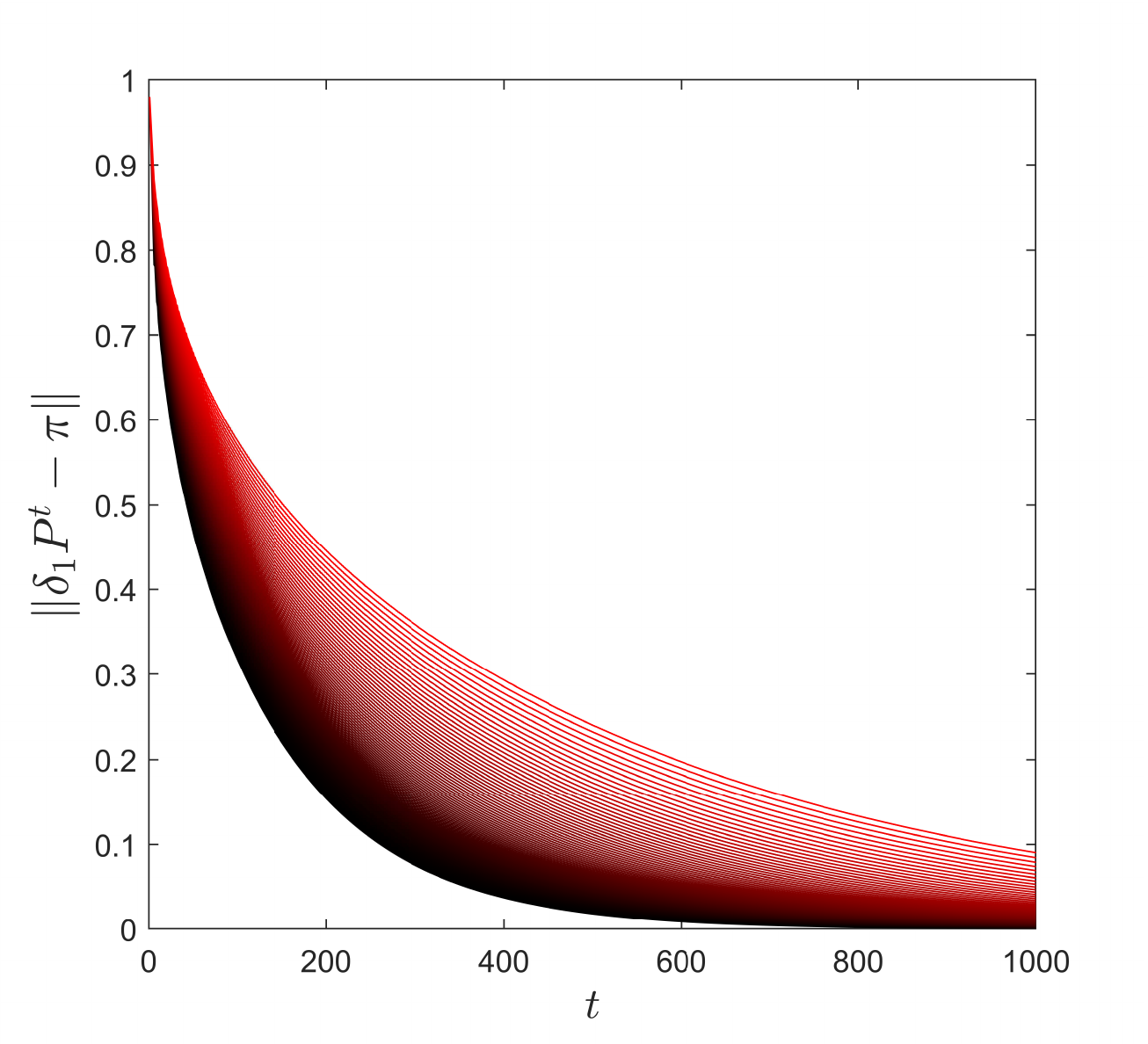}\includegraphics[scale=0.5]{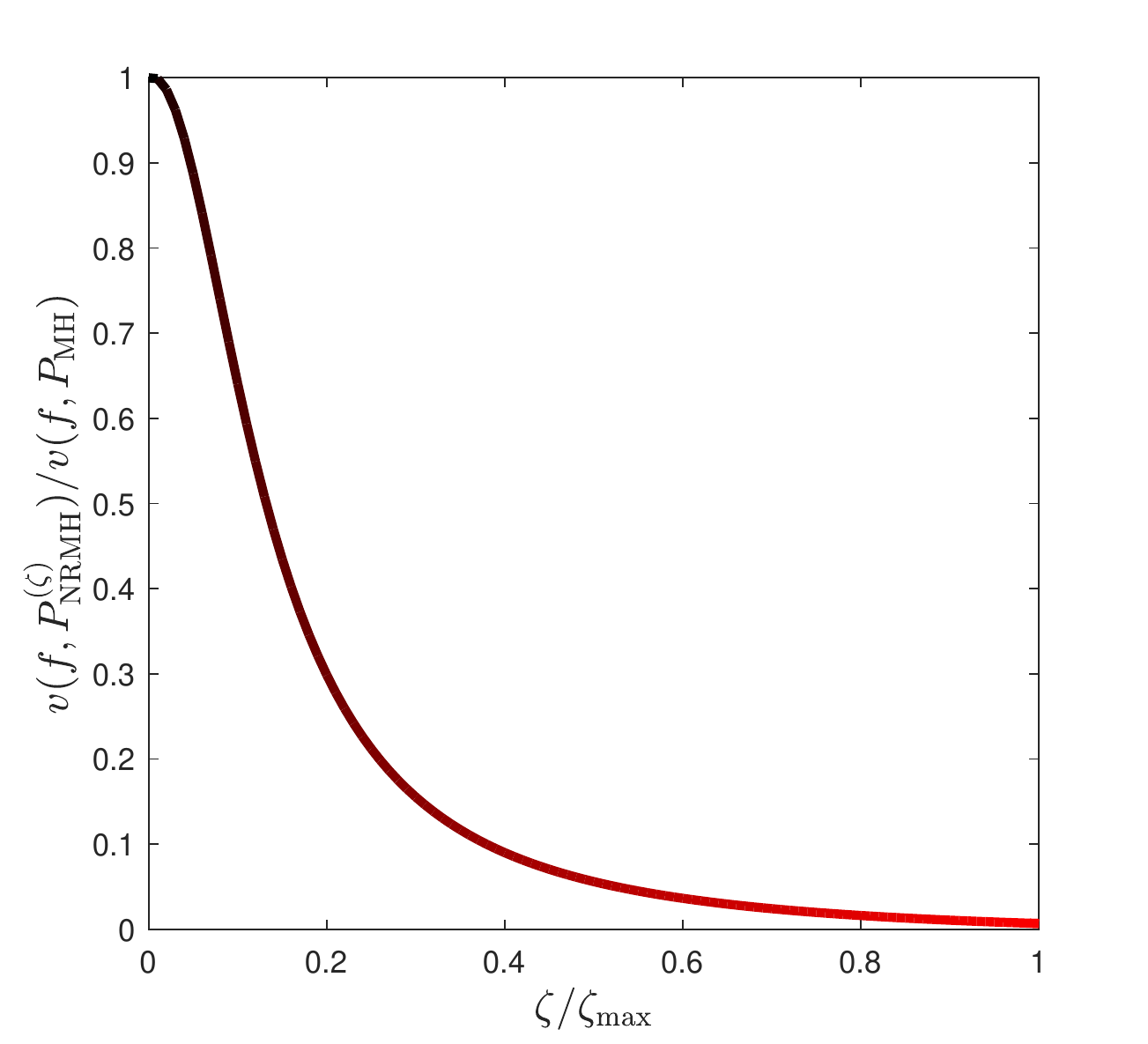}
\caption{ (Example \ref{ex_cercle}) Parameters: $S = 50$ and $\epsilon = 10^{-1}$. \textbf{Left panel} -- Evolution of the TV distance $\|\delta_1 P_{\zeta}^{t}-\pi\|$ in function of $t$ and $\zeta$ such that $\zeta/\zeta_{\text{max}}\in[0,1]$. Visual representation by colors ranging from black to red when $\zeta/\zeta_{\max}$ increases in $[0,1]$.  \textbf{Right panel} -- Ratio between the asymptotic variance of the MCMC estimate of $\mathbb{E}(X)$ provided by NRMH and MH, for the same series of $\zeta$ and the same color code as the right panel. Since NRMH with $\zeta=0$ is MH, the black plots correspond to MH while the red ones stand for NRMH with the largest non-reversibility parameter $\zeta=\zeta_{\max}$. \label{NRMH_circle}}
\end{figure}

\noindent Figure \ref{NRMH_circle} reports the efficiency of NRMH in the context of Example \ref{ex_cercle} for different values of $\zeta$ such that $\zeta/\zeta_{\max}\in [0,1]$. The phenomenon already observed on the other examples occurs here again: the most ``irreversible chain'' is the most asymptotically efficient but also the slowest to converge. It is thus necessary to use a skew-symmetric perturbation with intermediate intensity to reduce MH asymptotic variance without compromising its distributional convergence. At this point, one can wonder if there is a more elegant way to address this tradeoff. Indeed, finding an optimal parameter $\zeta\in[0,\zeta_{\max}]$ (for a specific criterion) is probably challenging and problem-specific. To mitigate the risk of slow convergence due to the fact that the vorticity field $\Gamma$ might not be well-suited for the topology of $\pi$, a natural idea would be to alternate, in some way, the Markov kernels using the fields $\Gamma$ and  $-\Gamma$. This is precisely the purpose of the following Section.

To motivate the following Section, we provide a short analysis of NRMH in the context of Example \ref{ex2} which exhibits some strong asymmetrical features that we expect to be favorable for a non-reversible sampler to converge faster than a random walk, provided that $\Gamma$ is \textit{well-chosen}. Again, we denote by
$
\left\{P_{\zeta}\,,\,\zeta\in[-\zeta_{\text{max}},\zeta_{\text{max}}]\right\}\,,
$
the family of NRMH transition kernels with the vorticity field $\Gamma_\zeta$ defined at Eq. \eqref{eq:Gamma_mat} and parameterized by $\zeta$ (so that $P_0$ is the MH kernel). When $\zeta>0$, NRMH increases the probability to transition to the state located in the counterclockwise direction, while with $\zeta<0$ the probability to transition to the state located in the clockwise direction is increased. Since $\pi(k)\propto k$, choosing $\zeta>0$ is expected to speed up the convergence as the vorticity field follows the probability mass gradient. The top row of Figure \ref{NRMH_ex2} partially confirms this statement, at least in the asymptotic regime. Indeed, in the transient phase since the initial distribution is $\delta_{1}$, setting $\zeta<0$ allows to quickly reach the high density region which consists of states located in the clockwise direction of $\{X=1\}$. Hence, comparing the convergence rate of the transient phase, the function $\zeta\mapsto\|\pi-\delta_1 P_{\zeta}^t\|$ increases when $\zeta$ browses $[-\zeta_{\text{max}},\zeta_{\text{max}}]$. However, near the stationary regime the initial distribution influence is minor and, as anticipated, the asymptotic convergence rate is much faster for $\zeta=\zeta_{\text{max}}$ than for $\zeta=-\zeta_{\text{max}}$. As a quantitative illustration, the following bounds hold for $S\leq 1000$ but we speculate that they hold for all $S$.

\begin{proposition}\label{prop:nrmh:tv:ex2}
  In the context of Example \ref{ex2} the following bounds hold for $S\leq 1000$:
  \begin{eqnarray}
    \|\delta_1 P_{0}^t-\pi\|^2 &\leq & \frac{2}{S(S+1)}\left(1-\frac{9}{S^2}\right)^{2t}\,, \\
    \|\delta_1 P_{\zeta_{\mathrm{max}}}^t-\pi\|^2 &\leq & \frac{1}{2S(S+1)}\left(1-\frac{17}{S^2}+\frac{30}{S^3}\right)^{t}\,, \\
    \|\delta_1 P_{-\zeta_{\mathrm{max}}}^t-\pi\|^2 &\leq & \frac{1}{2S(S+1)}\left(1-\frac{6}{S^2}+\frac{8}{S^3}\right)^{t}\,.
  \end{eqnarray}
\end{proposition}

\begin{proof}
The first bound is a simple application of \cite[Proposition 3]{diaconis1991geometric}. With symbolic calculation, we conjectured that $\sup\{|\lambda|\,:\,\lambda\in\text{Sp}(P_0)\backslash\{1\}\}=1+\alpha_0/S^2+o(1/S^2)$, which is upper bounded, for $S\leq 1000$, by $1-9/S^2$. This value was obtained by numerical adjustment.  For the NRMH bounds, we used \cite[Theorem 2.1]{fill1991eigenvalue} which can be seen as an extension of \cite[Proposition 3]{diaconis1991geometric} for non-reversible chains. This result follows from considering the multiplicative reversibilization of $P_\zeta$ defined as $M_\zeta:=P_\zeta P_\zeta^\ast$ where $P_\zeta^\ast$ is the self-adjoint of $P_\zeta$ in $\Ltwo_\pi$. Similarly to the reversible case, we find that the largest eigenvalue of $M_\zeta$ (restricted to $\Ltwo_{\pi,0}:=\{f\in\Ltwo_\pi,\,\int f\rmd \pi=0\}$) is $1+\alpha_\zeta/S^2+\beta_\zeta/S^3+o(1/S^3)$ and is upper bounded, for $S\leq 1000$, by $1-6/S^2+8/S^3+o(1/S^3)$ when $\zeta=-\zeta_{\text{max}}$ and $1-17/S^2+30/S^3+o(1/S^3)$ when $\zeta=\zeta_{\text{max}}$.
\end{proof}
\begin{remark}
  While the bounds of Prop. \ref{prop:nrmh:tv:ex2} are rather loose for NRMH (which can be explained by the embedding nature of the proof of \cite[Theorem 2.1]{fill1991eigenvalue}) they are quite accurate for MH. Moreover, for large $S$, all those bounds give very precise estimate of the asymptotic convergence rate. In particular, we have that
  $$
  \lim_{t\to\infty} r_{\zeta}(t) \|\delta_1 P_{\zeta}^t-\pi\|=0
  $$
  with
  $$
    r_{0}(t) =\exp{\frac{9t}{S^2}}\,,\quad
    r_{\zeta_{\mathrm{max}}}(t)=\exp{\frac{17t} {2S^2}}\,, \quad
    r_{-\zeta_{\mathrm{max}}}(t)=\exp{\frac{3t}{S^2}}\,.
    $$
 This shows that asymptotically, the convergence rate of NRMH with an appropriate vorticity field is similar to MH (it is in fact slightly faster, an information which is not reflected in the bounds of Prop. \ref{prop:nrmh:tv:ex2}). In contrast, a poor choice of vorticity field leads, in this example, to an asymptotic convergence rate inferior to MH.
\end{remark}
As for the asymptotic variance, for polynomial functions $p_n(x)=x^n$ ($n\in\nset$), we found that $P_{-\zeta}$  slightly dominates $P_{\zeta}$ (see for example $f=p_1=\text{Id}$ at Figure \ref{NRMH_ex2}, bottom row) while for functions $q_n(x)=x^{-n}$ ($n\in\nset$),  we found that $P_{\zeta}$ dominates $P_{-\zeta}$, and significantly so for large $n$ (see for example the limiting case $f=\1_{x=1}=\lim_{n\to\infty} q_n$ at Figure \ref{NRMH_ex2}, bottom row).

\begin{figure}
\centering
\hspace*{-0.7cm}
\includegraphics[scale=0.5]{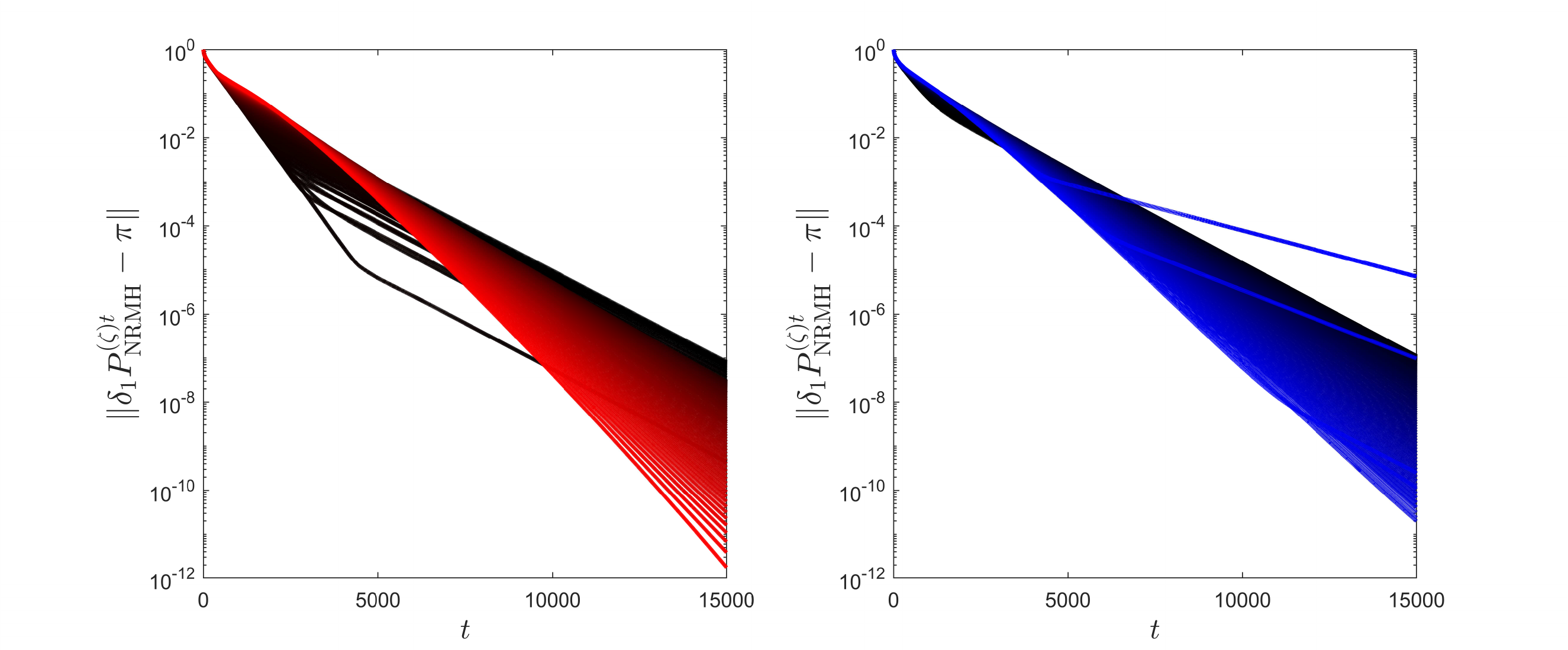}

\includegraphics[scale=0.5]{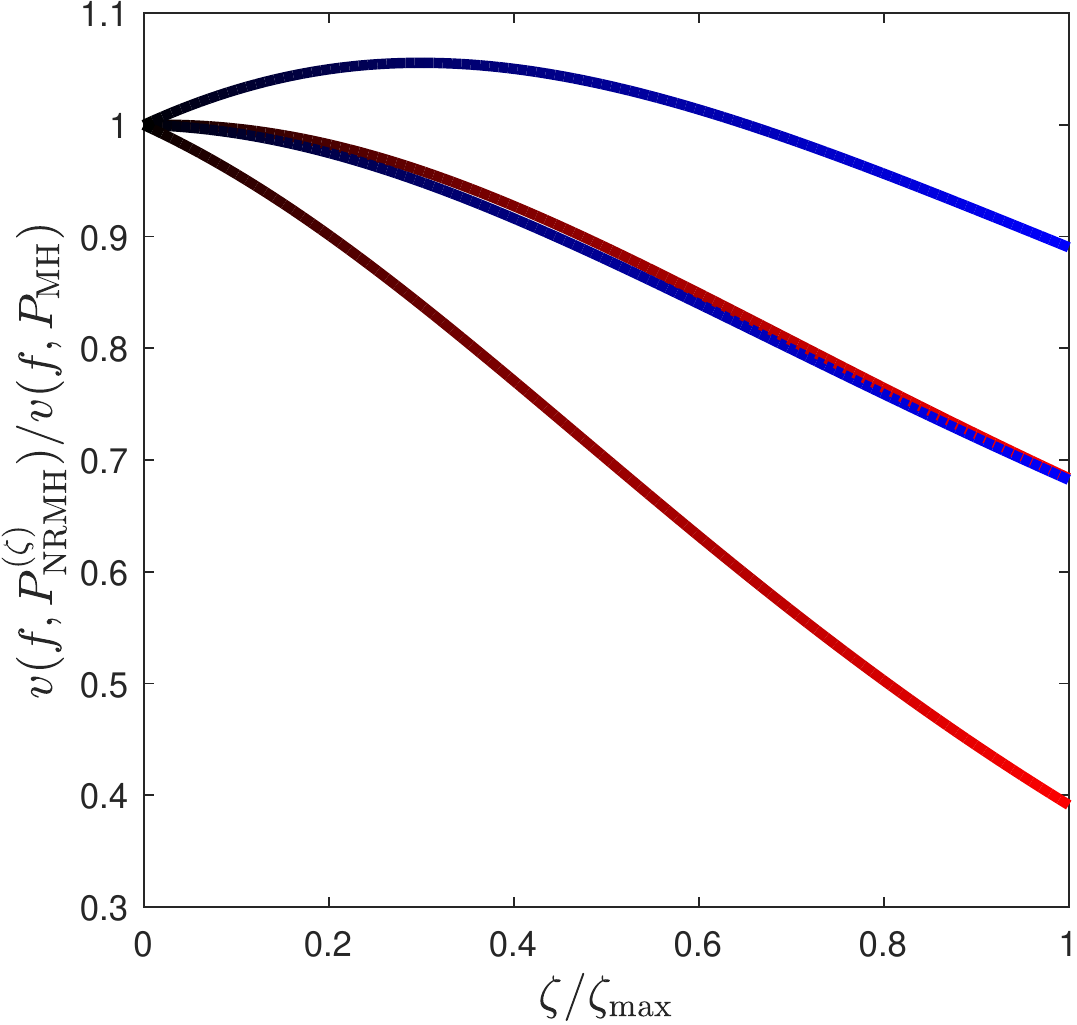}

\caption{ (Example \ref{ex2}) Parameters $S=101$. \textbf{Top row} -- Evolution of the TV distance $\|\delta_1 P_{\text{NRMH}}^{(\zeta)t}-\pi\|$ in function of $t$ and $\zeta$ such that $\zeta/\zeta_{\text{max}}\in[0,1]$ (on the left hand side, plots color ranging from black to red) and $\zeta/\zeta_{\text{max}}\in[-1,0]$ (on the right hand side, plots color ranging from blue to black). \textbf{Bottom row} -- Ratio between the asymptotic variance of the MC estimate of $\mathbb{E}(f(X))$ of NRMH and MH, for the same series of $\zeta$ and the same color code as the top row. For $\zeta$ such that $\zeta/\zeta_{\text{max}}\in[0,1]$, we can see that $P_{\text{NRMH}}^{(-\zeta)}$ dominates slightly $P_{\text{NRMH}}^{(\zeta)}$ for $f=\text{Id}$, while for $f=\1_{x=1}$, $P_{\text{NRMH}}^{(\zeta)}$ dominates significantly $P_{\text{NRMH}}^{(-\zeta)}$.  \label{NRMH_ex2}}
\end{figure}

\section{Two vorticity flows and a skew-detailed balance condition}
\label{sec:two_vort}
Let $\Gamma_1$ be a vorticity field satisfying Assumptions \ref{assumption1} and \ref{assumption2}. A sensible way to mitigate the risk that $\Gamma_1$ might not be well suited to sample from $\pi$ (see Section \ref{sec:NRMH}), is to combine two NRMH kernels with different non-reversible drifts, \ie $P_1\equiv P_{\Gamma_1}$ and $P_{-1}\equiv P_{\Gamma_{-1}}$, where $\Gamma_{-1}$ is some other vorticity field satisfying Assumptions \ref{assumption1} and \ref{assumption2}. We start with the following observation that combining two NRMH transition kernels with opposite vorticity fields in a \textit{blind} way is not necessarily advantageous.

\begin{proposition}
\label{prop:sec5:peskun}
Let $\Gamma$ satisfying  Assumptions \ref{assumption1}-\ref{assumption2} and let $P_{\Gamma}$ (resp. $P_{-\Gamma}$) be the NRMH transition kernel with proposal kernel $Q$ and vorticity field $\Gamma$ (resp. $-\Gamma$). Then for all $f\in\Ltwo(\pi)$,
$$
v(f,P_{\mathrm{MH}})\leq v\left(f,(1/2)P_{\Gamma}+(1/2)P_{-\Gamma}\right)\,.
$$
\end{proposition}

\begin{proof}
  The proof is postponed to \ref{proof:sec5:peskun}.
\end{proof}

In the spirit of the Guided Walk (see Section \ref{sec:lifted}), it would be desirable to embed a momentum variable $\zeta\in\{-1,1\}$ in the design of a Markov chain moving according to $P_\zeta$ until a NRMH candidate is rejected, at which point the momentum $\zeta$ is possibly switched (resulting in $\zeta'$) and sampling resumes with $P_{\zeta'}$. To construct such a scheme we follow the framework presented in \cite[Section 3.3]{andrieu2019peskun} and the two vorticity fields $\Gamma_1$ and $\Gamma_{-1}$ should satisfy the following assumption, referred to as a skew-detailed balance condition.
\begin{assumption}
\label{assumption3}
For all $(x,y) \in \mathcal{S}^2$,
\begin{equation}\label{eq:sdbe}
{\pi}(x)Q(x,y)A_{\Gamma_1}(x,y) = {\pi}(y)Q(y,x)A_{\Gamma_{-1}}(y,x)\,.
\end{equation}
\end{assumption}

\noindent The following observation indicates that Assumption \ref{assumption3} is rather strong and practically restricts the discussion of this Section to discrete state space sampling problems.

\begin{proposition}
\label{prop:sec5}
The vorticity fields $\Gamma_{1}$ and $\Gamma_{-1}$ satisfy Assumptions \ref{assumption1}, \ref{assumption2} and \ref{assumption3} if and only if $\Gamma_1=\Gamma_{-1}$ is the null operator on $\Scal\times\Scal$ or $\Gamma_1$ satisfies Assumptions \ref{assumption1}, \ref{assumption2}, $\Gamma_{-1}=-\Gamma$ and $Q$ is $\pi$-reversible.
\end{proposition}

\noindent The proof is postponed to \ref{proof_prop5}.

\begin{remark}
Beyond the trivial case $\Gamma_1=\Gamma_{-1}$, Proposition \ref{prop:sec5} shows that constructing a lifted NRMH using Assumption \ref{assumption3} requires a strong assumption on $Q$. One could for instance think to choose $Q\equiv P_{\text{MH}}$, a $\pi$-reversible MH transition kernel based on some  proposal kernel. While such a choice is rarely possible for general state space sampling problems, as it is typically impossible to evaluate the kernel $P_{\text{MH}}$ pointwise, it appears reasonable when the state space is discrete. In such a context, it can even be seen as a construction to dereversibilize MH.
\end{remark}

\noindent Proposition \ref{prop:sec5} imposes $\Gamma_{-1}=-\Gamma_1$ while Proposition \ref{prop:sec5:peskun} shows that in such a case, sampling the r.v. $\{X_t\,,t\in\nset\}$ and $\{\zeta_t\,,t\in\nset\}$ independently may lead to a Markov chain which is less efficient than Metropolis-Hastings. We consider Algorithm \ref{algo_NRMHAV}, refered to as Non-reversible Metropolis-Hastings algorithm with Auxiliary Variable (NRMHAV), which is parameterized by a proposal kernel $Q$ on $(\Scal,\Salg)$, a vorticity field $\Gamma$ and a refreshment rate $\varrho\in[0,1]$.

\begin{algorithm}[H]
\begin{algorithmic}[1]
\caption{\label{algo_NRMHAV} Non-reversible Metropolis-Hastings algorithm with auxiliary variable.}
\STATE Initialize in $(X_0,\zeta_0)\sim \mu_0$

Transition $(X_t,\zeta_t)=(x,\zeta)\to (X_{t+1},\zeta_{t+1})$:
\STATE Propose $Y \sim Q(x,\,\cdot\,)\rightsquigarrow y$
\STATE Set $(X_{t+1},\zeta_{t+1})=(y,\zeta)$ with probability $A_{\zeta\Gamma}(x,y) = 1 \wedge R_{\zeta\Gamma}(x,y)$ where
\begin{equation} \label{eq:NRMH_ratio}
R_{\zeta\Gamma}(x,y) := \begin{cases} \left(\zeta\Gamma(x,y) + \pi(y)Q(y,x)\right)\slash{\pi(x)Q(x,y)} & \mbox{ if } \pi(x)Q(x,y) \neq 0 \\ 1 & \mbox{ otherwise} \end{cases}
\end{equation}
\STATE If the move attempted at step 3: is rejected, set
\begin{equation}
(X_{t+1},\zeta_{t+1})=\left\{
\begin{array}{cc}
 (x,\zeta)& \text{with probability}\quad 1-\varrho\,,\\
 (x,-\zeta)&  \text{with probability}\quad \varrho\,.
 \end{array}
\right.
\end{equation}
\end{algorithmic}
\end{algorithm}
\noindent Algorithm \ref{algo_NRMHAV} simulates a Markov chain on the product space $(\Scal\times \{-1,1\},\Salg\otimes\{-1,1\})$ which is characterized by the following transition kernel:
\begin{multline}\label{NRMHAV_kernel}
K_\varrho\left((x,\zeta),(\rmd y,\zeta')\right) =
\1_{\{\zeta'=\zeta\}}\bigg\{\delta_{\Scal\backslash\{x\}}(\rmd y)Q(x,y)A_{\zeta\Gamma}(x,y)\\
+\delta_{x}(\rmd y)\left(Q(x,x)+(1-\varrho)\int_{\Scal}Q(x,\rmd z)(1-A_{\zeta\Gamma}(x,z))\right)\bigg\}\\
 \1_{\{\zeta'=-\zeta\}}\delta_{x}(\rmd y)\varrho\int_{\Scal}Q(x,\rmd z)(1-A_{\zeta\Gamma}(x,z))\,.
\end{multline}
\noindent The following Proposition gives conditions under which Algorithm \ref{algo_NRMHAV} generates a $\tpi$-invariant Markov chain, where we recall that for all $A\in\Salg$ and $\zeta\in\{-1,1\}$, $\tpi(A,\zeta)=(1/2)\pi(A)$ and $\tpi(A,\zeta)$ is null if $\zeta\not\in\{-1,1\}$. Under such conditions, the marginal collection of r.v. $\{X_t,\,t\in\nset\}$ is $\pi$-invariant.

\begin{proposition}
\label{thm_NRMHAV}
Let $\pi$ be a probability distribution on $\Scal\times\Salg$ whose density is nowhere null. Let $Q$ be a $\pi$-reversible kernel and $\Gamma$ be a vorticity field  such that Assumptions \ref{assumption1} and \ref{assumption2} hold. Then, for all $\varrho\in[0,1]$, the Markov chain $\{(X_t,\zeta_t),\,t\in\nset\}$ generated by Algorithm \ref{algo_NRMHAV} is ${\tpi}$-invariant. Moreover, $\{(X_t,\zeta_t),\,t\in\nset\}$ is $\tpi$-reversible if and only if $\Gamma=\mathbf{0}$.
\end{proposition}

\begin{proof}
This result is proved for a discrete state space $\mathcal{S}$ in \ref{proof_NRMHAV} but its extension to the general state space case is straightforward.
\end{proof}

It is possible to carry out the analysis of Algorithm \ref{algo_NRMHAV} in the light of the lifted Markov chain unifying framework developed in \cite[Section 3.3]{andrieu2019peskun}. The Authors consider two sub-stochastic kernels $T_1$ and $T_{-1}$ on the marginal space $(\Scal,\Salg)$ which satisfies a skew-detailed balance equation of the form $\pi(\rmd x) T_1(x,\rmd y)=\pi(\rmd y) T_{-1}(y,\rmd x)$. The generic lifted Markov chain is outlined at Algorithm \ref{algo_NRMHAV2} and considers a state-dependent switching rate $\rho\equiv \rho_\zeta(x)$, which does not correspond exactly to the refreshment rate $\varrho$ of Algorithm \ref{algo_NRMHAV} as we shall soon see. In particular, from \cite{andrieu2019peskun}, we know that in order to sample from $\tpi$, it is sufficient to design the switching rate so that it satisfies for all $\zeta\in\{-1,1\}$ and all $x\in\Scal$:
\begin{equation}
\label{eq:cdt_lifted}
0\leq \rho_\zeta(x)\leq 1-T_\zeta(x,\Scal)\quad\text{and}\quad \rho_\zeta(x)-\rho_{-\zeta}(x)\leq T_{-\zeta}(x,\Scal)-T_{\zeta}(x,\Scal)\,.
\end{equation}
In the sequel,   the Markov kernel associated to the lifted construction of Algorithm \ref{algo_NRMHAV2} with switching rate $\rho$ is denoted $L_\rho$.

\begin{algorithm}[H]
\begin{algorithmic}[1]
\caption{\label{algo_NRMHAV2} Lifted Markov chain in the framework of \cite{andrieu2019peskun}.}
\STATE Initialize $(X_0,\zeta_0)\sim \mu_0$

Transition $(X_t,\zeta_t)=(x,\zeta)\to (X_{t+1},\zeta_{t+1})$:

\STATE With proba. $\rho_\zeta(x)$, set $(X_{t+1},\zeta_{t+1})=(x,-\zeta)$,
\STATE With proba. $1-\rho_\zeta(x)-T_\zeta(x,\Scal)$, set $(X_{t+1},\zeta_{t+1})=(x,\zeta)$,
\STATE With proba. $T_{\zeta}(x,\Scal)$, draw
$$
X_{t+1}\sim \frac{T_\zeta(x,\cdot)}{T_\zeta(x,\Scal)}
$$
and set $\zeta_{t+1}=\zeta$.
\end{algorithmic}
\end{algorithm}

A special case of the lifting construction of Algorithm \ref{algo_NRMHAV2} occurs when $T_1(x,\Scal)=T_{-1}(x,\Scal)$, for all $x\in\Scal$. Interestingly, in such situation, the conditions of Eq. \eqref{eq:cdt_lifted} boil down to
\begin{equation}\label{eq:cdt_lifted_nrmh}
  0\leq \rho_\zeta(x)\leq 1-\int_{\Scal} Q(x,\rmd y)A_{\zeta\Gamma}(x,y)\quad \text{and}\quad  \rho_\zeta(x)=\rho_{-\zeta}(x)\equiv\rho(x)\,,
\end{equation}
for all $(x,\zeta)\in\Scal\times\{-1,1\}$. The refreshment rate $\rho$ can thus be set as a constant, in which case, it should satisfy
\begin{equation}\label{eq:cst:swt:rte}
0\leq \rho\leq 1-\sup_{x\in\Scal} \int_{\Scal} Q(x,\rmd y)A_{\zeta\Gamma}(x,y)\,.
\end{equation}
In particular, this leaves the possibility to choose $\rho=0$. A worthy consequence of Proposition 3.5 in \cite{andrieu2019peskun} is to note that,  among all the constant switching rates satisfying Eq. \eqref{eq:cst:swt:rte}, the choice $\rho=0$ minimizes the function $\rho\mapsto v_\lambda(f,L_\rho):=\|f\|_\pi^2+2\sum_{k\geq 1}\lambda^k\pscal{f}{L^kf}$, for any $f\in\Ltwo_0(\pi)$ and $\lambda\in[0,1)$, which is closely related to $v(f,L_\rho)$ since, if the limit exists, $\lim_{\lambda\to 1}v_\lambda(f,P)=v(f,P)$.

In the lifted NRMH (Alg. \ref{algo_NRMHAV}) context $T_\zeta(x,\rmd y)=Q(x,\rmd y)A_{\zeta\Gamma}(x,y)$ and, according to Lemma \ref{lemma1}, $\int_\Scal Q(x,\rmd y)A_{\zeta\Gamma}(x,y)=\int_\Scal Q(x,\rmd y)A_{-\zeta\Gamma}(x,y)$, for all $x\in\Scal$.  It should be noted that Algorithm \ref{algo_NRMHAV} with a refreshment rate $\varrho\in[0,1]$ coincides with Algorithm \ref{algo_NRMHAV2} with a switching rate
$$
\rho(x)=\left(1-\int_{\Scal} Q(x,\rmd y)A_{\zeta\Gamma}(x,y)\right)\varrho\,,
$$
and hence Eq. \eqref{eq:cdt_lifted_nrmh} is satisfied.
The previous analysis shows that for any $\lambda\in[0,1)$, the function $\varrho\mapsto v_\lambda(f,K_\varrho)$ is monotonically increasing with $\varrho$, hence the choice $\varrho=0$ is expected to lead to the smallest asymptotic variance among all the lifted NRMH algorithms (Alg. \ref{algo_NRMHAV}). However, quantifying the rate of convergence of $K_\varrho$ to $\tpi$ is not straightforward and we illustrate the algorithm in the context of Examples \ref{ex2} and \ref{ex_cercle}. Since the asymptotic variance is minimized for $\varrho=0$, it would be informative to assess the rate of convergence of NRMHAV in function of $\varrho$ and especially for $\varrho$ close to zero.

\setcounter{example}{1}
\begin{example}[continued]
 Figure \ref{fig:ex2:nrmhav} (Appendix) illustrates the mixing time of NRMHAV, for this example. If $\varrho=0$, the joint process $\{(X_t,\zeta_t),\,t\in\nset\}$ does not converge to $\tpi$ (since the momentum is fixed) but nevertheless $\{X_t,\,t\in\nset\}$ converges to $\pi$ marginally since for all $(x,\zeta)\in\Scal\times\{-1,1\}$, $\delta_{\{x,\zeta\}}K_0=\delta_{x}P_{\zeta}$ which is NRMH with vorticity field $\zeta_0\Gamma$. In this example, $\varrho=0$ turns out to be the NRMHAV optimal parameter, both for the asymptotic variance and for the convergence rate.
\end{example}


\setcounter{example}{2}
\begin{example}[continued]
NRMH with auxiliary variable (Alg. \ref{algo_NRMHAV}) is used to sample from the distribution $\pi$ of Example \ref{ex_cercle}. Figure \ref{fig:ex3:sec5} illustrates this scenario. Since $Q$ is $\pi$-reversible, \cite[Theorem 3.15]{andrieu2019peskun} applies and yields that for all $f\in\Ltwo(\pi)$ and all $\varrho\in(0,1]$, $v(f,P_{\text{NRMH}})=v(f,K_{0})\leq v(f,K_{\varrho})$. Hence, the variance reduction effect of NRMHAV (relatively to MH) is less significant than NRMH's but, remarkably, there exists a certain range of parameters $\varrho$ for which NRMHAV converges faster than NRMHAV and MH while still reducing MH's asymptotic variance. For example, if $S=100$ and $\varrho=0.003$, NRMHAV converges faster than MH and NRMH while reducing the MH's asymptotic variance by a factor of about $100$ (which is ten times less than NRMH's variance reduction factor). Hence the introduction of a vorticity matrix with a large inertia parameter coupled with a direction switching parameter with reasonably low intensity leads, in this example, to a better algorithm than MH. NRMHAV inherits from NRMH its variance reduction feature while avoiding its dramatically slow speed of convergence.
\end{example}

\begin{figure}
\centering

\includegraphics[scale=0.5]{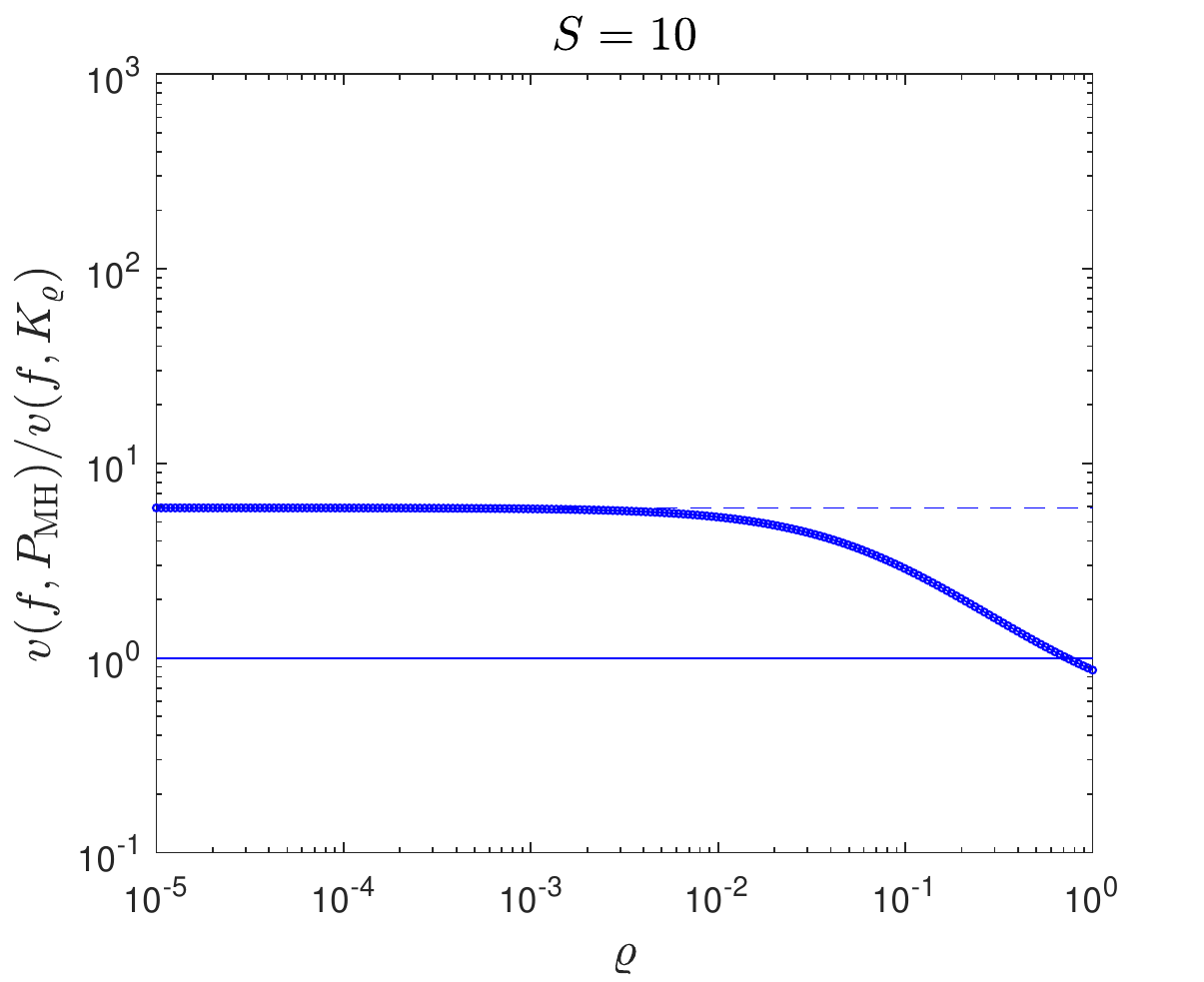}\includegraphics[scale=0.5]{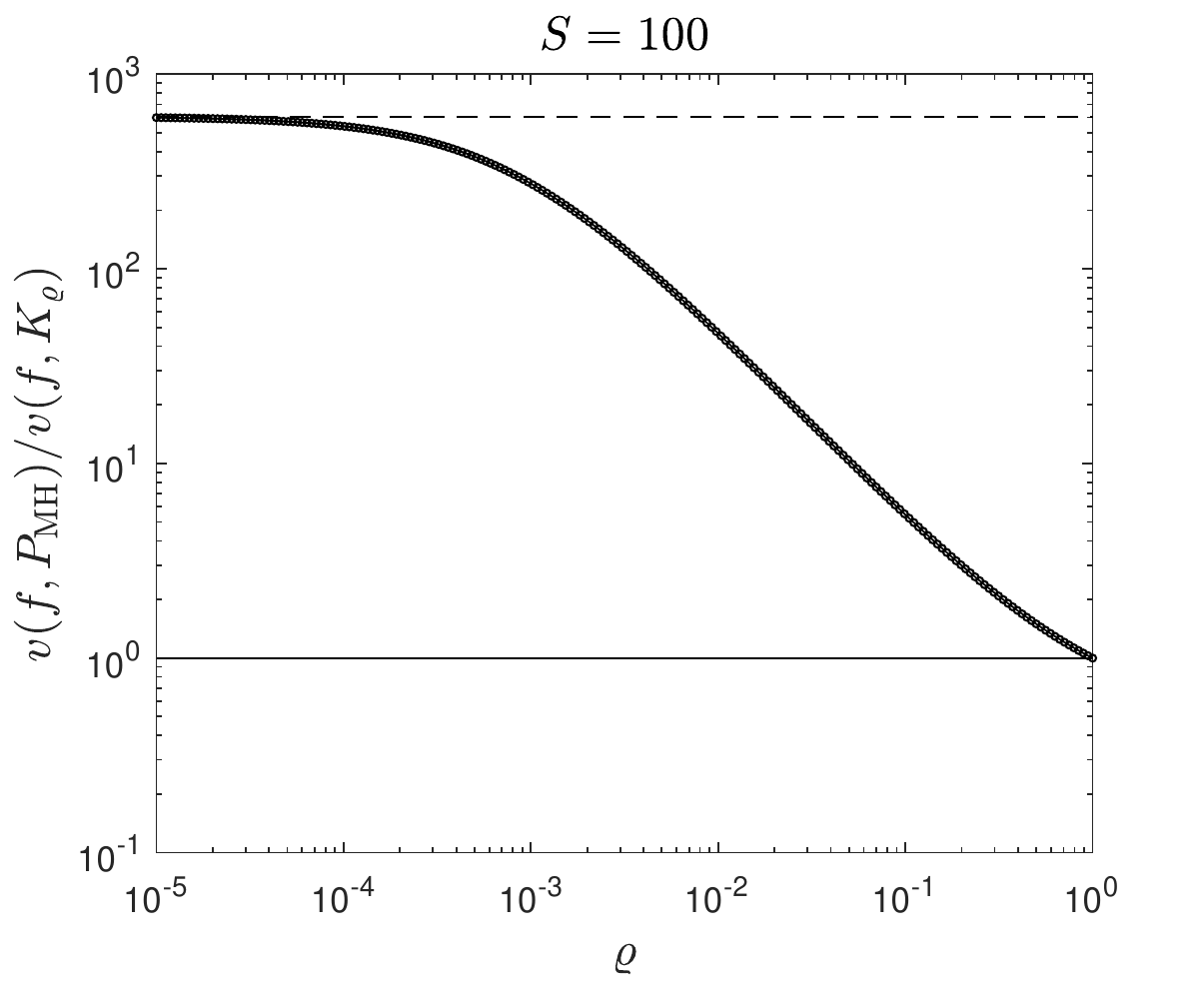}

\includegraphics[scale=0.5]{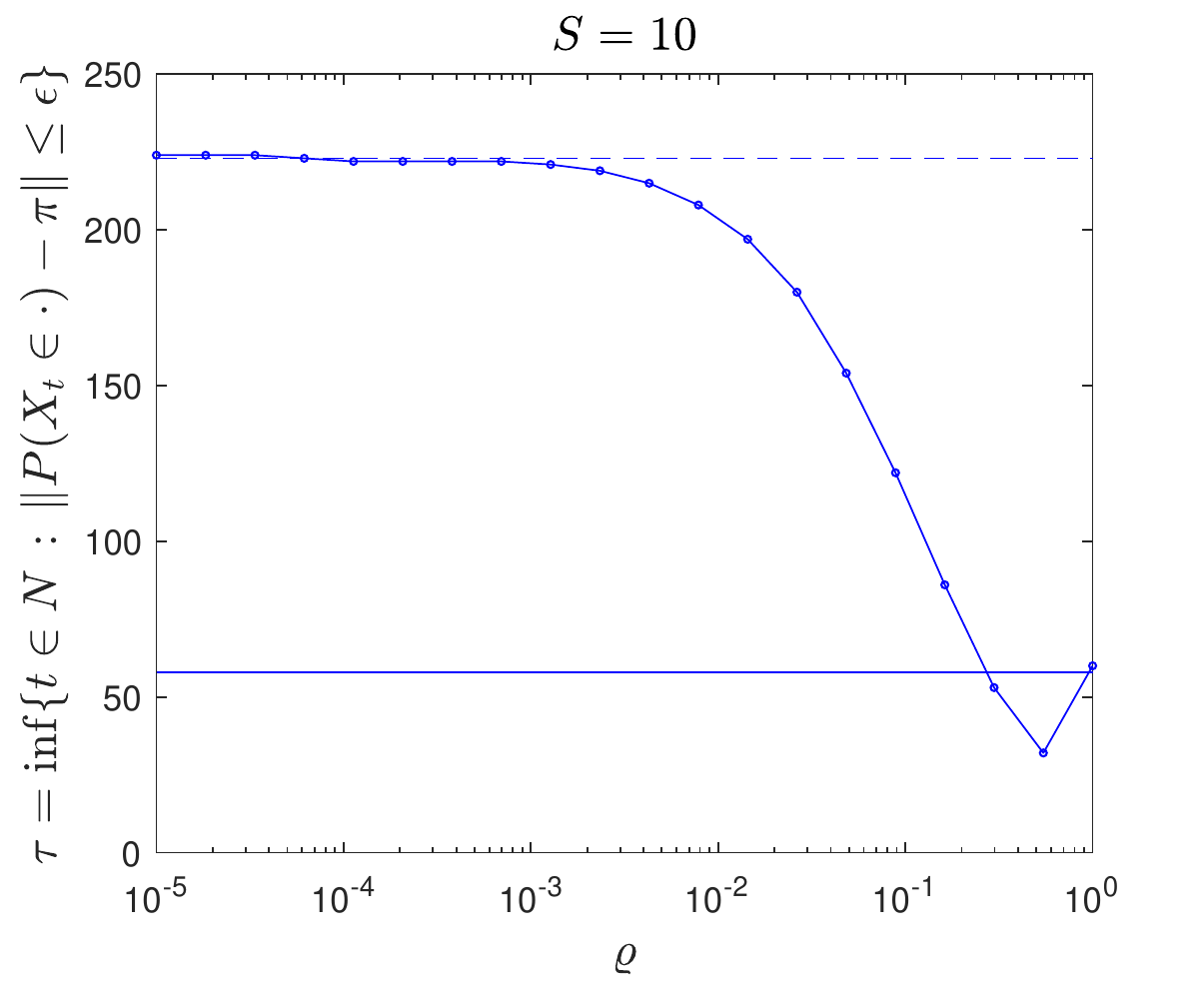}\includegraphics[scale=0.5]{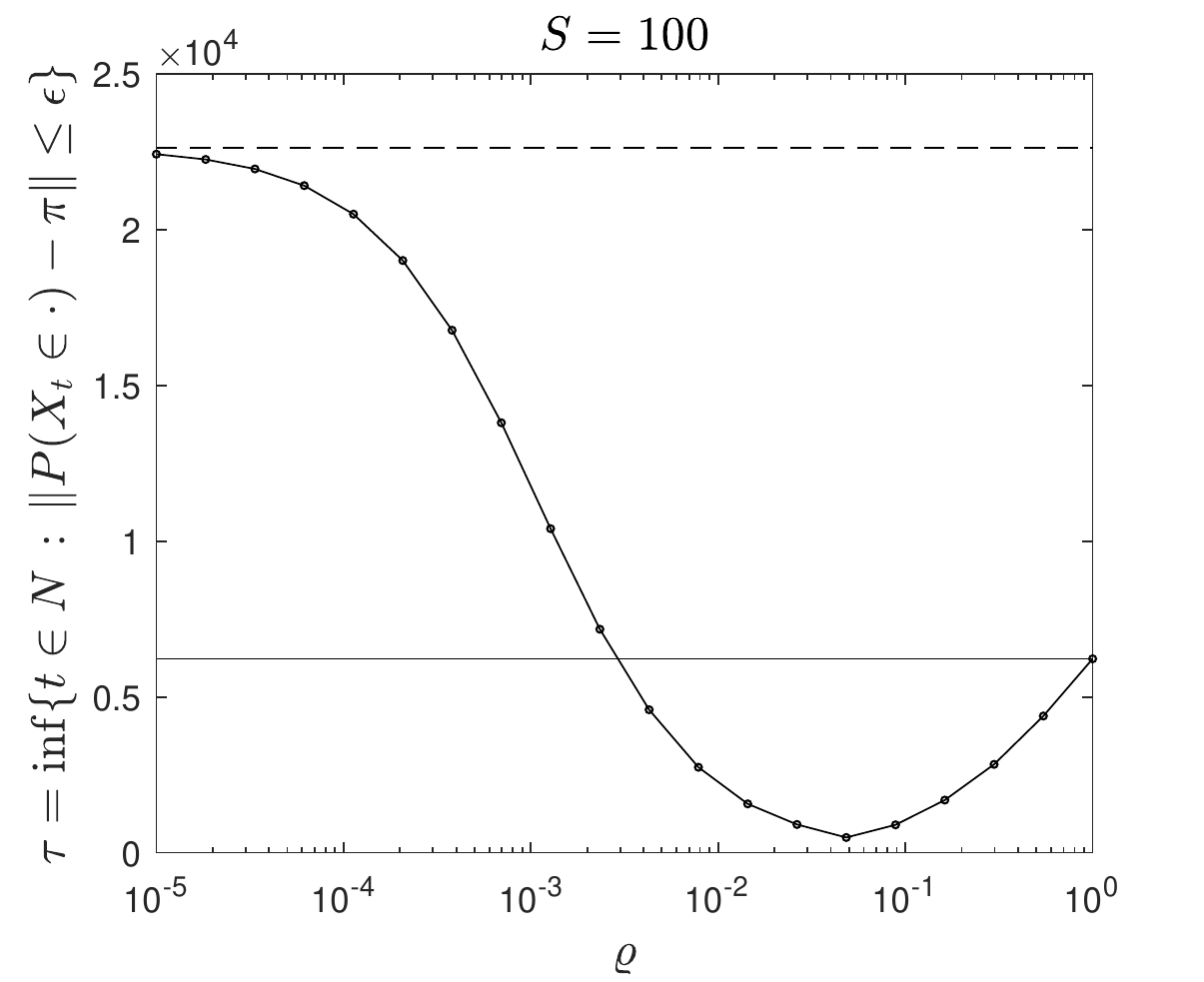}

\caption{(Example \ref{ex_cercle}) Asymptotic variance (top row) and mixing time (bottom row) achieved by NRMHAV, NRMH (which is NRMHAV with $\varrho=0$) and MH (which is NRMHAV with $\Gamma=\mathbf{0})$ for $S=10$ (left) and $S=100$ (right). The plain line corresponds to MH and the dashed one to NRMH. For NRMH and NRMHAV, the vorticity field was set to $\Gamma\equiv \Gamma_{\zeta_{\text{max}}}$ where for some $a>0$, $\Gamma_a$ is defined at Eq. \eqref{eq:Gamma_mat} with $\zeta_{\text{max}}=(1-\epsilon)/2S$. The test function for the asymptotic variance was set as $f=\text{Id}$.\label{fig:ex3:sec5}}
\end{figure}

\noindent We conclude this Section with an example which illustrates graphically the benefits of Algorithm  \ref{algo_NRMHAV}.

\begin{example}
  \label{ex4}
  The state space $\Scal$ is defined as $\Scal=\{1,\ldots,S\}\times\{1,\ldots,S\}$. The distribution $\pi$ is illustrated at Figure \ref{fig:ex4_pi}: it has a sigma shape with uniform mass located at the centre of the space and some variations of probability mass along the vertical direction near the vertical borders. We consider the family of MH Markov chains with proposal kernel $Q$ which attempts to move to any neighboring state (north, east, west, south) with the same probability. For states at the boundary of $\Scal$, the proposal kernel allows to jump only to two or three neighboors, \ie $Q$ has a nontoroidal support. In addition to MH, we consider NRMH (Alg. \ref{algo_NRMH}) with two possible vorticity fields $\Gamma$ and $-\Gamma$. Because constructing a matrix $\Gamma$ which satisfies Assumption \ref{assumption1} is not straightforward when considering nontoroidal random walk kernels, the reader is referred to \ref{app:vorticity_grid} for the description of a technique that generates valid vorticity fields for this type of proposal. We consider four Markov kernels: MH, NRMH1 and NRMH2, which use $\Gamma_{\zeta_{\text{max}}}$ and $-\Gamma_{\zeta_{\text{max}}}$ respectively, and NRMHAV, which uses $P_{\text{MH}}$ for proposal kernel and some switching rate parameter $\varrho\in[0,1]$. Note that in the case $\varrho=0$, the NRMHAV kernel does not coincide with either NRMH kernels since their two proposal kernels are different. For each Markov kernel, assessment of the convergence rate and of the asymptotic variance $v(f,P)$ with $f=\text{Id}$ is reported at Figure \ref{fig:ex4_pi}.
\end{example}

The characteristics of NRMHAV in Example \ref{ex4} vary significantly in function of $\varrho\in[0,1]$. The NRMHAV asymptotic variance, which increases monotonically with $\varrho$, can even be smaller than NRMH, if $\varrho$ is small enough. On the convergence front, the pattern observed in the previous examples is again repeated: NRMHAV converges faster when $\varrho$ increases from zero until a tipping point $\varrho^\ast$ is reached and past which, the convergence gets slower monotonically for $\varrho>\varrho^\ast$. Around $\varrho^\ast$, the NRMHAV asymptotic  variance (for $f=\text{Id}$) is smaller than MH but the variance reduction factor is about half of what is achieved with NRMH1 and NRMH2. NRMHAV can thus be seen as a achieving a tradeoff between MH and NRMH. For illustrative purpose, the convergence of $|\Pr(X_t=k\,|\,X_0)-\pi_k|$ to $0$ for each $k\in\Scal$ and for each Markov chain MH, NRMH1, NRMH2 and NRMHAV is reported at Figure \ref{fig:ex4_cv_err} for the three time steps $t\in\{100,5000,40000\}$. An animation of the convergence is available online at \href{https://maths.ucd.ie/~fmaire/vm19/nrmhav_ex4_S3600.gif}{https://maths.ucd.ie/$\sim$fmaire/vm19/nrmhav\_ex4\_S3600.gif}. It can be seen that the asymptotic rate of convergence of NRMHAV is faster than for the three other samplers. Indeed, for $t>15000$ it can be seen that for NRMHAV the graphical illustration of $\{|\Pr(X_t=k\,|\,X_0)-\pi_k|,\,k\in\Scal\}$ has much darker points than the other  samplers, indicating that $\{|\Pr(X_t=k\,|\,X_0)-\pi_k|,\,k\in\Scal\}$ is closer to zero than for the other samplers. Because of the symmetry of $\pi$, NRMH1 and NRMH2 converges at the same rate. Indeed, in some areas the vorticity field $\Gamma_{\zeta_\text{max}}$ is more adapted than $-\Gamma_{\zeta_\text{max}}$ to the topology of $\pi$ and conversely for some other parts of the state space. In other words, the non-reversibility speeds up the convergence in some areas and slows it down in some other areas, which explains the oscillating nature of the NRMH convergence animation available online and that, on average, NRMH asymptotic convergence rate is similar to MH. By contrast, this oscillating convergence pattern is suppressed by the fact that NRMHAV alternate in a relevant fashion between the two vorticity fields and this yields a faster asymptotic convergence rate.

\begin{figure}
\centering
\includegraphics[scale=0.56]{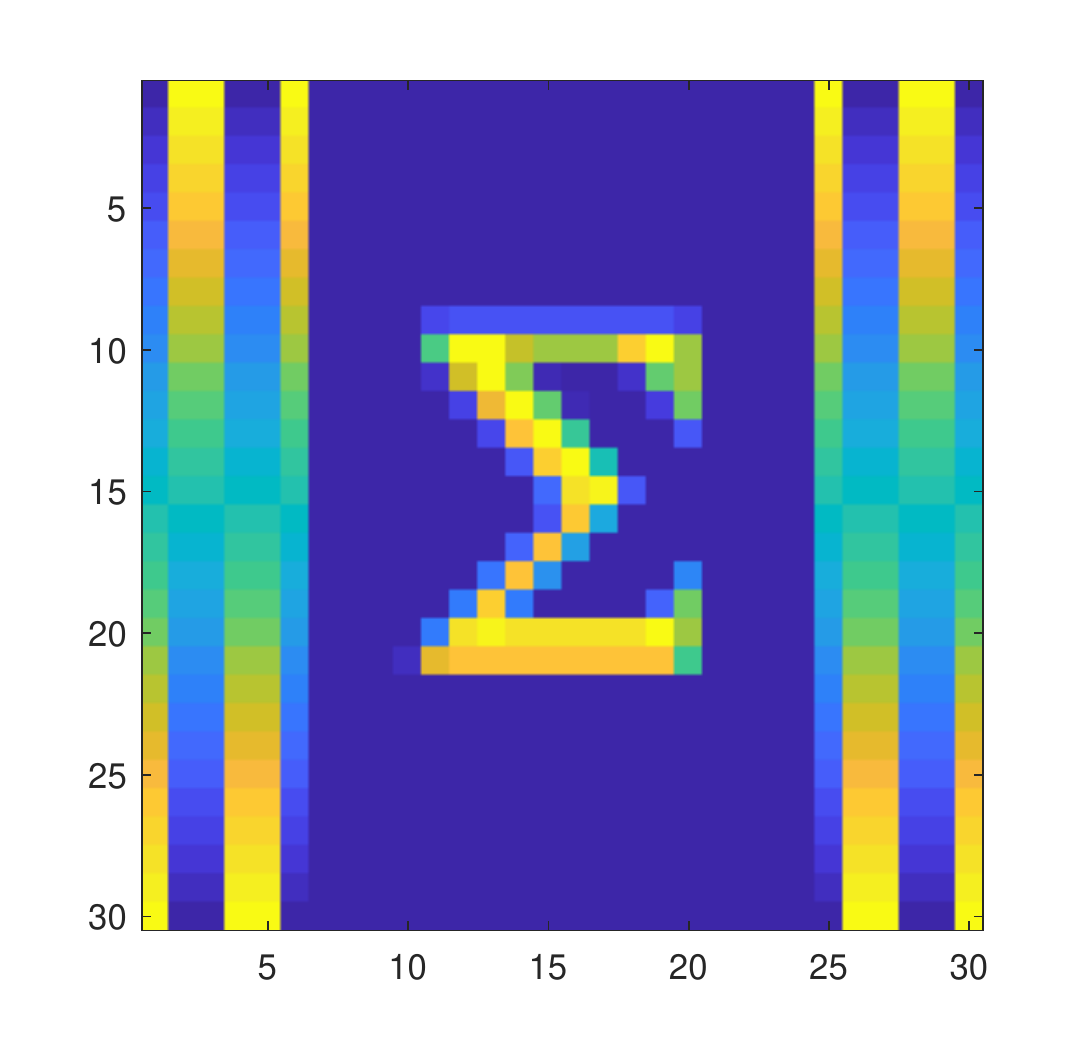}\includegraphics[scale=0.5]{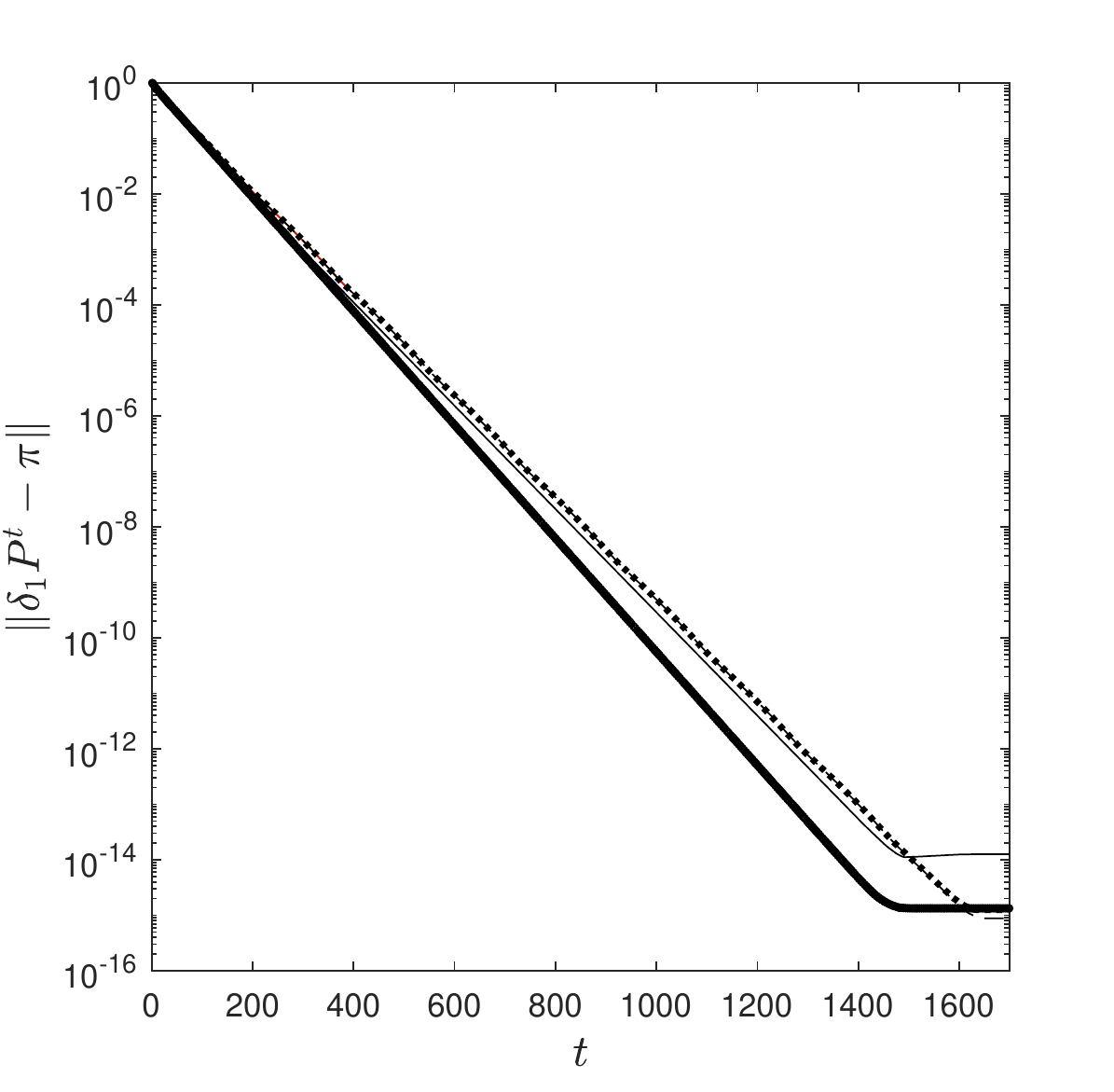}

\includegraphics[scale=0.5]{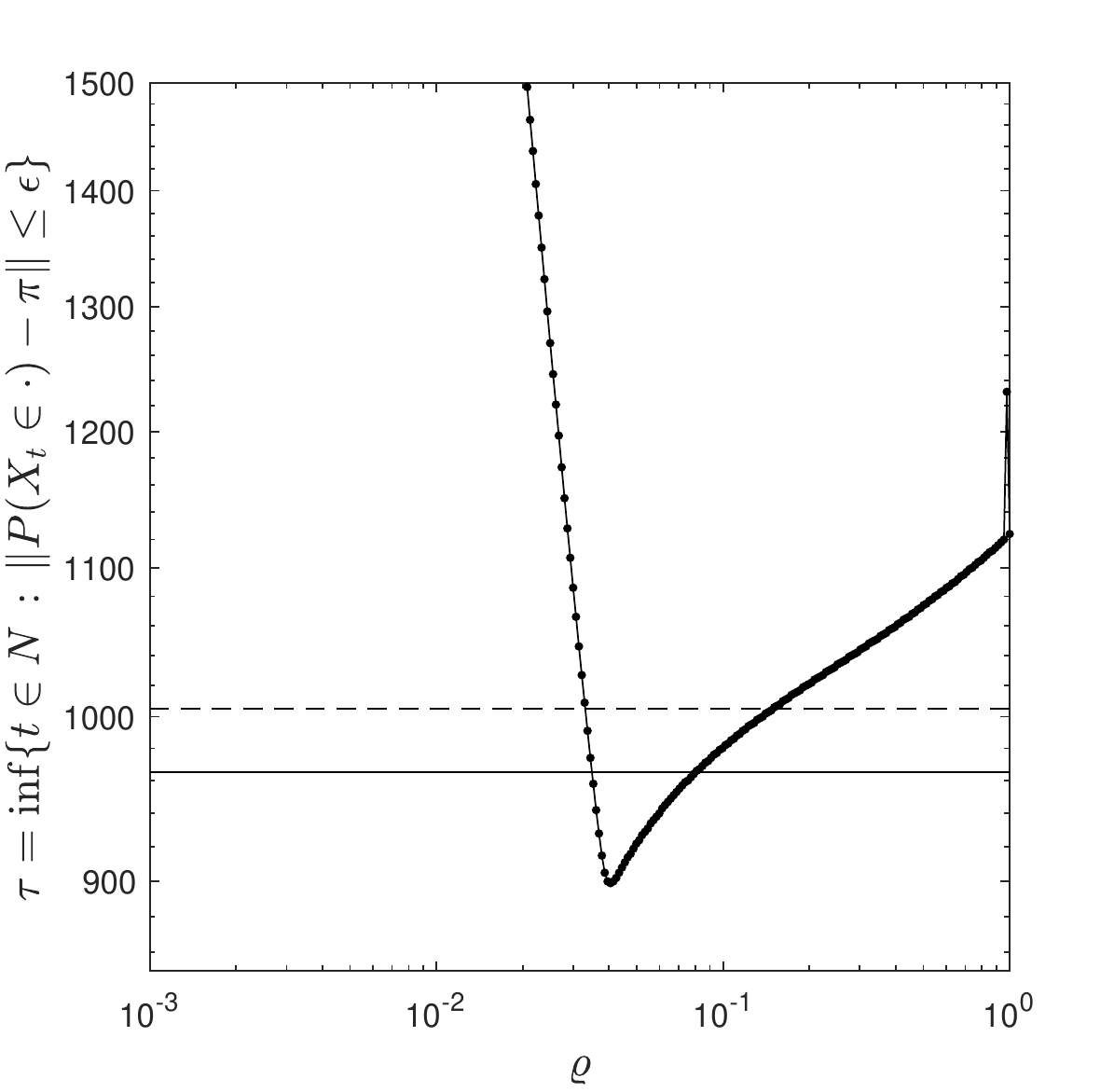}\includegraphics[scale=0.5]{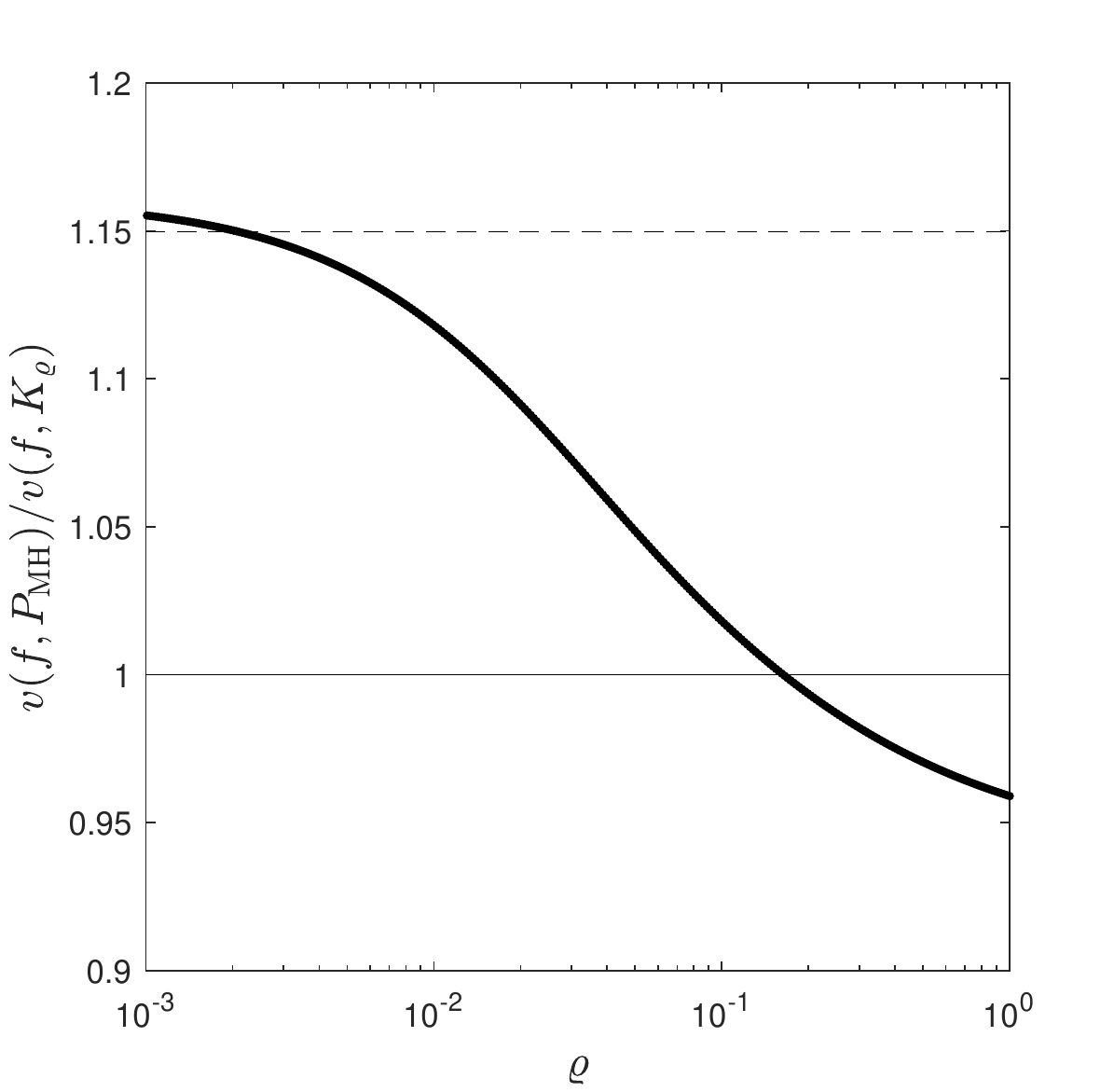}

\caption{(Example \ref{ex4}, with $S=30$) Illustration of $\pi$ (top left-hand corner) ; note that for all $(x,y)\in\Scal^2$, $\pi(x)/\pi(y)<3/2$. The top right-hand corner shows the convergence of MH (plain light line), NRMH1 and NRMH2 (dashed and dotted lines, which overlap) and NRMHAV (thick line). The bottom row shows the efficiency of the Markov chains via their mixing time (left) and their asymptotic  variance (right). It can be seen that for the switching rate which maximizes the NRMHAV convergence rate, its asymptotic variance (for the identity function) is smaller than MH but the variance reduction factor is about half of what is achieved with NRMH1 and NRMH2.\label{fig:ex4_pi}}
\end{figure}

\begin{figure}
\centering
\includegraphics[scale=0.56]{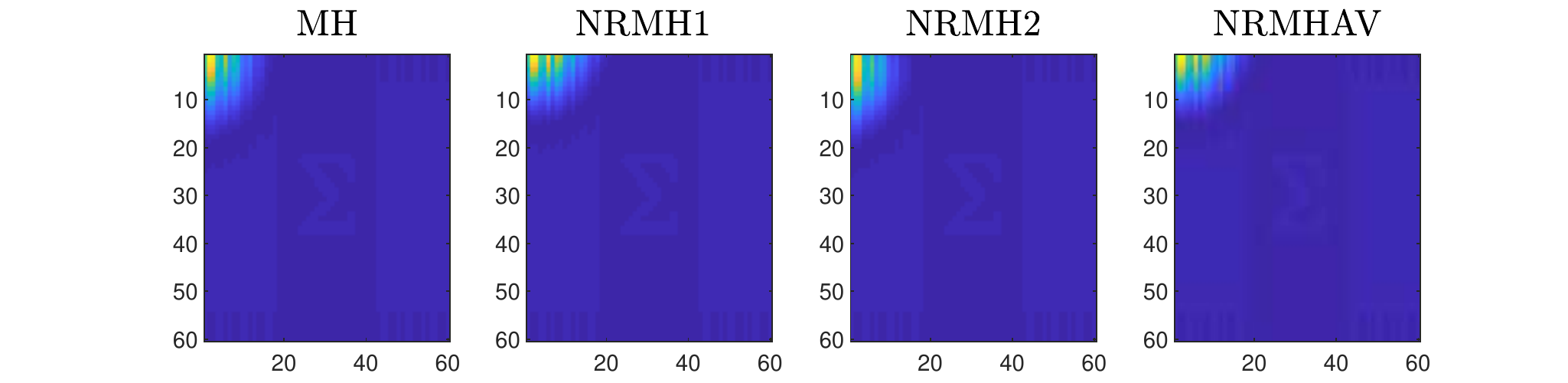}
\includegraphics[scale=0.56]{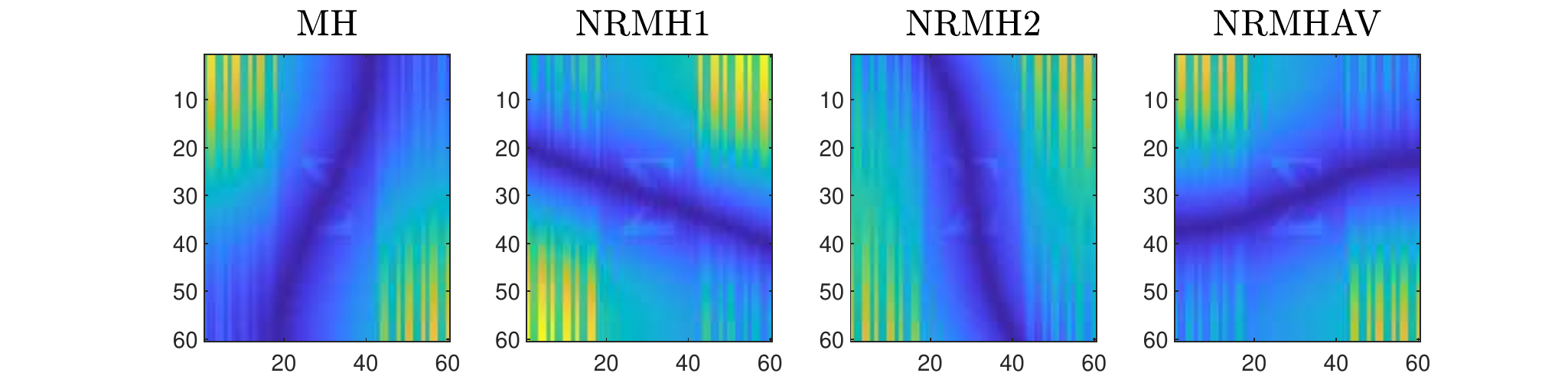}
\includegraphics[scale=0.56]{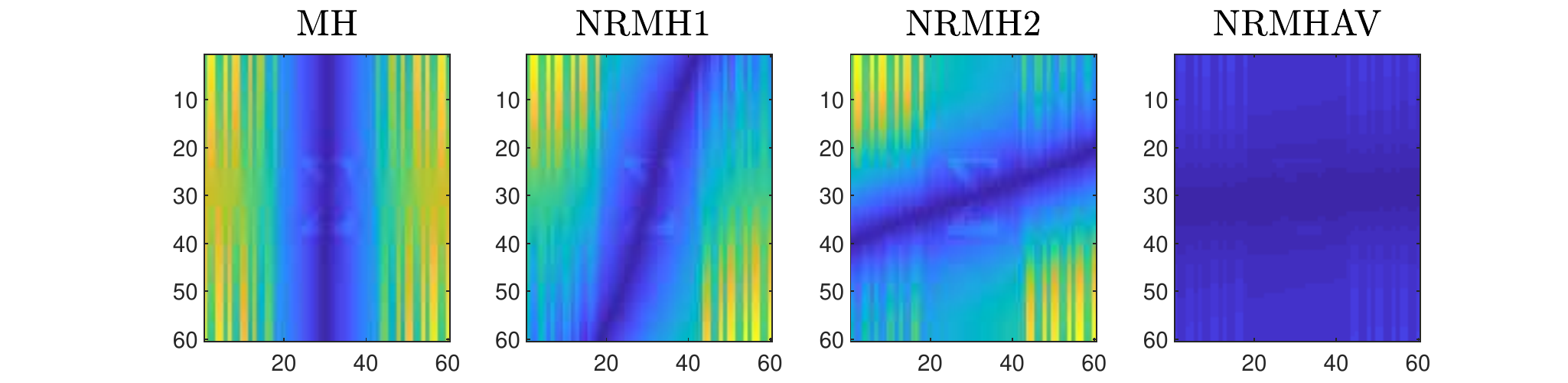}
\caption{(Example \ref{ex4}, $S=60$) Evolution of $\{|\Pr(X_t=k\,|\,X_0)-\pi_k|\,,k\in\Scal\}$ for $t=100$ (top row), $t=5000$ (middle row) and $t=40,000$ (bottow row), for the four Markov chains, all started at $X_0=\{1,1\}$. The entire sequence can be seen at \href{https://maths.ucd.ie/~fmaire/vm19/nrmhav_ex4_S3600.gif}{https://maths.ucd.ie/$\sim$fmaire/vm19/nrmhav\_ex4\_S3600.gif}. Note that all the images of a row have the same colormap so that the convergence of the different samplers can be visually compared. \label{fig:ex4_cv_err}}
\end{figure}

\section{Discussion}

Non-reversible Markov chains are often thought to be faster to converge than their reversible counterpart. In an attempt to clarify such a statement, this work has investigated several questions related to the convergence of two Metropolis-Hastings based non-reversible MCMC algorithms, namely the Guided Walk (GW) \cite{gustafson1998guided} and the non-reversible Metropolis-Hastings (NRMH) \cite{bierkens2016non}. This research has not developed new tools to analyse those Markov chains but has instead applied several existing frameworks (see \eg \cite{diaconnis2000,andrieu2019peskun}) to a collection of examples. This effort has allowed to gain a quantitative insight on how those two different constructions generating non-reversibility (GW and NRMH) alter the Markov chain convergence.

A first step was to embed non-reversible kernels in a framework which encompasses their reversible version, see Sections \ref{sec:lifted} and \ref{sec:NRMH}. For instance, the GW and NRMH Markov kernels can be reparameterized as $P_\theta$, where $\theta\in[0,1]$  indicates the degree of non-reversibility. For GW, $\theta=(1-\alpha)$ (Eq. \eqref{eq:GW_mom_swi}) and for NRMH $\theta=\zeta/\zeta_{\text{max}}$ (Eqs. \ref{eq:NRMH_ratio_0} and \ref{eq:Gamma_mat}). While $\theta=0$ coincides with the reversible version of those algorithms, we observed for both algorithms that as $\theta\uparrow 1$, the asymptotic variance usually decreases but the convergence rate of the Markov chains slows down. In other words, non-reversibility does reduce the asymptotic variance but may degrade the speed of convergence in the process, see for instance Fig. \ref{NRMH_circle} for an illustration. This also suggests the existence of an optimal parameter $\theta^\ast\in(0,1)$ controlling simultaneously both convergence aspects.

Comparing marginal and lifted non-reversible schemes is more difficult. However, due to its higher level of symmetry,  the later appears ``less'' irreversible than the former. Intuitively, NRMH is ``more''  irreversible than GW as it imposes one (and only one) privileged direction to the Markov chain, while GW's  changes of privileged direction lead to algorithms that  are ``less'' irreversible. In some sense, the switching parameter compensates the introduction of the irreversible flow. In fact, it is possible to show, using the recent results of \cite{andrieu2019peskun}, that for GW and NRMH kernels using a $\pi$-reversible Markov kernel for proposition, there is an ordering $\text{NRMH}<\text{GW}<\text{MH}$ for the asymptotic variances. In contrast, the lack of symmetry of NRMH which uses an antisymmetric vorticity flow  may lead to a much slower convergence rate than its reversible counterpart if the flow is not adequate for the topology of $\pi$, as for instance in Examples \ref{ex2} and \ref{ex4}. We consequently considered a lifted version of NRMH, referred to as NRMHAV (Section \ref{sec:two_vort}), which combines two opposite vorticity flows. In the spirit of the GW and NRMH analyses carried out at Sections \ref{sec:lifted} and \ref{sec:NRMH}, our work has shown the existence of $\pi$-invariant and non-reversible Markov kernel midway between GW and NRMH which mitigates the NRMH risk of having slow-convergence because $\Gamma$ is not adapted to the topology of $\pi$, while retaining some aspects of the variance reduction feature of NRMH over GW.

This work deals essentially with Markov chains on discrete state spaces. Even though the questions, the concepts and some conclusions have direct equivalent in general state spaces, the relevance of these considerations may be questioned by practical limitations. Indeed, NRMH and lifted kernels are notoriously difficult to construct when $\Scal$ is not finite. In practice, the most popular non-reversible samplers include the discrete time Partially Deterministic Markov Processes (PDMP) \cite{vanetti2017piecewise} such as the discrete time Bouncy Particle Sampler \cite{sherlock2017discrete}. They are also more difficult to study and even though some recent works such as \cite{andrieu2018hypercoercivity} have developed novel tools to analyse them, the research carried out in this paper on simpler algorithm can be regarded as a necessary first step. We leave the study of a possible trade-off between asymptotic variance and convergence rate in function of the irreversibility degree of discrete time PDMP for future research.

Finally, at a more general level, the concept of \textit{irreversibility measure} of a Markov chain deserves to be further developed at a theoretical level. In particular, one can wonder if a (partial) ordering of MCMC algorithms according to their irreversibility measure can be established, in a Peskun ordering style \citep{peskun1973optimumMC} for non-reversible Markov chains.

\paragraph{Acknowledgements} This research work was funded by ENSAE ParisTech and the Insight Center for Data Analytics -- University College Dublin.


\bibliographystyle{elsarticle-harv}
\bibliography{main}


\appendix

\section{Lifted non-reversible Markov chain}

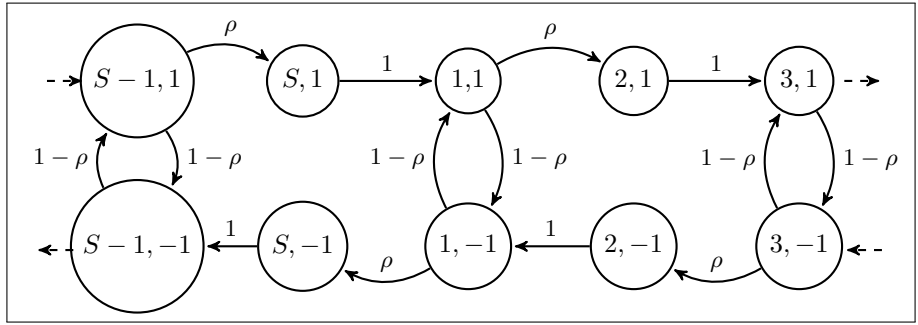
\begin{figure}[H]
\centering
\fbox{
\begin{tikzpicture}[,->,>=stealth',shorten >=1pt,auto,node distance=2.2cm,
                    thick,main node/.style={circle,draw,font=\sffamily\normalfont}]
                    \draw [dashed] (5,0) -- (5.5,0);
                      \draw [dashed] (5.5,-2.25) -- (5,-2.25);
                                          \draw [dashed] (-5.6,0) -- (-5.1,0);
                      \draw [dashed] (-5.25,-2.25) -- (-5.75,-2.25);
  \node[main node] (1) {1,1};
  \node[main node] (2) [below of=1] {$1,-1$};
  \node[main node] (3) [right of=1] {$2,1$};
  \node[main node] (4) [below of=3] {$2,-1$};
  \node[main node] (5) [right of=3] {$3,1$};
  \node[main node] (6) [below of=5] {$3,-1$};
  \node[main node] (7) [left of=1] {$S,1$};
  \node[main node] (8) [below of=7] {$S,-1$};
  \node[main node] (9) [left of=7] {$S-1,1$};
  \node[main node] (10) [below of=9] {$S-1,-1$};

  \path[every node/.style={font=\sffamily\small}]
    (1) edge [bend left] node[above] {$\rho$} (3)
        edge [bend left] node[right] {$1-\rho$} (2)
    (2) edge [bend left] node[above] {$\rho$} (8)
        edge [bend left] node[left] {$1-\rho$} (1)
    (3) edge [ right] node[above] {$1$} (5)
    (5) edge [bend left] node[right] {$1-\rho$} (6)
        (6) edge [bend left] node[left] {$1-\rho$} (5)
        (6) edge [bend left] node[above] {$\rho$} (4)
               (4) edge [left] node[above] {$1$} (2)
(7) edge [left] node[above] {$1$} (1)
(8) edge [left] node[above] {$1$} (10)
(9) edge [bend left] node[above] {$\rho$} (7)
   (9) edge [bend left] node[right] {$1-\rho$} (10)
      (10) edge [bend left] node[left] {$1-\rho$} (9);
\end{tikzpicture}}
\caption{(Example \ref{ex1}) GW Markov chain transition. Each circle corresponds to one state with $(x,\xi)\in\{1,\ldots,S\}\times\{-1,1\}$. The top and bottom row correspond respectively to $\xi=1$, i.e. counter-clockwise inertia and to $\xi=-1$, i.e. clockwise inertia.\label{fig:ex1_chain}}
\end{figure}

\begin{figure}[H]
\centering
\begin{tikzpicture}[->, >=stealth', auto, semithick, node distance=3cm]
\tikzstyle{every state}=[fill=white,draw=black,thick,text=black,scale=1]
\node[state]    (A)                     {$1$};
\node[state]    (B)[left of=A]   {$2$};
\node[state]    (C)[below of=B]   {$3$};
\node[state]    (D)[right of=C]   {$n-1$};
\node[state]    (F)[above right=1.5cm of D]   {$n$};
     \node[draw] at (-1.6,-3.7) {...};
 \draw [dashed] (-2.5,-3.4) -- (-1.9,-3.8);
  \draw [dashed] (-1.9,-3.65) -- (-2.6,-3.2);
   \draw [dashed] (-1.3,-3.65) -- (-0.55,-3.2);
  \draw [dashed] (-0.6,-3.4) -- (-1.3,-3.8);
\path
(A) edge[bend left]         node{$1/2$}         (B)
(A) edge[bend right]         node{$1/2$}         (F)
(B) edge[bend left]         node{$1/4$}           (A)
(A) edge [loop above]  (A)
(B) edge [loop left]  (B)
(C) edge [loop left]  (C)
(B) edge[bend left]         node{$1/2$}           (C)
(C) edge[bend left]         node{$1/3$}           (B)
(D) edge[bend left]         node{$1/2$}     (F)
(D) edge [loop below] (D)
(F) edge [loop right] (F)
(F) edge[bend right]         node[above right]{$1/2n$}     (A)
(F) edge[bend left]        node{$(n-1)/{2n}$}     (D);
\end{tikzpicture}

\begin{tikzpicture}[->, >=stealth', auto, semithick, node distance=3cm]
\tikzstyle{every state}=[fill=white,draw=black,thick,text=black,scale=1]
\node[state]    (A)                     {$(1,1)$};
\node[state]    (B)[left of=A]   {$(2,1)$};
\node[state]    (C)[below of=B]   {$(3,1)$};
\node[state]    (D)[right of=C]   {$(n-1,1)$};
\node[state]    (F)[above right=.75cm and 0.2cm of D]   {$(n,1)$};
\node[state]    (A2)[above right=.75cm and 0.2cm of A]{$(1,-1)$};
\node[state]    (B2)[above left=.75cm and 0.2cm of B]   {$(2,-1)$};
\node[state]    (C2)[below left=.75cm and 0.2cm of C]   {$(3,-1)$};
\node[state]    (D2)[below right=.75cm and 0.2cm of D]   {$(n-1,-1)$};
\node[state]    (F2)[right=1.25cm and 0.75cm of F]   {$(n,-1)$};
\draw [dashed] (-2.6,-3.45) -- (-1.97,-3.7);
\draw [dashed] (-1.4,-3.7) -- (-0.73,-3.5);

    \draw [dashed]  (-2,-5.5)--(-3.5,-5);
    \draw [dashed]  (.5,-5.1)--(-1.43,-5.5);
    \node[draw] at (-1.7,-3.7) {...};
       \node[draw] at (-1.7,-5.5) {...};
\path
(A) edge         node{$1$}         (B)
(B) edge         node{$1$}           (C)
(D) edge         node{$1$}     (F)
(F) edge[bend left]         node{$1/n$}     (A)
(F) edge[bend right]         node[below]{$\frac{n-1}{n}$}     (F2)
(F2) edge[bend left]         node{$\frac{n-1}{n}$}     (D2)
(D2) edge[bend left]         node{$\frac{1}{n-1}$}     (D)
(C2) edge[bend left]         node{$2/3$}     (B2)
(B2) edge[bend left]         node{$1/2$}     (A2)
(A2) edge         node{$1$}     (F2)
(C2) edge[bend right]         node[right]{$\frac{1}{3}$}     (C)
(B2) edge[bend right]         node[left]{$\frac{1}{2}$}     (B)
(F2) edge[bend right]         node[above]{$\frac{1}{n}$}     (F);
\end{tikzpicture}
\caption{(Exemple \ref{ex2}) MH Markov chain (top) and Guided Walk Markov chain (bottom). For the MH chain, the probability to remain in each state is implicit. For the GW chain, the second coordinate of each state indicates the value of the auxiliary variable $\xi$. \label{fig:ex2_1}}
\end{figure}
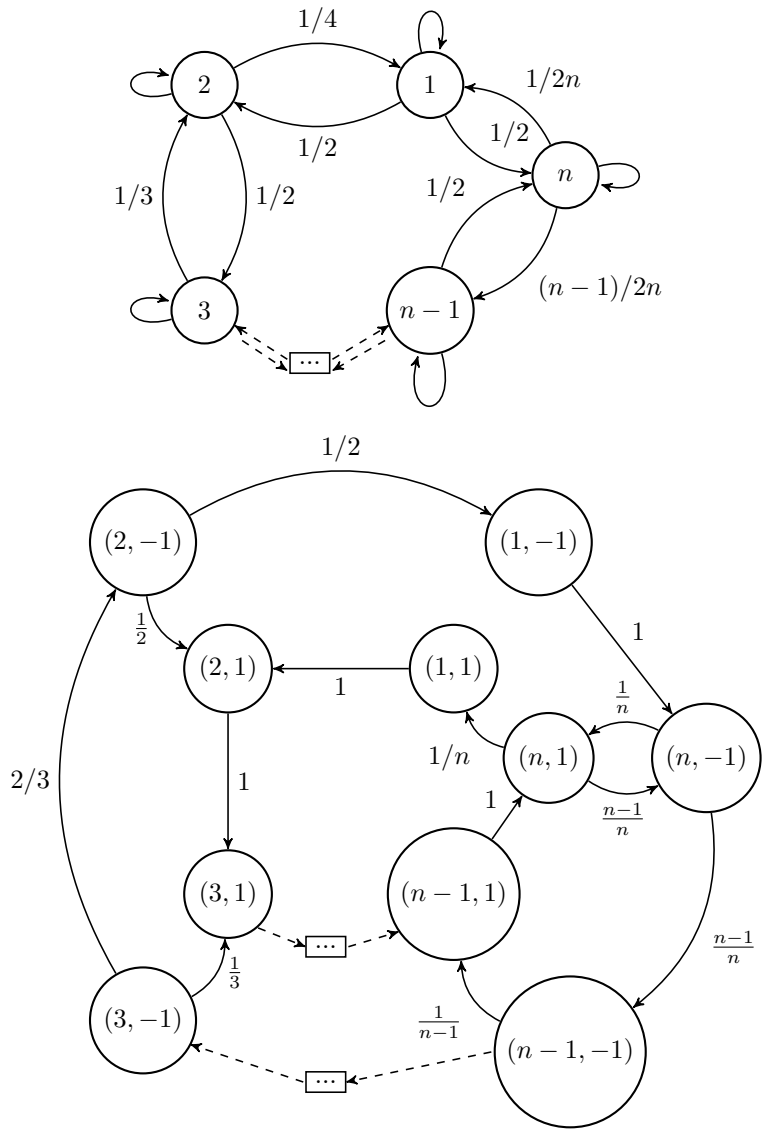

\begin{figure}[H]
\centering
\hspace*{-1.15cm}
\includegraphics[scale=0.5]{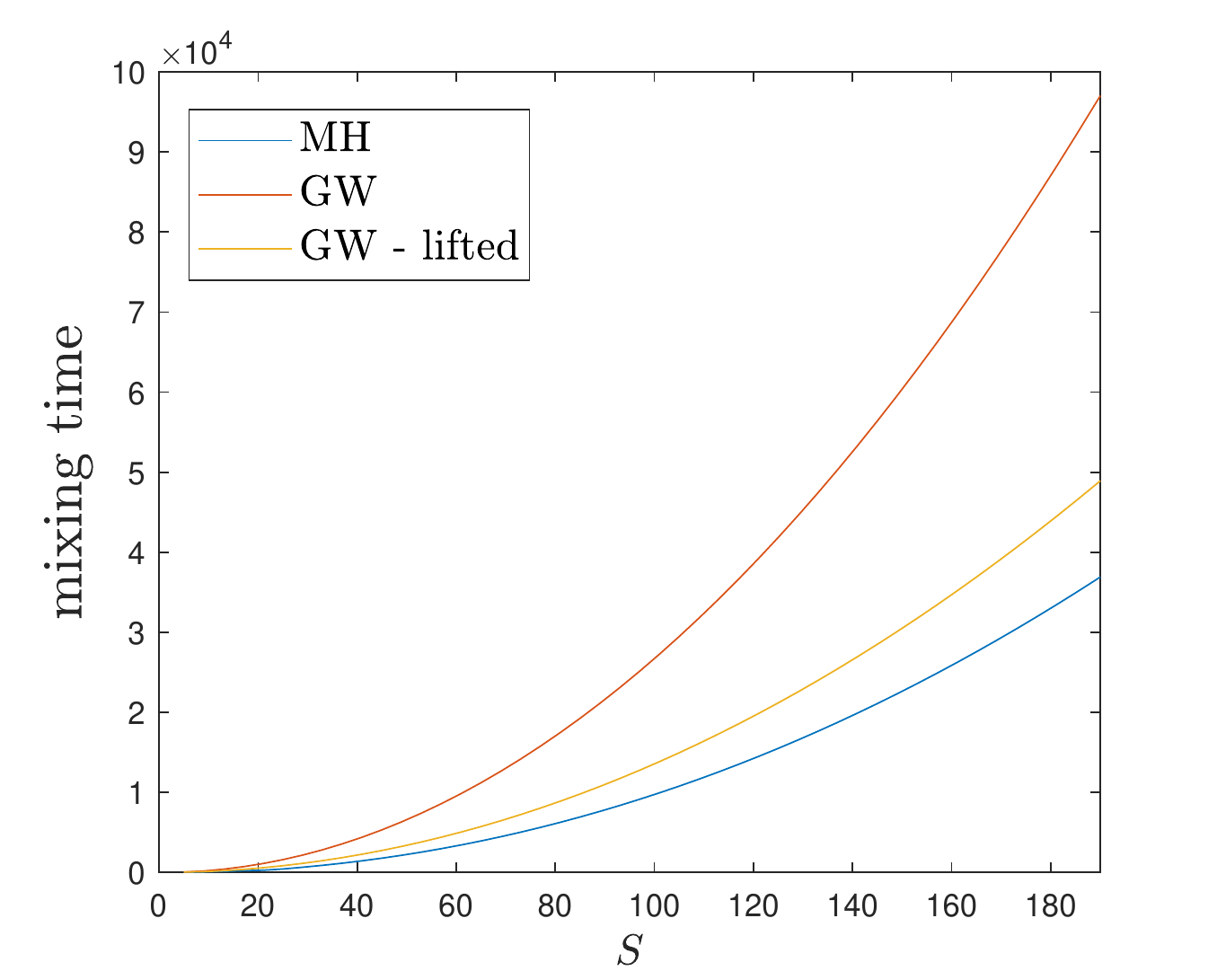}\includegraphics[scale=0.5]{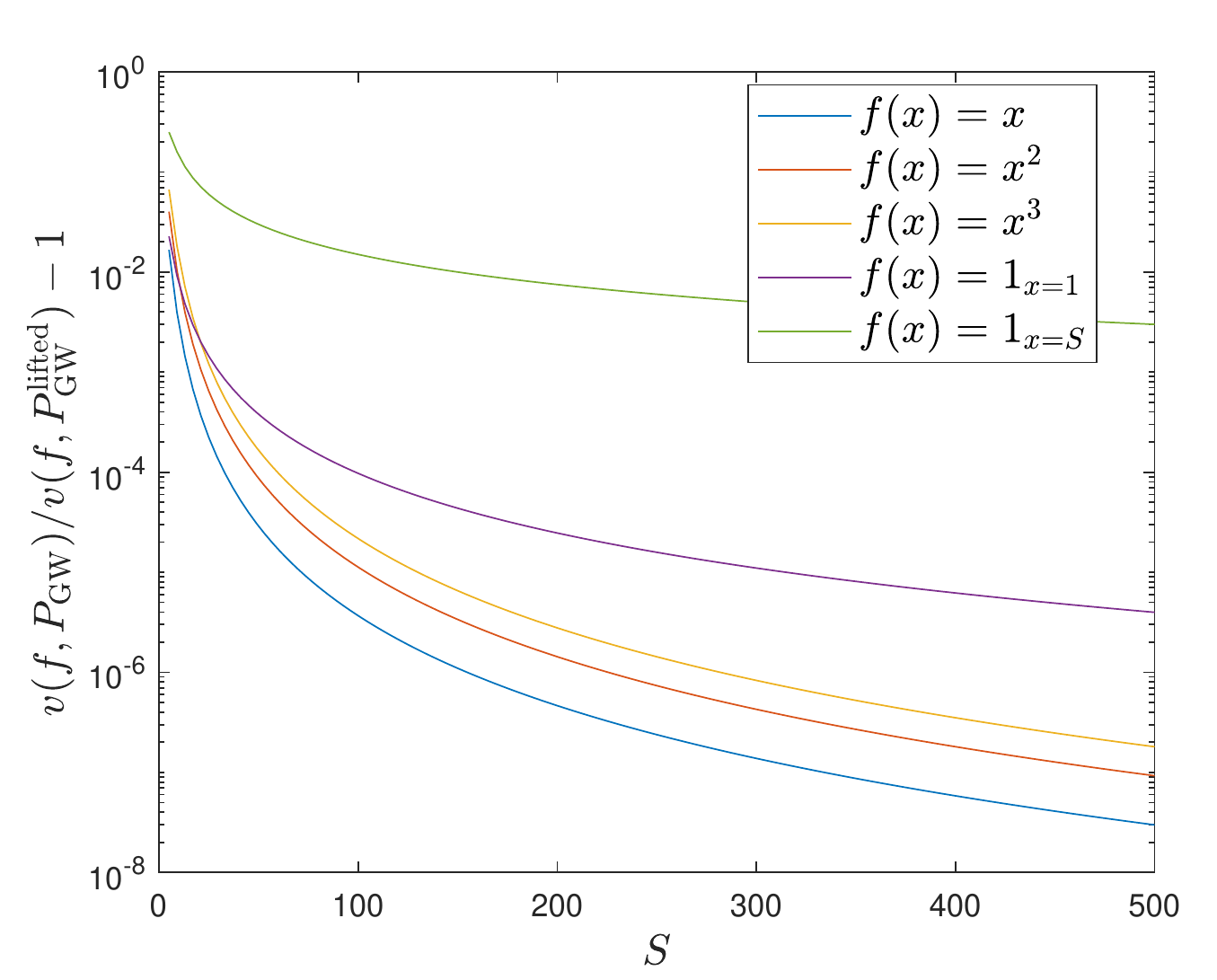}
\caption{(Exemple \ref{ex2}) Left: Comparison of the mixing time for GW and GW-lifted (see \cite{andrieu2019peskun}) in function of $S$. Here the mixing time $\tau$ is defined as $\tau(P):=\inf\{t\in\nset\,:\,\|\delta_1 P^t-\pi\|\leq \eps\}$ with $\epsilon=10^{-5}$.  Right: Comparison of GW and GW-lifted asymptotic variances for some test functions. \label{fig:liftedGW}}
\end{figure}

\section{Marginal non-reversible Markov chain}
\begin{figure}[H]

\centering
\hspace*{-0cm}
\includegraphics[scale=.6]{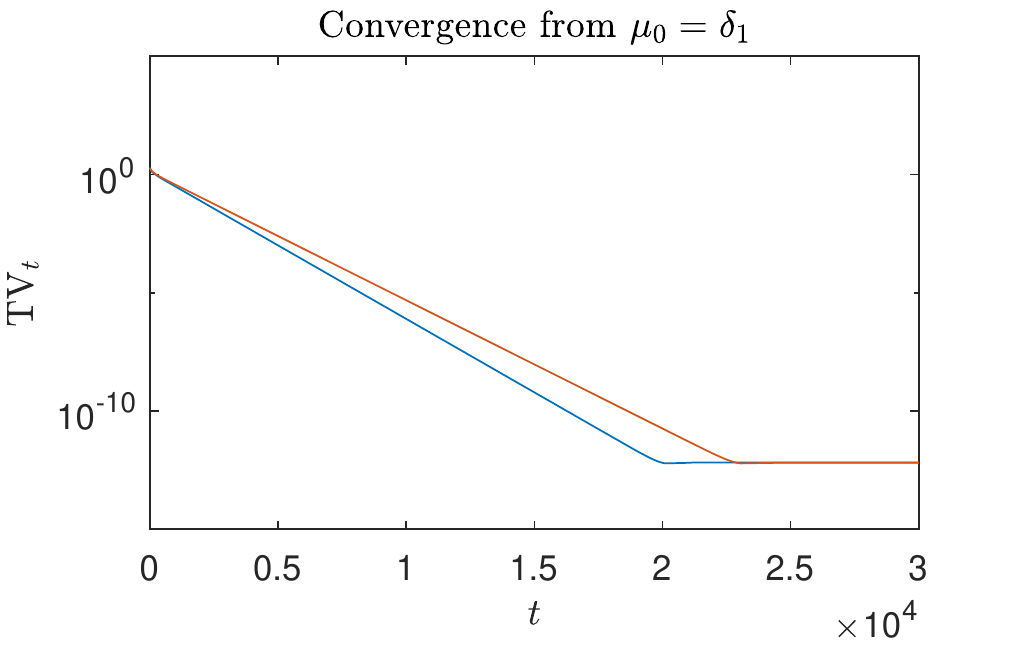}\includegraphics[scale=.6]{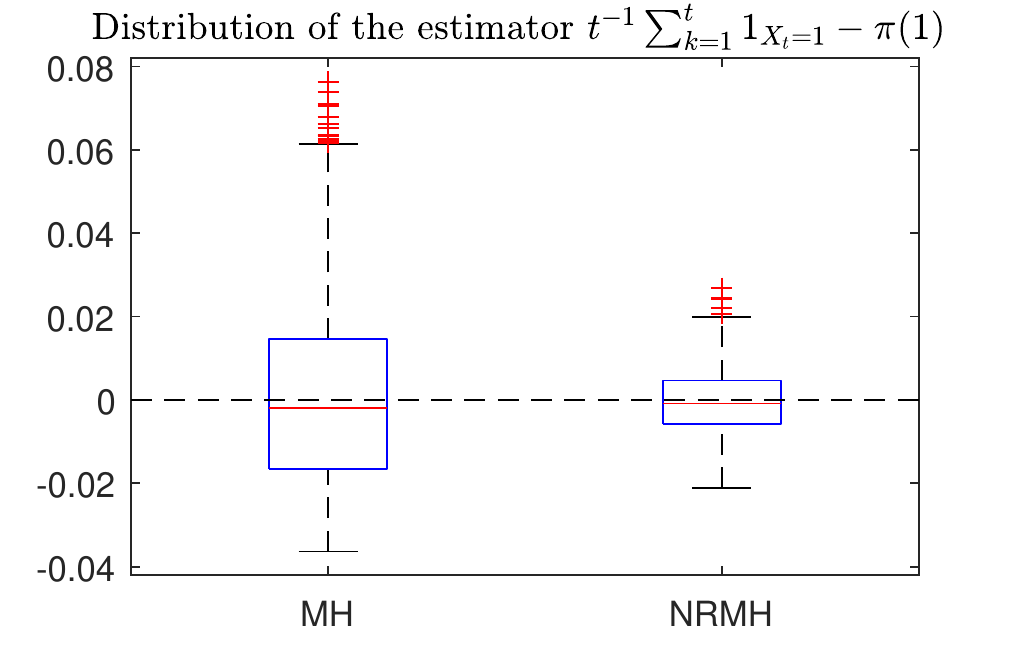}
\vspace{0.2cm}
\hspace*{-0cm}\includegraphics[scale=.65]{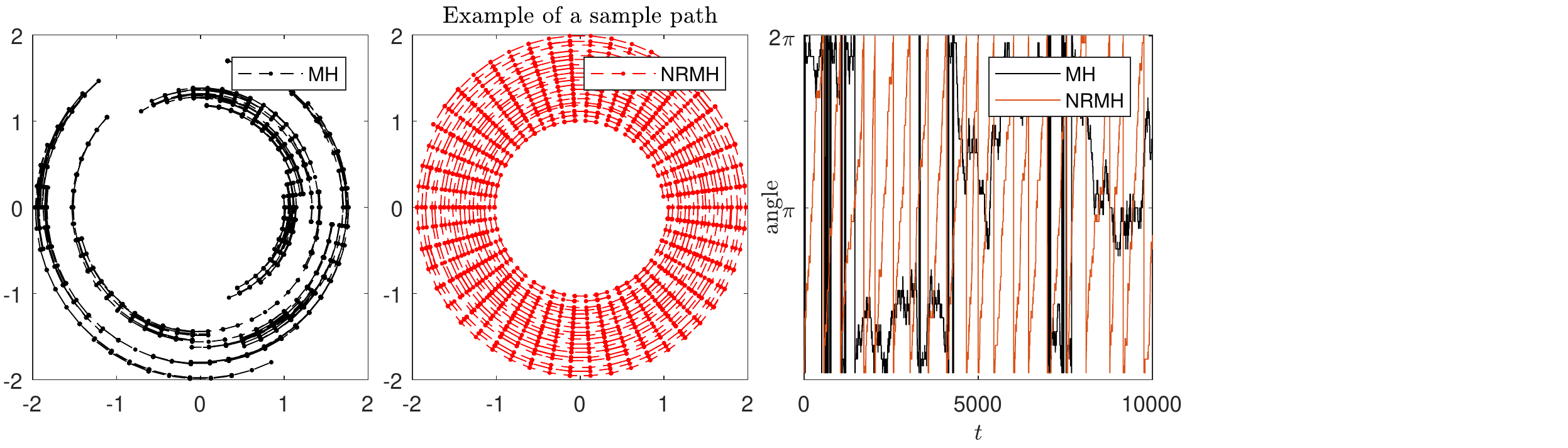}

\vspace{-.5cm}
\caption{(Example \ref{ex1}) Illustration of MH (Alg. \ref{algo_MH}) and non-reversible MH (Alg. \ref{algo_NRMH}) with $S=50$, $\rho=0.1$ and $\zeta=\zeta_{\text{max}}$. \textbf{First row} -- Convergence in total variation: the blue plot is MH and the red one is NRMH. The distribution of the Monte Carlo estimate was obtained using $L=1,000$ independent chains starting from $\pi$ and length $T=10,000$ for both algorithms. Other test functions of the type $f:x\mapsto\1_{x=i}$ for $i\in\Scal$ gave similar results. \textbf{Second row} -- Illustration of a particular sample path of length $T=10,000$ for both Markov chains. For better visibility of the two sample paths, the left and centre plots represent the function $\{(1+t/T)\cos(2\pi X_t/p),(1+t/T)\sin(2\pi X_t/p)\}$ for $t=1,\ldots,T$ for the MH and NRMH Markov chains, respectively. This shows that NRMH does explore the circle more efficiently.\label{fig:1}}
\end{figure}

\section{Proof of Prop. \ref{prop_geo}}
We first need to proof the following Lemma.
\begin{lemma}
\label{lem:ex2}
  The conductance of the MH Markov chain of Example \ref{ex2} satisfies
  \begin{equation}\label{lem_ex2}
    \frac{1+\sqrt{1+2S(S+1)}}{S(S+1)}\leq h(P)\leq \frac{2}{S+1}\,.
  \end{equation}
\end{lemma}
\begin{proof}
Let for all $A\in\Salg$, $\psi(A):={\sum_{x\in A}\pi(x)P(x,\bar{A})}\slash{\pi(A)\wedge (1-\pi(A))}$ be the quantity to minimize.
A close analysis of the MH Markov chain displayed at the top panel of Fig. \ref{fig:ex2_1} shows that the set $A$ which minimizes $\psi(A)$ has the form $A=(a_1,a_1+1,\ldots,a_2)$ for some $S\geq a_2\geq a_1\geq 1$. Indeed, since the Markov chain moves to neighbouring states only there are only two ways to exit $A$ for each transition. Since each way to exit $A$ contributes at the same order of magnitude to the numerator, taking contiguous states minimizes it and in particular
$$
\sum_{x\in A}\pi(x)P(x,\bar{A})=\pi(a_1)\frac{1\vee (a_1-1)}{2a_1}+\pi(a_2)\frac{1}{2}=\frac{1\wedge (a_1-1)+a_2}{S(S+1)}\,,
$$
so that for any $a_1<a_2$ satisfying $\pi(A)<1/2$, we have:
\begin{equation}\label{eq:lem2:2}
\psi(A)\geq \frac{1\vee (a_1-1)+a_2}{a_2(a_2+1)-a_1(a_1-1)}
\end{equation}
since
$$
\pi(A)=\frac{2}{S(S+1)}\sum_{k=a_1}^{a_2}k=\frac{a_2(a_2+1)-a_1(a_1-1)}{S(S+1)}\,.
$$
Fix $a_1$ and treat $a_2$ as a function of $a_1$ satisfying $\pi(A)<1/2$. On the one hand, note that for all $a_1$ the function mapping $a_2$ to the RHS of Eq. \eqref{eq:lem2:2} is decreasing. On the other hand, we have that  $\pi(A)<1/2\Leftrightarrow {a_2^\ast(a_2^\ast+1)-a_1(a_1-1)}<S(S+1)/2$, which yields
$$
a_2\leq a_2^\ast(a_1):=\left\lfloor\frac{-1+V(a_1,S)}{2}\right\rfloor\,,\quad V(a_1,S):=\sqrt{1+2S(S+1)+4a_1(a_1-1)}\,.
$$
Hence, for all $a_1$, the RHS of Eq. \eqref{eq:lem2:2} is lower bounded by
\begin{multline*}
\frac{4(1\vee (a_1-1))-2+2V(a_1,S)}{(-1+V(a_1,S))(1+V(a_1,S))-4a_1(a_1-1)}
=\frac{4(1\vee (a_1-1))-2+2V(a_1,S)}{V(a_1,S)^2-1-4a_1(a_1-1)}\\
=\frac{2(1\vee (a_1-1))-1+V(a_1,S)}{S(S+1)}\,.
\end{multline*}
Clearly, the numerator is an increasing function of $a_1$ and is thus minimized for $a_1=1$, which gives
the lower bound of Eq. \eqref{lem_ex2}. Finally, by definition $h(P)$ is upper bounded by $\psi(A)$ for any $A\in\Salg$ satisfying $\pi(A)<1/2$. In particular, taking $A=(1,2,\ldots,(S-1)/2)$ gives the upper bound of Eq. \eqref{lem_ex2}.
\end{proof}

\begin{proof}
Since $P_{\text{MH}}$ is reversible and aperiodic its spectrum is real with any eigenvalue different to one $\lambda\in\Lambda_{|\bfo^\perp}:=\text{Sp}(P_{\text{MH}})\backslash\{1\}$ satisfying $-1<\lambda<1$. The norm of $P_{\text{MH}}$ as an operator on the non-constant functions of $\Ltwo(\pi)$ is $\gamma:=\max\{\sup\Lambda_{|\bfo^\perp},|\inf \Lambda_{|\bfo^\perp}|\}$. It is well known (see \eg \cite{yuen2000applications}) that
$$
\|\delta_{1}P_{\text{MH}}^t-\pi\|_2\leq \|\delta_1-\pi\|_2\gamma^t\,.
$$
It can be readily checked that $\|\delta_1-\pi\|_2$ corresponds to the first factor on the RHS of Eq. \eqref{eq:ex2:MH}. The tedious part of the proof is to bound $\gamma$. Using again the reversibility, the Cheeger's inequality, (see \eg \cite{diaconis1991geometric} for a proof), writes
\begin{equation}\label{eq:chee}
1-2h(P)\leq \sup\Lambda\leq 1-h(P)^2\,,
\end{equation}
where $h(P)$ is the Markov chain conductance defined as
$$
h(P)=\inf_{\substack{A\in\Salg\,\\ \,\pi(A)<1/2}} \frac{\sum_{x\in A}\pi(x)P(x,\bar{A})}{\pi(A)}\,.
$$
Combining Cheeger's inequality and Lemma \ref{lem:ex2} yields
\begin{equation}\label{eq:sup_eig}
\sup\Lambda\leq 1-\frac{2}{S(S+1)}\,.
\end{equation}
However, to use the above bound to upper bound $\gamma$, we need to check that $\sup\Lambda_{|\bfo^\perp}\geq |\inf \Lambda_{|\bfo^\perp}|$. In general, bounding $|\inf \Lambda_{|\bfo^\perp}|$  proves to be more challenging than $\sup \Lambda_{|\bfo^\perp}$. However, in the context of this example, we  can use the bound derived in Proposition 2 of \cite{diaconis1991geometric}. It is based on a geometric interpretation of the Markov chain as a non bipartite graph with vertices (states) connected by edges (transitions), as illustrated in Fig. \ref{fig:ex2_1}. More precisely, the main result of this work to our interest states that
\begin{equation}\label{eq:inf_sp}
\inf \Lambda_{|\bfo^\perp}\geq -1+\frac{2}{\iota(P)}\,,
\end{equation}
with $\iota(P)=\max_{e_{a,b}\in\Gamma}\sum_{\sigma_x\ni e_{a,b}}|\sigma_x|\pi(x)$, where
\begin{itemize}
\item $e_{a,b}$ is the edge corresponding to the transition from state $a$ to $b$,
  \item $\sigma_x$ is a path of odd length going from state $x$ to itself, including a self-loop provided that $P(x,x)>0$, and more generally
  $\sigma_x=(e_{x,a_1},e_{a_1,a_2},\ldots,e_{a_{\ell},a_x})$ with $\ell$ even.
  \item $\Gamma$ is a collection of paths $\{\sigma_1,\ldots,\sigma_S\}$ including exactly one path for each state,
    \item $|\sigma_x|$ represents the ``length'' of path $\sigma_x$ and is formally defined as
    $$
    |\sigma_x|=\sum_{e_{a,b}\in\sigma_x}\frac{1}{\pi(a) P(a,b)}\,.
    $$
\end{itemize}
Let us consider the collection of paths $\Gamma$ consisting of all the self loops for all states $x\geq 2$. It can be readily checked that the length of such paths is
$$
|\sigma_{x}|=(\pi(x)P(x,x))^{-1}=\left(\frac{x}{\Delta}\frac{1}{2x}\right)^{-1}=S(S+1)\,.
$$
For state $x=1$, let us consider the path consisting of the walk around the circle $\sigma_1:(e_{1,2},e_{2,3},\ldots,e_{S,1})$. It may have been possible to take the path $e_{1,2},e_{2,2},e_{2,1}$, but it is unclear if paths using the same edge twice are permitted in the framework of Prop. 2 of \cite{diaconis1991geometric}. The length of path $\sigma_1$ is
\begin{multline*}
|\sigma_{1}|=\frac{1}{\pi(1)P(1,2)}+\cdots+\frac{1}{\pi(S)P(S,1)}\\
=S(S+1)+\frac{S(S+1)}{2}+\cdots
\frac{S(S+1)}{S-1}+S(S+1)=S(S+1)\left(1+\sum_{k=1}^S\frac{1}{k}\right)\,.
\end{multline*}
We are now in a position to calculate $\iota(P)$. First note that, by construction, each edge belonging to any path $\sigma_k$ contained in $\Gamma$ appears once and only once. Hence, the constant $\iota(P)$ simplifies to the maximum of the set $\{|\sigma_x|\pi(x)\,,\sigma_x\in\Gamma\}$ that is
\begin{equation}\label{eq:iotaa}
\max\left\{2\left(1+\sum_{\ell=1}^S\frac{1}{\ell}\right),2k\,:\,2\leq k\leq S\right\}
=2S\,,
\end{equation}
since on the one hand $\sum_{\ell=1}^S {1}\slash{\ell}\leq 1+\log(S)$ and on the other hand $S\geq 5$. Combining Eqs. \eqref{eq:inf_sp} and \eqref{eq:iotaa} yields to
\begin{equation}\label{eq:lower_bound_sp}
\inf \Lambda_{|\bfo^\perp}\geq -1+\frac{1}{S}\,.
\end{equation}
It comes that if $\inf \Lambda_{|\bfo^\perp}\geq 0$, then $\gamma\leq \sup \Lambda_{|\bfo^\perp}$ and otherwise we have
$$
0>\inf \Lambda_{|\bfo^\perp}\geq-1+\frac{1}{S}\Leftrightarrow 0>\inf \Lambda_{|\bfo^\perp}\;\text{and}\;
\left|\inf \Lambda_{|\bfo^\perp}\right|\leq1-\frac{1}{S}\,,
$$
which combines with Eq. \eqref{eq:sup_eig} to complete the proof as
$$
\max\{\sup\Lambda_{|\bfo^\perp},|\inf \Lambda_{|\bfo^\perp}|\}\leq 1-\frac{1}{S}\vee 1-\frac{2}{S(S+1)}
\leq 1-\frac{2}{S(S+1)}\,,
$$
since $S\geq 5$.
%

\end{proof}

\section{Proof of Prop. \ref{prop:comp_CV}}

\begin{proof}
By straightforward calculation we have:
\begin{equation}\label{eq:boundGW}
\left\|\delta_1 P_{\text{GW}}^{(S-1)}(\,\cdot\,\times \{-1,1\})-\pi\right\|_2^2=1-\frac{8}{3S}+o(1/S)\,.
\end{equation}
Using Proposition \ref{prop_geo}, we have that
\begin{multline*}
\left\|\delta_1 P_{\text{MH}}^{(S-1)}-\pi\right\|_2^2\leq
\left\{1-\frac{4}{S(S+1)}+\frac{2(2S+1)}{3S(S+1)}\right\}
    \left(1-\frac{2}{S(S+1)}\right)^{S-1}\\
 =1-\frac{2}{3S}+o(1/S)\,.
\end{multline*}
Comparing the complexity of the former bound with Eq. \eqref{eq:boundGW}, the inequality of Eq. \eqref{eq:prop:compL2} cannot be concluded. In fact, we need to refine the bound for the MH convergence. Analysing the proof of Lemma \ref{lem:ex2}, the lower bound of the conductance seems rather tight as resulting from taking the real bound on $a_2^\ast(a_1)$ as opposed to the floor of it. To illustrate this statement, the value of the bound is compared to the actual conductance for some moderate size of $S$, the calculation being otherwise too costly. Then, we calculated the numerical value of $\sup\Lambda_{|1^\perp}$ for $S\leq 500$ and compared with the lower bound derived from Cheeger's inequality in the proof of Prop. \ref{prop_geo}. It appears that the Cheeger's bound is in this example too lose to justify Eq. \eqref{eq:prop:compL2}. However, taking a finer lower bound such as
$$
\sup\Lambda_{|1^\perp}\leq 1- 8/S^2\,,
$$
yields
$$
\left\|\delta_1 P_{\text{MH}}^{(S-1)}-\pi\right\|_2^2\leq1-\frac{20}{3S}+o(1/S)
$$
which concludes the proof.
\begin{figure}
\label{fig:TV_conv}
\centering
\includegraphics[scale=0.6]{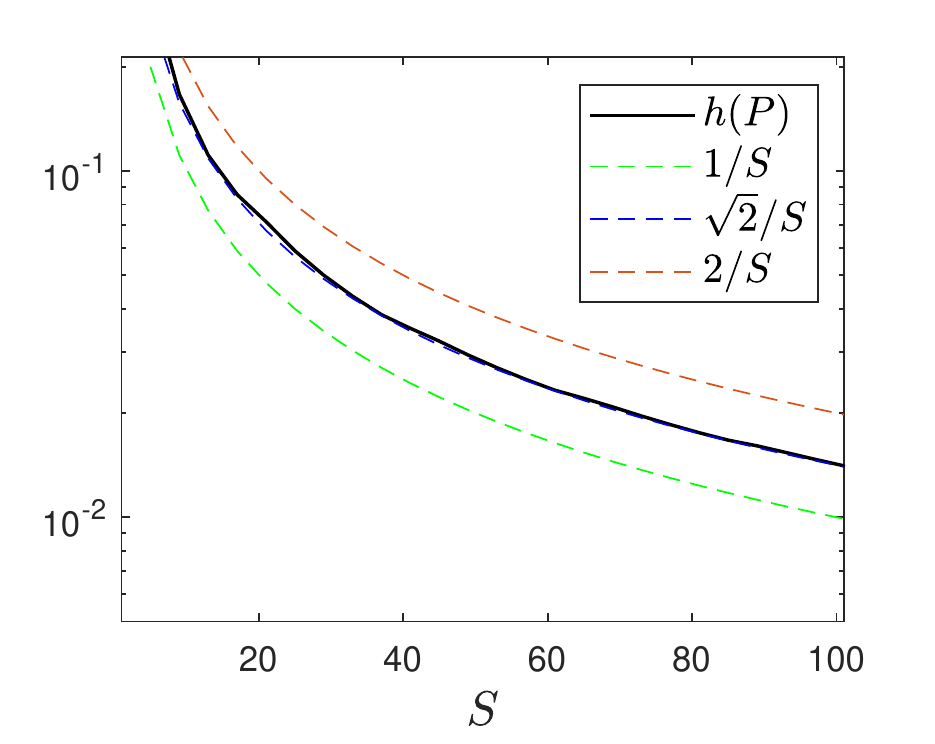}
\includegraphics[scale=0.6]{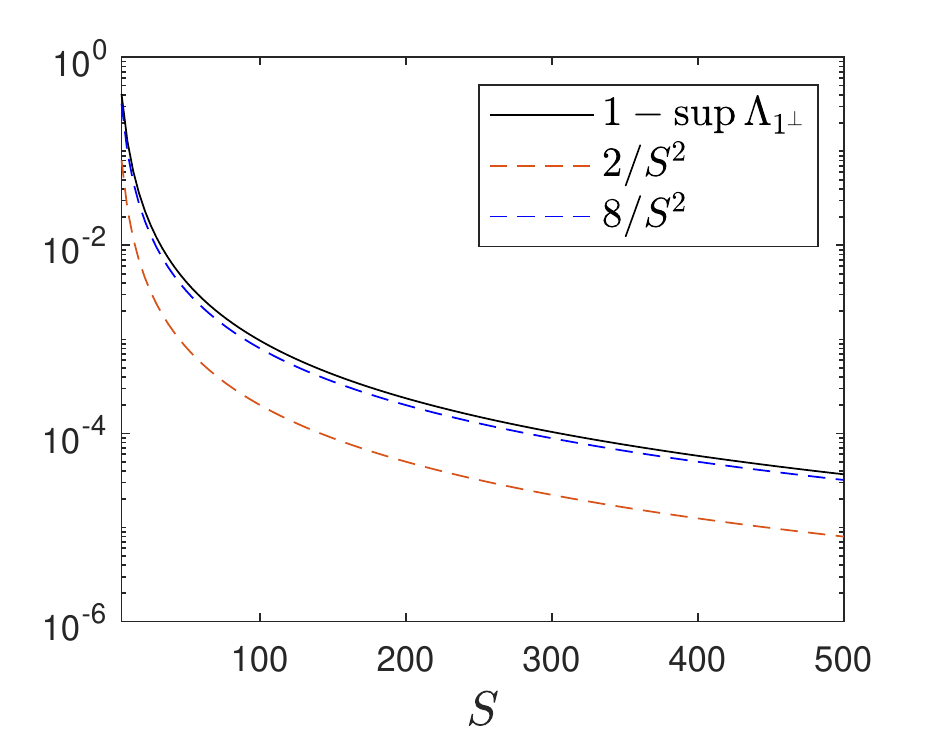}

\caption{(Example \ref{ex2}): Conductance $h(P)$ and different approximations, including the lower bound $\sqrt{2}/S$ derived in Lemma \ref{lem:ex2} (left) and comparison of $\sup\Lambda_{|1^\perp}$ with the upper bound derived in Proposition \ref{prop_geo} $2/S^2$ and an estimated finer upper bound $8/S^2$. For readability, we have represented one minus these quantities and in log scale, hence upper bounds become lower bounds.}
\end{figure}
\end{proof}

\section{Proof of Proposition \ref{prop:sec5:peskun}}
\label{proof:sec5:peskun}
\begin{proof}
First, denote by $R$ the mixture of the two NRMH kernels with weight $1/2$. We start by showing that this kernel is $\pi$-reversible. Indeed, the subkernel of $R$ satisfies:
\begin{multline*}
  \pi(\rmd x) Q(x,\rmd y)(A_{\Gamma}(x,y)+A_{-\Gamma}(x,y)) \\
  = \rmd x\rmd y \big(\left[\pi(x)Q(x,y)\wedge\pi(y)Q(y,x)+\Gamma(x,y)\right]+\\
  \left[\pi(x)Q(x,y)\wedge\pi(y)Q(y,x)-\Gamma(x,y)\right]\big)\\
  =\rmd x\rmd y \big(\left[\pi(x)Q(x,y)-\Gamma(x,y)\wedge\pi(y)Q(y,x)\right]+\\
  \left[\pi(x)Q(x,y)+\Gamma(x,y)\wedge\pi(y)Q(y,x)\right]\big)\\
  =\pi(\rmd y)Q(y,\rmd x) \left(\left[\frac{\pi(x)Q(x,y)-\Gamma(x,y)}{\pi(y)Q(y,x)}\wedge 1\right]+
  \left[\frac{\pi(x)Q(x,y)+\Gamma(x,y)}{\pi(y)Q(y,x)}\wedge1\right]\right)\,.
\end{multline*}
Now, note that for all $x\in\Scal$ and all $A\in\Salg$,
$$
R(x,A\backslash\{x\})=\frac{1}{2}\int_{A\backslash\{x\}}Q(x,\rmd z)(A_{\Gamma}(x,z)+A_{-\Gamma}(x,z))
$$
and since for any two positive number $a$ and $b$, $(1\wedge a)+(1\wedge b)\leq 2\wedge (a+b)$, we have all $(x,z)\in\Scal^2$,
\begin{multline*}
  A_{\Gamma}(x,z)+A_{-\Gamma}(x,z)=1\wedge \frac{\pi(y)Q(y,x)+\Gamma(x,y)}{\pi(x)Q(x,y)}+
1\wedge \frac{\pi(y)Q(y,x)-\Gamma(x,y)}{\pi(x)Q(x,y)}\\
\leq 2\left(1\wedge\frac{\pi(y)Q(y,x)}{\pi(x)Q(x,y)}\right)
\end{multline*}
since by Assumption \ref{assumption2}, $\pi(y)Q(y,x)+\Gamma(x,y)\geq 0$ for all $(x,y)\in\Scal^2$. This yields a Peskun-Tierney ordering $R\prec P_{\text{MH}}$, since
$$
R(x,A\backslash\{x\})\leq \frac{1}{2}\int Q(x,\rmd z)\left(1\wedge\frac{\pi(y)Q(y,x)}{\pi(x)Q(x,y)}\right)=P_{\text{MH}}(x,A\backslash\{x\})
$$
and the proof is concluded by applying Theorem 4 of \cite{tierney1998note}.
\end{proof}

\section{Proof of Proposition \ref{prop:sec5}}
\label{proof_prop5}
\begin{proof}
Note that if $\Gamma_{1}$ satisfies Assumptions \ref{assumption1} and \ref{assumption2} then
\begin{multline*}
  \pi(x)Q(x,y)\wedge \left(\pi(y)Q(y,x)+\Gamma_1(x,y)\right)\\=
  \Gamma_1(x,y)+\left[\pi(y)Q(y,x)\wedge \left(\pi(x)Q(x,y)+\Gamma_{1}(y,x)\right)\right]\,.
\end{multline*}
Thus, if  $\Gamma_{1}$ and $\Gamma_{-1}$ satisfy Assumptions \ref{assumption1}, \ref{assumption2} and \ref{assumption3} then
\begin{multline*}
  \Gamma_1(x,y)=\left[\pi(y)Q(y,x)\wedge \left(\pi(x)Q(x,y)+\Gamma_{-1}(y,x)\right)\right] \\
  -\left[\pi(y)Q(y,x)\wedge \left(\pi(x)Q(x,y)+\Gamma_{1}(y,x)\right)\right]\,.
\end{multline*}
Hence, we have
\begin{multline*}
  \Gamma_1(y,x)=\left[\pi(x)Q(x,y)\wedge \left(\pi(y)Q(y,x)+\Gamma_{-1}(x,y)\right)\right] \\
  -\left[\pi(x)Q(x,y)\wedge \left(\pi(y)Q(y,x)+\Gamma_{1}(x,y)\right)\right]\,,\\
  =\left[\left(\pi(x)Q(x,y)-\Gamma_{-1}(x,y)\right)\wedge \pi(y)Q(y,x)\right] \\
  -\left[\left(\pi(x)Q(x,y)-\Gamma_{1}(x,y)\right)\wedge \pi(y)Q(y,x)\right]+\Gamma_{-1}(x,y)-\Gamma_{1}(x,y)\\
    =\Gamma_{1}(x,y)+\Gamma_{-1}(x,y)-\Gamma_{1}(x,y)\,,
\end{multline*}
and thus $\Gamma_{-1}=-\Gamma_1$, which replacing in Eq. \eqref{eq:sdbe} leads to
\begin{multline}
\label{eq:proof:sec5}
 \pi(x)Q(x,y)\wedge \left(\pi(y)Q(y,x)+\Gamma_1(x,y)\right)\\
  =\left(\pi(x)Q(x,y)+\Gamma_1(x,y)\right) \wedge \pi(y)Q(y,x)
\end{multline}
for all $(x,y)\in\Scal^2$. Conversely, it can be readily checked that if $\Gamma_1$ satisfies Assumptions \ref{assumption1}, \ref{assumption2} and Eq. \eqref{eq:proof:sec5}, then setting $\Gamma_{-1}=-\Gamma_1$ implies that $\Gamma_1$ and $\Gamma_{-1}$ satisfy Assumptions \ref{assumption1}, \ref{assumption2} and the skew-detailed balance equation (Eq. \eqref{eq:sdbe}). The proof is concluded by noting that Eq. \eqref{eq:proof:sec5} holds if and only if $\Gamma_1$ is the null operator on $\Scal\times\Scal$ or $Q$ is $\pi$-reversible.
\end{proof}

\section{Proof of Proposition \label{proof_NRMHAV}}

We prove Proposition \ref{thm_NRMHAV} that states that the transition kernel \eqref{NRMHAV_kernel} of the Markov chain generated by Algorithm \ref{algo_NRMHAV} is $\tilde{\pi}$-invariant and is non-reversible if and only if $\Gamma = 0$.

\begin{proof}To prove the invariance of $K_\rho$, we need to prove that
$$
\sum_{y\in\Scal,\eta\in\{-1,1\}} \tilde{\pi}(y,\eta)K_\rho(y,\eta;x,\xi) = \tilde{\pi}(x,\xi)\,,
$$
for all $(x,\xi)\in\Scal\times\{-1,1\}$ and $\rho\in[0,1]$.
\begin{eqnarray}
\label{eq:proof0}
\sum_{y,\eta} \hspace{-0.6cm}&&\tilde{\pi}(y,\eta)K_\rho(y,\eta;x,\xi)= \sum_y \tilde{\pi}(y,\xi)K_\rho(y,\xi;x,\xi) + \sum_y \tilde{\pi}(y,-\xi)K_\rho(y,-\xi;x,\xi) \nonumber\\
&& = \tilde{\pi}(x,\xi)K_\rho(x,\xi;x,\xi)    + \sum_{y \neq x} \tilde{\pi}(y,\xi)K_\rho(y,\xi;x,\xi)+ \tilde{\pi}(x,{-\xi})K_\rho(x,{-\xi};x,\xi)   \nonumber\\
&& = \tilde{\pi}(x,\xi) \bigg\{ Q(x,x) + (1-\rho)
 \sum_z Q(x,z)(1-A_{\xi\Gamma}(x,z)) \nonumber\\
&& + \rho\sum_z Q(x,z)(1-A_{-\xi\Gamma}(x,z)) \bigg\}
  + \sum_{y \neq x} \tilde{\pi}(y,\xi)Q(y,x)A_{\xi\Gamma}(y,x)
\end{eqnarray}
the second equality coming from the fact that $K_\rho(y,{-\xi};x,\xi) \neq 0$ if and only if $x=y$ and the third from the fact that $\tilde{\pi}(x,\xi) = \tilde{\pi}(x,-\xi) = \pi(x)/2$. Now, let $A(x,\xi) := \sum_{y \neq x} \tilde{\pi}(y,\xi)Q(y,x)A_{\xi\Gamma}(y,x)$ and note that:
\begin{eqnarray}
\label{eq:proof1}
A(x,\xi)
\hspace{-.6cm}&&= (1/2)\sum_{\substack{y \neq x \\ \pi(y)Q(y,x)>0}} \pi(y)Q(y,x)  \wedge\left\{ {\xi\Gamma(y,x) + \pi(x)Q(x,y)} \right\} \nonumber\\
&& = (1/2) \sum_{\substack{y \neq x \\ \pi(y)Q(y,x)>0}} \left\{\pi(y)Q(y,x)-\xi\Gamma(y,x)\right\}  \wedge \pi(x)Q(x,y)\nonumber\\
&& +(\xi/2) \sum_{\substack{y \neq x \\ \pi(y)Q(y,x)>0}}\Gamma(y,x)\,,\nonumber\\
&& = (1/2) \sum_{\substack{y \neq x \\ \pi(y)Q(y,x)>0}} \pi(x)Q(x,y) A_{\xi\Gamma}(x,y)  \nonumber\\
&&+(\xi/2) \sum_{\substack{y \neq x}}\Gamma(y,x)\1_{\pi(y)Q(y,x)>0}\,,
\end{eqnarray}
Assumption \ref{assumption2} together with the fact that $\pi(x)>0$ for all $x\in\Scal$ yields $\pi(y)Q(y,x)>0$ if and only if $\pi(x)Q(x,y)>0$. It can also be noted that the lower-bound condition on $\Gamma$ implies that $\Gamma(x,y)=0$ if $Q(x,y)=0$. This leads to
\begin{multline}
\label{eq:proof1_2}
A(x,\xi) = (1/2) \sum_{\substack{y \neq x \\ \pi(x)Q(x,y)>0}} \pi(x)Q(x,y) A_{\xi\Gamma}(x,y)+(\xi/2) \sum_{\substack{y \neq x}}\Gamma(y,x)\\
=\tilde{\pi}(x,\xi) \sum_{y \neq x} Q(x,y)A_{\xi\Gamma}(x,y)
\end{multline}
since for all $x\in\Scal$, $\sum_{y\in\Scal}\Gamma(x,y)=0$.  Similarly, define
$$
B(x,\xi):= \tilde{\pi}(x,\xi) \sum_z Q(x,z) \left\{(1-\rho)(1-A_{\xi\Gamma}(x,z)) + \rho(1-A_{-\xi\Gamma}(x,z)) \right\}\,.
$$
Using Lemma \ref{lemma1}, we have:
\begin{eqnarray}
\label{eq:proof2}
 B(x,\xi) \hspace{-0.6cm}&&= \tilde{\pi}(x,\xi) \sum_{z\in\Scal} Q(x,z)(1-A_{\xi\Gamma}(x,z))\\
  &&= \tilde{\pi}(x,\xi) \sum_{z \neq x} Q(x,z)(1-A_{\xi\Gamma}(x,z))\,, \nonumber\\
  &&= \tilde{\pi}(x,\xi) \sum_{z \neq x} Q(x,z)-A(x,\xi)\,,
\end{eqnarray}
where the penultimate equality follows from $A_{\Gamma}(x,x) = 1$ for all $x\in\Scal$. 
Finally, combining Eqs. \eqref{eq:proof0} and \eqref{eq:proof2}, we obtain:
\begin{eqnarray*}
\sum_{y,\eta}\tilde{\pi}(y,\eta)K_\rho(y,\eta;x,\xi)\hspace{-0.6cm}&&=\tilde{\pi}(x,\xi) Q(x,x)A_{\xi\Gamma}(x,x)+B(x,\xi)+A(x,\xi)\,,\nonumber\\
&&=\tilde{\pi}(x,\xi) Q(x,x)+\tilde{\pi}(x,\xi) \sum_{z \neq x} Q(x,z)\,,\nonumber\\
&&=\tilde{\pi}(x,\xi) \,,
\end{eqnarray*}
since $\sum_{y\in\Scal}Q(x,y)=1$, for all $x\in\Scal$.
We now study the $\tpi$-reversibility of $K_\rho$, \ie conditions on $\Gamma^\xi$ such that for all $(x,y)\in\Scal^2$ and $(\xi,\eta)\in\{-1,1\}^2$ such that $(x,\xi)\neq (y,\eta)$, we have:
\begin{equation}
\label{eq:proof3}
\tpi(x,\xi)K_\rho(x,\xi;y,\eta)=\tpi(y,\eta)K_\rho(y,\eta;x,\xi)\,.
\end{equation}
First note that if $x=y$ and $\xi=-\eta$, then Eq. \eqref{eq:proof3} is equivalent to
$$
\sum_{z\in\Scal}Q(x,z)\left(A_{\xi\Gamma}(x,z)-A_{-\xi\Gamma}(x,z)\right)=0
$$
which is true from Lemma \ref{lemma1} and the fact that $\pi$ is non-zero almost everywhere. Second, for $x\neq y$ and $\xi=-\eta$, Eq. \eqref{eq:proof3} is trivially true by definition of $K_\rho$, see \eqref{NRMHAV_kernel}. Hence, condition(s) on the vorticity matrix to ensure $\tpi$-reversibility are to be investigated only for the case $\xi=\eta$ and $x\neq y$. In such a case Eq. \eqref{eq:proof3} is equivalent to
$$
\pi(x)Q(x,y)A_{\xi\Gamma}(x,y)=\pi(y)Q(y,x)A_{-\xi\Gamma^\xi}(y,x)\,,
$$
which is equivalent $\Gamma=\mathbf{0}$. Hence $K_\rho$ is $\tpi$-reversible if and only if $\Gamma=\mathbf{0}$.
\end{proof}

\begin{lemma}
\label{lemma1}
Under the Assumptions of Proposition \ref{prop:sec5}, we have for all $x\in\Scal$ and $\xi\in\{-1,1\}$
$$
\pi(x)\sum_{z\in\Scal} Q(x,z) \left\{A_{\xi\Gamma}(x,z)-A_{-\xi\Gamma}(x,z)\right\}=0\,.
$$
\end{lemma}
\begin{proof}
Using that for three real numbers $a,b,c$, we have $a\wedge b=(a-c\wedge b-c) +c$, together with the fact that $\Gamma(x,y)=-\Gamma(y,x)$, we have:
\begin{eqnarray}
\label{eq:lem1}
\pi(x)Q(x,y)A_{\xi\Gamma}(x,y)\hspace{-0.5cm} && = \pi(x)Q(x,y) \left\{ 1 \wedge \frac{\xi\Gamma(x,y) + \pi(y)Q(y,x)}{\pi(x)Q(x,y)} \right\}\,,\nonumber \\
&& = \pi(y)Q(y,x) \left\{ 1 \wedge \frac{\xi\Gamma(y,x) + \pi(x)Q(x,y)}{\pi(y)Q(y,x)} \right\} + \xi\Gamma(x,y) \,, \nonumber\\
&& = \pi(y)Q(y,x)A_{\xi\Gamma}(y,x) + \xi\Gamma(x,y)\,.
\end{eqnarray}
The proof follows from combining the skew-detailed balance equation \eqref{eq:sdbe} and Eq. \eqref{eq:lem1}:
\begin{eqnarray*}
\pi(x)\hspace{-0.5cm}&&\sum_{z\in\Scal} Q(x,z)\{A_{\xi\Gamma}(x,z)-A_{-\xi\Gamma}(x,z)\} \\
&&=\sum_{z\in\Scal} \left\{\pi(x)Q(x,z)A_{\xi\Gamma}(x,z) - \pi(x)Q(x,z)A_{-\xi\Gamma}(x,z)\right\}\,, \\
&& = \sum_{z\in\Scal} \left\{\pi(x)Q(x,z)A_{\xi\Gamma}(x,z) - \pi(z)Q(z,x)A_{\xi\Gamma}(z,x) \right\}\,,\\
&& = \sum_{z\in\Scal} \xi\Gamma(x,z)\,,\\
&&=0\,.
\end{eqnarray*}

\end{proof}

\section{Illustration of NRMHAV on Example \ref{ex2}}
\label{sec:ex2NRMHAV}

\begin{figure}[H]
\centering

\hspace*{-.7cm}\includegraphics[scale=0.55]{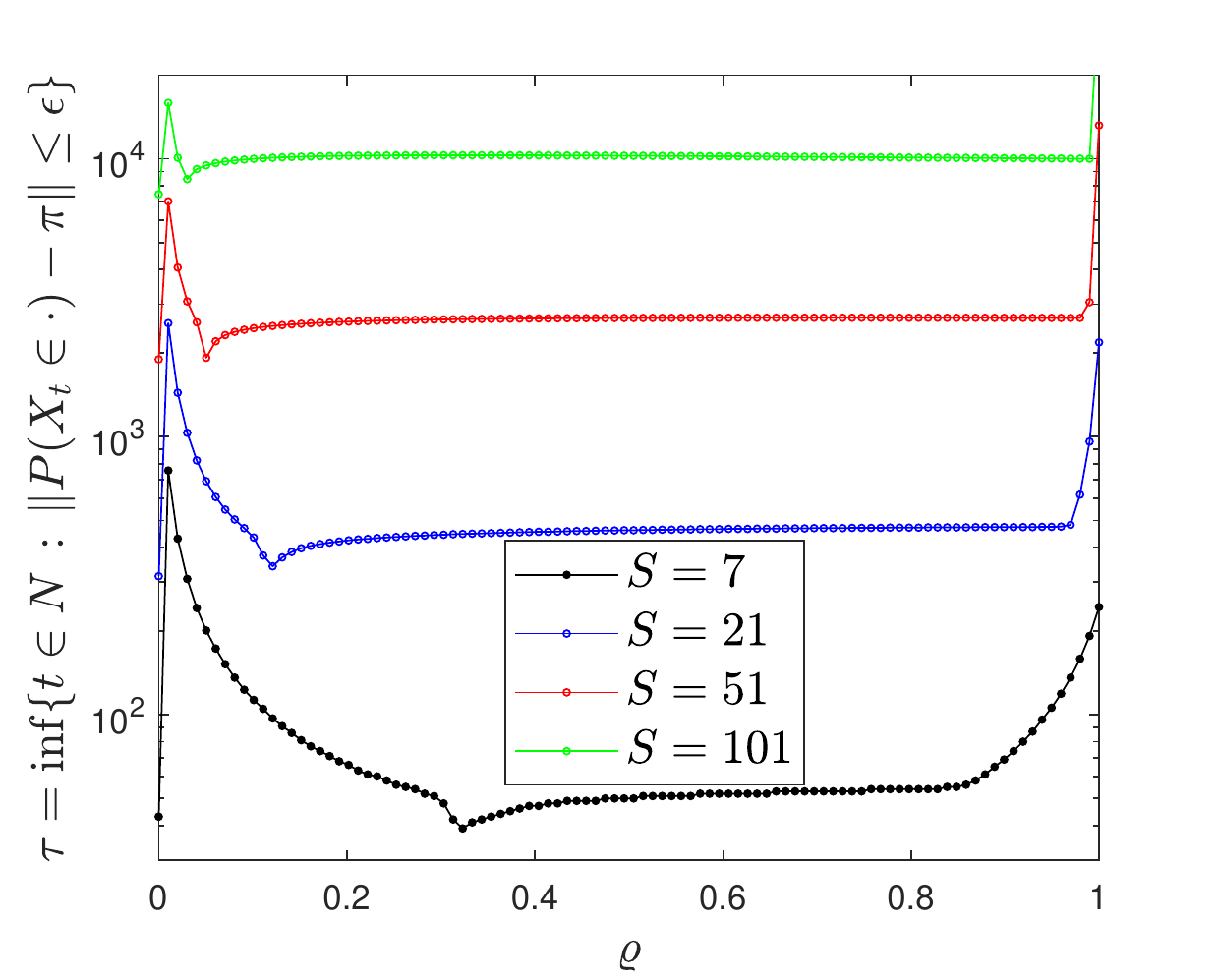}\includegraphics[scale=0.55]{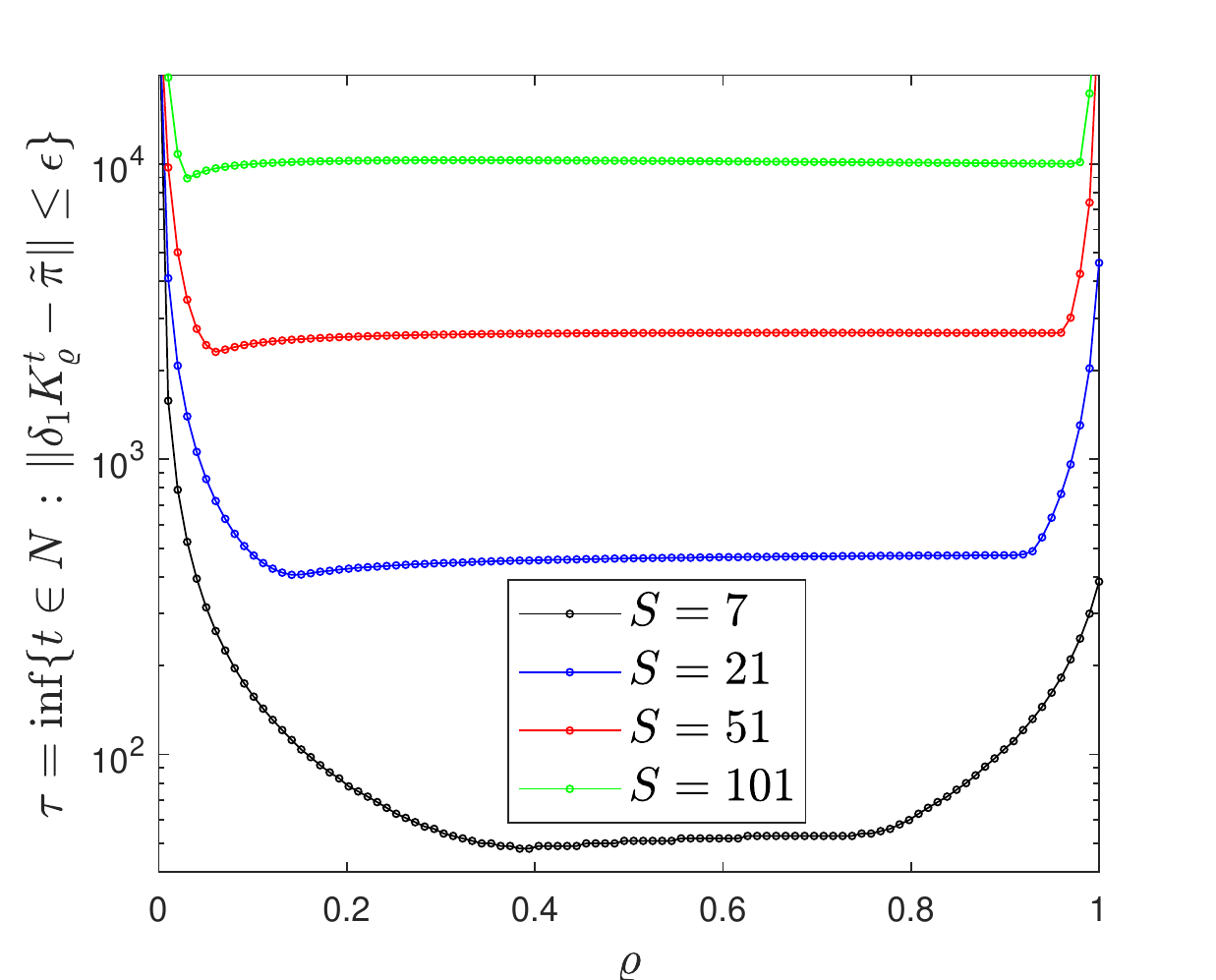}

\includegraphics[scale=0.55]{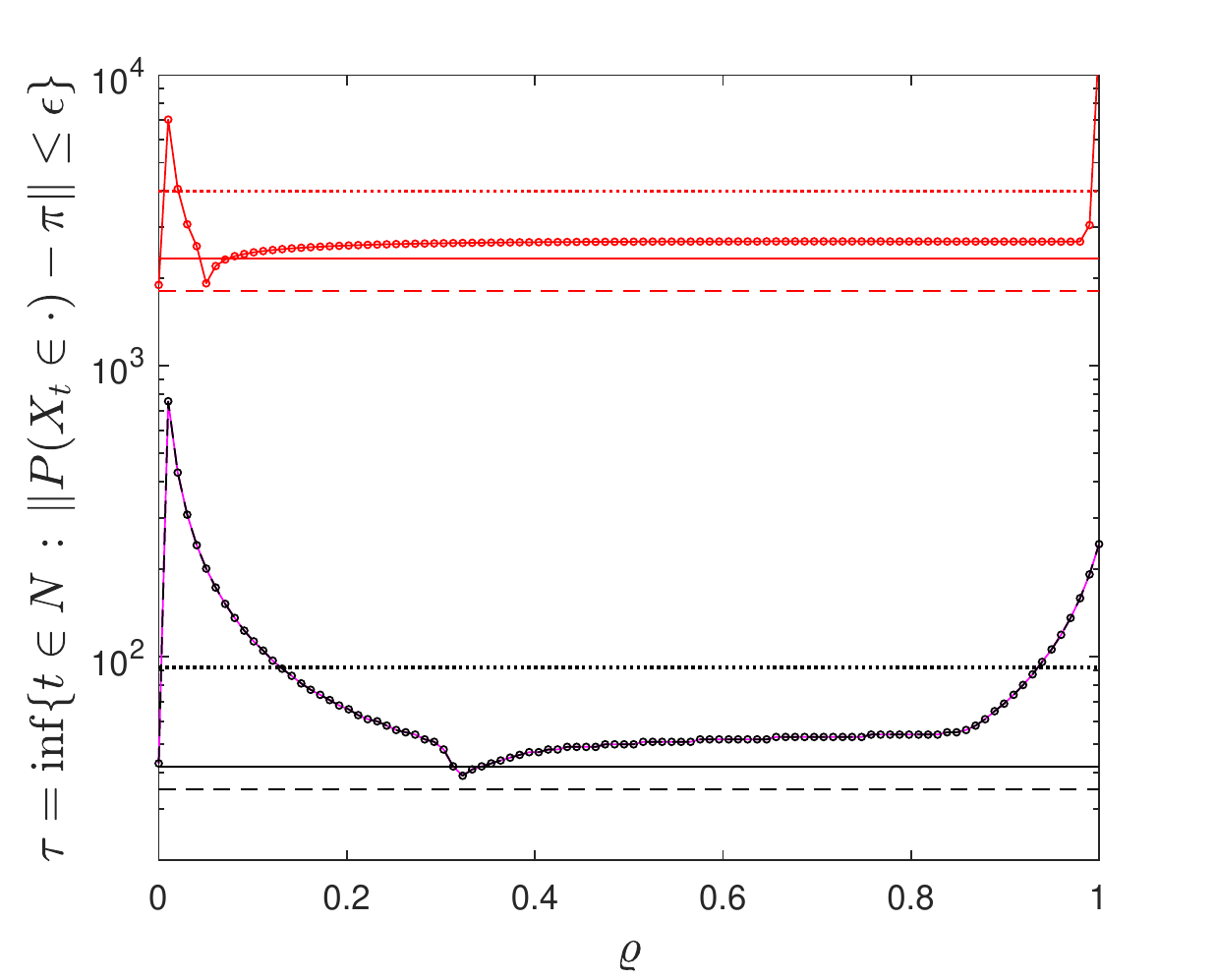}

\caption{({Example \ref{ex2}}) Mixing time of NRMHAV (Alg. \ref{algo_NRMHAV}) in function of $\varrho\in[0,1]$ and for $S\in\{7,21,51,101\}$. Top: convergence of the lifted Markov chain $\{(X_t,\zeta_t),\,t\in\nset\}$ to $\tpi$ (left) and convergence of the marginal sequence $\{X_t,\,t\in\nset\}$ to $\pi$ (left). Bottom: comparison of the convergence of $\{X_t,\,t\in\nset\}$ for MH (plain line), NRMH with $\Gamma$ (dashed), NRMH with $-\Gamma$ (dotted) and NRMHAV (dashed with points), for $S=7$ (black) and $S=51$ (green). \label{fig:ex2:nrmhav}}
\end{figure}

%
%
%

\section{Generation of vorticity matrices on $S\times S$ grids}
\label{app:vorticity_grid}

We detail a method to generate vorticity matrices satisfying Assumption \ref{assumption1} in the context of Example \ref{ex4}. In the general case of a random walk on an $S \times S$ grid, $\Gamma_\zeta$ is an $S^2 \times S^2$ matrix that can be constructed systematically using the properties that $\Gamma_\zeta(x,y) = -\Gamma_\zeta(y,x)$ for all $(x,y)\in\Scal^2$ and $\Gamma_\zeta \mathbf{1} = \mathbf{0}$. It has a block-diagonal structure:

\begin{equation}
\label{vorticity_grille}
\Gamma_\zeta = \begin{pmatrix}
B & 0 & 0 & \cdots & 0 \\
0 & B & 0 & \cdots & 0 \\
0 & 0 & B & & 0 \\
\vdots & \vdots & & \ddots & \vdots
\end{pmatrix}
\end{equation} where each $2S \times 2S$ diagonal block $B$ has the following structure: \begin{equation}
B = \begin{pmatrix}
B_D & B_{OD} \\
-B_{OD} & -B_D
\end{pmatrix}
\end{equation} where \begin{equation*}\footnotesize
B_D = \begin{pmatrix}
0 & -\zeta & 0 & 0 & \cdots & 0 & 0\\
\zeta & 0 & -\zeta & 0 & \cdots & 0 & 0 \\
0 & \zeta & 0 & -\zeta & 0 & \cdots & 0 \\
\vdots & \ddots & \ddots & \ddots & \ddots & \ddots & \vdots \\
0 & \cdots & 0 & \zeta & 0 & -\zeta & 0 \\
0 & 0 & \cdots & 0 & \zeta & 0 & -\zeta \\
0 & 0 & \cdots & 0 & 0 & \zeta & 0
\end{pmatrix}
\end{equation*} and \begin{equation*}\footnotesize
B_{OD} = \begin{pmatrix}
\zeta & 0 & & & \cdots & & 0 \\
0 & 0 &&& \cdots && 0 \\
\vdots &&& \ddots &&& \vdots \\
0 &&& \cdots && 0 & 0 \\
0 & & & \cdots & & 0 & -\zeta \\
\end{pmatrix}
\end{equation*}
and $\zeta$ is  such that the MH ratio \eqref{eq:NRMH_ratio} is always non-negative. The vorticity matrix is of size $S^2 \times S^2$, meaning that the number of diagonal blocks varies upon $S$: \begin{itemize}
\item[$\bullet$] \textbf{if $S$ is even:} $\exists k \in \mathbb{N} \mbox{ s.t. } s = 2k \mbox{ } \Rightarrow \mbox{ } s^2 = 4k^2$ and each block $B$ is a square matrix of dimension $4k$, then there are exactly $k$ $B$-blocks in the vorticity matrix $\Gamma_\zeta$ ;
\item[$\bullet$] \textbf{if $S$ is odd:} $\exists k \in \mathbb{N} \mbox{ s.t. } s = 2k+1 \mbox{ } \Rightarrow \mbox{ } s^2 = (2k+1)^2$ and each block $B$ is a square matrix of dimension $2(2k+1)$, then as $\frac{(2k+1)^2}{2(2k+1)} = k + \frac{1}{2}$, $\Gamma_\zeta$ is made of $k$ $B$-blocks and the last terms of the diagonal are completed with zeros.
\end{itemize}
For instance, if $S = 3$ (resp. if $S = 4$), the vorticity matrix is given by $\Gamma_\zeta^{(3)}$ (resp. $\Gamma_\zeta^{(4)}$) as follows: \begin{equation*}
\Gamma_\zeta^{(3)} = {\scriptsize \left( \begin{array}{cccccc|ccc}
0 & -\zeta & 0 & \zeta & 0 & 0 & 0 & 0 & 0 \\
\zeta & 0 & -\zeta & 0 & 0 & 0 & 0 & 0 & 0 \\
0 & \zeta & 0 & 0 & 0 & -\zeta & 0 & 0 & 0 \\
-\zeta & 0 & 0 & 0 & \zeta & 0 & 0 & 0 & 0 \\
0 & 0 & 0 & -\zeta & 0 & \zeta & 0 & 0 & 0 \\
0 & 0 & \zeta & 0 & -\zeta & 0 & 0 & 0 & 0 \\ \hline
0 & 0 & 0 & 0 & 0 & 0 & 0 & 0 & 0 \\
0 & 0 & 0 & 0 & 0 & 0 & 0 & 0 & 0 \\
0 & 0 & 0 & 0 & 0 & 0 & 0 & 0 & 0
\end{array} \right) },
\end{equation*} \begin{equation*}
\Gamma_\zeta^{(4)} = \begin{pmatrix}
B_4 & \mathbf{0}_8 \\
\mathbf{0}_8 & B_4
\end{pmatrix}
\end{equation*} where \begin{equation*}\scriptsize
B_4 = \begin{pmatrix}
0 & -\zeta & 0 & 0 & \zeta & 0 & 0 & 0 \\
\zeta & 0 & -\zeta & 0 & 0 & 0 & 0 & 0 \\
0 & \zeta & 0 & -\zeta & 0 & 0 & 0 & 0 \\
0 & 0 & \zeta & 0 & 0 & 0 & 0 & -\zeta \\
-\zeta & 0 & 0 & 0 & 0 & \zeta & 0 & 0 \\
0 & 0 & 0 & 0 & -\zeta & 0 & \zeta & 0 \\
0 & 0 & 0 & 0 & 0 & -\zeta & 0 & \zeta \\
0 & 0 & 0 & \zeta & 0 & 0 & -\zeta & 0
\end{pmatrix}
\end{equation*} and $\mathbf{0}_m$ stands for the zero-matrix of size $m \times m$.

\begin{figure}
\centering
\includegraphics[scale=1]{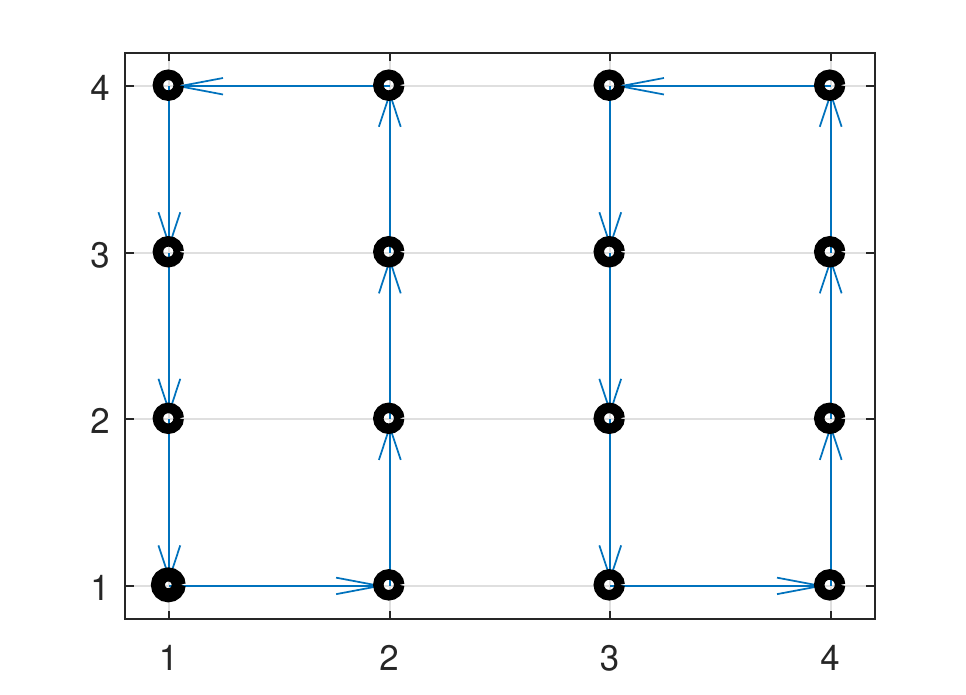}
\caption{\label{fig:flow_s4}
 Illustration of the generic vorticity matrix specified by the previous Algorithm in the case $S=4$.}
\end{figure}

\end{document}